\documentclass[11pt]{article}
\usepackage{graphicx}
\usepackage{enumitem}
\usepackage{amsmath, amssymb, amsthm}
\usepackage{float}
\usepackage[hidelinks]{hyperref}
\usepackage{lineno}
\usepackage{enumitem}
\usepackage{tikzgrid} 
\usepackage{subcaption}

\usepackage{subcaption}
\usepackage{tikz}
\usepackage{xcolor}
\usetikzlibrary{backgrounds} %
\pgfdeclarelayer{background}    %
\pgfsetlayers{background,main}  %

\usepackage[margin = 1in]{geometry}

\newcommand{\mP}{\mathcal{P}}
\newcommand{\R}{\mathbb{R}}
\newcommand{\rr}{\mathbb{R}}
\newcommand{\N}{\mathbb{N}}
\newcommand{\Z}{\mathbb{Z}}
\newcommand{\C}{\mathcal{C}}

\newcommand{\viable}{\mathcal{V}}

\newcommand{\branch}{\mathcal{B}}
\newcommand{\Alg}{\textsf{SelectMarkedTree}}

\theoremstyle{plain}
\newtheorem{theorem}{Theorem}
\newtheorem{lemma}[theorem]{Lemma}
\newtheorem{claim}[theorem]{Claim}

\newtheorem{proposition}[theorem]{Proposition}
\newtheorem{corollary}[theorem]{Corollary}

\theoremstyle{definition}
\newtheorem{definition}[theorem]{Definition}

\numberwithin{theorem}{section}

\newcommand{\eps}{\varepsilon}

\newcommand{\abs}[1]{\left|{#1}\right|}

\newcommand{\suchthat}{\ | \ }

\newcommand{\txt}[1]{\text{#1}}

\newcommand{\stext}[1]{\ \ \ \ \ \text{(#1)}}

\newcommand{\ipncm}[3]{\begin{figure}[H]\begin{center}\includegraphics[scale = {#1}]{Media/#2.pdf}\caption{#3}\end{center}\end{figure}}

\usepackage{todonotes}
\newcommand{\jcmNote}[1]{}
\newcommand{\gjhNote}[1]{}
\newcommand{\gsNote}[1]{}
\newcommand{\scNote}[1]{}
\newcommand{\jtfNote}[1]{}
\newcommand{\krtNote}[1]{}

\date{\today}
\title{The Balanced Up-Down Walk}
\author{Hugo A. Akitaya\thanks{University of Massachusetts Lowell; \url{hugo_akitaya@uml.edu}.  Supported in part by NSF award CCF-2348067.}
\and Sarah Cannon\thanks{Claremont McKenna College; \url{scannon@cmc.edu}.  Supported in part by NSF award CCF-2443221.}
\and Gregory Herschlag\thanks{Duke University; \url{gregory.herschlag@duke.edu}}
\and Gabe Schoenbach\thanks{University of Chicago; \url{gschoenbach@uchicago.edu}}
\and Kristopher Tapp\thanks{Saint Joseph's University \url{ktapp@sju.edu}}
\and Jamie Tucker-Foltz\thanks{Yale School of Management; \url{j.tuckerfoltz@yale.edu}}}

\begin{document}

\maketitle
\thispagestyle{empty}
\begin{abstract}

Markov chains based on spanning trees have been hugely influential in algorithms for assessing fairness in political redistricting. The input graph represents the geographic building blocks of a political jurisdiction. The goal is to output a large ensemble of random graph partitions, often to understand a probability measure on partitions, which is done by drawing and splitting random spanning trees; the subtrees then induce the parts of the partition. Crucially, these subtrees must be \emph{balanced}, since political districts are required to have equal population. The Up-Down walk (on trees or forests) repeatedly adds a random edge then deletes a random edge to produce a new tree or forest; it can be used to efficiently generate a large ensemble of trees or forests, but the rejection rate to maintain balance grows exponentially with the number of parts. ReCom and its variants, the most widely-used class of Markov chains, circumvent this complexity barrier by merging and splitting pairs of districts at a time. These run fast in practice but can have trouble exploring the state space of all possible partitions.

To overcome these efficiency and mixing barriers, we propose a new Markov chain called the Balanced Up-Down (BUD) walk. The main idea is to run the Up-Down walk on the space of trees, but require all steps to preserve the property that the tree is splittable into balanced subtrees. This chain generalizes the moves of two recently proposed Markov chains, the Cycle Walk and the Marked Edge Walk. Furthermore, the BUD walk samples from a known invariant measure under exact balance. We prove that the BUD walk is irreducible (meaning every possible tree is reachable) in several special cases. %
Notably, this %
includes a regime where ReCom suffers from locked configurations, from which there are no valid moves.

Running the BUD walk efficiently presents multiple algorithmic challenges, especially when parts are allowed to deviate from their ideal size within a certain tolerance, as is the case in practice. A key subroutine is determining whether a tree is splittable into subtrees that all fall within the given tolerance. We give an improved analysis of an existing algorithm for this problem, showing it is possible to determine this faster than was previously known. We also show this is an inherently challenging problem by proving the associated counting problem is \#P-complete.
We empirically validate the usefulness of the BUD walk by comparing its performance to that of other existing methods for sampling partitions, leveraging additional structural results and simplifications to make it computationally feasible.

\end{abstract}

\newpage
\tableofcontents
\newpage
\setcounter{page}{1}
\pagestyle{plain}
\section{Introduction}\label{sec:intro}

In the United States (and around the world), states and regions are divided into districts for the purpose of electing representatives. How these districts are drawn can have a profound effect of who ends up being elected. Computationally, such a region can be represented as a graph with vertices for small geographic units (e.g. census blocks or voting precincts, weighted by population), and edges that represent geographic adjacency. A political districting plan with $k$ districts is then a partition of this graph into $k$ connected subgraphs, each with approximately equal total weight, which we call a {\it balanced} partition. Sampling districting plans has become a widely-used tool for creating a baseline of possibilities, evaluating enacted or proposed districting plans, and detecting and quantifying gerrymandering, especially in the United States~\cite{chikina2019separating, Chikina_Frieze_Pegden_2017, deford2019redistricting, fifield2020automated, herschlag2020quantifying, NCSenate2021, zhao2022mathematically}, including in courts~\cite{chen2022brief, duchin2017expert, duchin2019brief, hirsch2022brief, jcmReport, jcmReportHarperVHallMoore}. 

Most methods for sampling balanced partitions randomly draw spanning trees and split them into forests, with each tree in the forest inducing one part of the partition.  
From a combinatorial perspective, sampling random forests has long served as a canonical example in the study of matroid basis exchange and related Markov chains (e.g. \cite{anari2021logconcavepolynomialsivapproximate}). For a graph with $n$ vertices and any $\ell<n$, the set of all forest subgraphs with $\ell$ total edges forms the bases of a matroid with rank $\ell$, and set of all spanning trees ($\ell = n-1$) forms the bases of the {\it graphic matroid}.
This builds on a rich classical literature about sampling/counting trees and forests in graphs that dates back to Kirchhoff’s matrix–tree theorem \cite{Kirchhoff1847} and includes both exact and Markov chain–based algorithms \cite{tam2025exact,russo2018linking,tutte2001graph,wilsonGeneratingRandomSpanning1996}. However, the natural problem of counting/sampling forests whose component trees are similar is size has remained woefully underexplored. 

Beyond the redistricting application that motivates this work, sampling graph partitions with additional constraints such as balance also arises in certain probabilistic graph clustering methods. In these settings, Markov Chain Monte Carlo methods are used to perform Bayesian inference to produce a graph clustering, that is, a partition of the graph.  These methods have been shown to be successful in application areas such as ocean temperature-salinity relationships~\cite{Luo2021Bayesian}, cancer mortality and income~\cite{Bhattacharyya2025Constrained}, single cell RNA sequences~\cite{Rebaudo2025GARP}, and anomaly detection on road networks~\cite{Lee2021T-LoHo}.

\subsection{Existing sampling methods}

Early sampling approaches developed for the redistricting problem relied on Glauber dynamics (also known as the Flip Walk; see \cite{chikina2019separating,Chikina_Frieze_Pegden_2017,fifield2020automated,herschlag2020quantifying}), where one random unit is moved to a new district in each step. This can exhibit poor mixing behavior, particularly on large or finely discretized graphs, where it takes a long time to produce a partition that is significantly different from the initial one.

A more global Markov chain, known as ReCom (or Recombination), instead merges two entire districts, draws a spanning tree on their union, and then cuts this spanning tree into two pieces to form two new districts. ReCom has demonstrated strong empirical mixing properties \cite{deford2019recombination}, but little is rigorously known about its mixing time or invariant measure. Subsequent work further generalized the state space from partitions to {\it spanning forests} (forests including all vertices of the underlying graph) and allowed sampling from a specified class of probability measures by coupling ReCom with tools like Metropolis-Hastings and Sequential Monte Carlo \cite{autry2023metropolized,autry2021metropolized,cannon2022spanning,mccartan2020sequential}. The practical success of ReCom and its variants has led to widespread use in applications, including the analysis and auditing of gerrymandering in both the academic literature and court cases (e.g.\ \cite{RuchoVCC,deford2019redistricting,duchinPAreport,jcmReportHarperVHallMoore,zhao2022mathematically}). Despite this empirical success, however, key theoretical questions remain unresolved. 

Beyond ReCom, other methods rely on enforcing the equal-size requirement via rejection sampling after producing a (not necessarily balanced) sample~\cite{AnariVinzantVuong2021, cannon2024sampling, charikar2023complexitysamplingredistrictingplans}. 
This includes the Up-Down Markov chain on spanning forests containing $k$ trees, which iteratively adds a random edge and then removes a random edge to produce a new spanning forest~\cite{charikar2023complexitysamplingredistrictingplans}. However, the forests produced are not necessarily balanced, and an additional rejection step is needed to produce balanced samples. As an alternative approach, this same Up-Down chain can be used on a different state space to sample random spanning trees, iteratively adding a random edge to a spanning tree and then removing a random edge from the resulting cycle; this is the basis exchange walk on the dual of the graphic matroid. However, again an additional rejection step --- keeping only spanning trees that are splittable into $k$ equal parts --- is needed to produce a balanced partition. 
Because of this, investigation has focused on the likelihood that a random spanning tree admits a valid balanced cut, which turns out to be closely related to the likelihood a random $k$-tree forest subgraph is balanced~\cite{cannon2024sampling,gallagher_tapp_2026,gillman2025splittable}.
Though this approach has produced a provably polynomial-time sampling algorithm on grids and grid-like graphs (the natural simplified settings to consider)~\cite{cannon2024sampling,charikar2023complexitysamplingredistrictingplans}, this only holds when the number of parts $k$ is constant; the expected running time grows exponentially in $k$. Moreover, even for relatively small $k$, the polynomials involved are too large (and therefore the rejection rates are too high) to make this sampling method practicable.

\subsection{Theoretical guarantees (or lack thereof) for existing sampling methods}

Despite their wide empirical success, there remain many open questions about the properties of Recombination and related Markov chains. For example, it is not known whether recombination moves that merge two districts and split them in a new way suffice to reach any possible districting plan, a property known as {\it irreducibility}. This has motivated a growing body of work demonstrating irreducibility of recombination-based chains in special cases or restricted settings~\cite{akitaya2023reconfiguration, akitaya2022recom, charikar2023complexitysamplingredistrictingplans}.
The broadest known irreducibility results are for unweighted graphs and $k = 3$ districts when partitions can differ in size by $\pm 1$ vertices on triangular regions of the triangular lattice~\cite{cannon2024irreducibility} and grid subgraphs~\cite{akitaya2026redistricting}. In the case of exact balance (rather than $\pm1$), there exist several examples of {\it locked configurations}, partitions of the grid where no recombination moves can change the partition~\cite{LockedPolyominoTilings}. No positive irreducibility results are known for broad graph classes under exact balance. Thus far, ReCom and related chains have also resisted any mixing time analysis.  Though empirically it seems these chains mix fairly quickly, there is no theory to back this up.

One the other hand, a lot is rigorously known about the Up-Down walk, on forests and on trees. It can be viewed as a basis-exchange walk on a matroid, so is known to be irreducible on the space of all trees/forests, whether or not they are splittable/balanced. This chain is also known to mix in $O(n \log n)$ steps, and each step can be implemented in amortized $O(\log n)$ time, giving an $O(n \log^2 n)$ running time for sampling random $k$-component forests or random spanning trees~\cite{AnariVinzantVuong2021}. These impressive theoretical guarantees are largely made possible by dropping the balance restriction and instead looking at a much larger state space. While results have shown that balanced forests (or splittable trees) form at least a polynomial fraction of the state space for grids and grid-like graphs, this polynomial has $k$, the number of districts, in the exponent. Because of this, balanced forests and splittable trees form too small of a fraction of the state space to make the Up-Down walk, coupled with rejection sampling, feasible in practice. 

A main goal of work in this area is to find a middle ground between methods such as Recombination that perform well in practice but lack theoretical guarantees, and methods such as the Up-Down Walk that come with some theoretical guarantees but are not practically useful.  This is what we aim to do with the BUD walk.

\subsection{Samping from a broader class of measures}

At a methodological level, much of the redistricting literature has focused on Markov chains that replace entire trees or forests at each step~\cite{ autry2023metropolized,autry2021metropolized,cannon2022spanning,deford2019recombination}, while combinatorial approaches have more often emphasized local moves, such as adding and then deleting edges \cite{anari2021logconcavepolynomialsivapproximate,russo2018linking}. Both paradigms are effective for sampling from uniform distributions on forests, or equivalently, the {\it spanning tree distribution} on partitions, where the weight of each partition is the product, across its parts, of the number of spanning trees of each part. While this is a natural distribution that has many nice properties~\cite{cannon2022spanning}, it is often desirable to sample from broader classes of probability measures, and this has proven challenging when chains' proposed moves rely heavily on tree structures.  
Several recent methods have begun to address this limitation, including approaches based on balanced spanning forests (Cycle Walk \cite{cyclewalk}), marked spanning trees (Marked Edge Walk, or MEW \cite{mcwhorter2025MEW}), and multiscale constructions \cite{chuang2024multiscaleparalleltemperingfast}. Both Cycle Walk and MEW operate by drawing and dissolving cycles in a forest or marked tree, enabling structural changes that can propagate across small, medium, and large scales. This desirable feature helps enable practical sampling of distributions that are less closely tied to tree weights. 

\subsection{The BUD walk}

Proposed but not studied in \cite{cyclewalk} and motivated by the Up-Down walk on spanning forests, %
we ask the following question: If the state space is restricted from all spanning trees to only those that can be cut into a $k$-component balanced forest, what properties does the resulting Markov chain exhibit? We refer to this process as the \emph{Balanced Up-Down} (BUD) walk. 
Formally, the state space of the BUD walk is all spanning trees that can be cut into $k$ balanced pieces, which we call {\it $k$-splittable trees}. Each move adds a random edge which results in exactly one cycle; and then a random edge of this cycle is removed, conditioned on it being possible to split the resulting tree into $k$ balanced pieces. The BUD walk is formally defined in Section~\ref{sub:tech:BUD}. 

The BUD walk merges desirable features of several Markov chains used for sampling graph partitions into one process. It naturally combines local moves that change only a small part of the tree and global moves that change large potions of the tree, depending on the random edge that is added and the length of the cycle created. 
Unlike ReCom, whose stationary distribution remains unknown (unless the chain is modified so as to make it reversible), the BUD walk's natural stationary distribution can be easily calculated. 
The Up-Down walks on spanning forests and spanning trees remain the only chains used in this area for which mixing time bounds are known, but high rejection rates limit their usefulness. By modifying the Up-Down walk on spanning trees to only walk among those trees which are splittable into balanced pieces, we lose the rigorous mixing time bound but no longer require a final rejection step. While we do not yet have mixing time bounds for the BUD walk, it is hoped that its similarities to the Up-Down walk might lead to such results in the future; in this work, we find some empirical evidence that the mixing times may be related. We are also able to show the BUD walk is irreducible in cases where ReCom isn't, further bolstering its potential as a natural bridge between empirical success and theoretical guarantees.

\subsection{Our contributions}

Our main contribution is the exploration of the properties of the BUD walk and how they compare to previous methods. In Section~\ref{sec:generalizing}, we show that the BUD walk generalizes the moves of both Cycle Walk and the Marked Edge Walk, incorporating all moves available to those chains while admitting additional transitions inaccessible to either.\footnote{This statement holds only for the way in which these two chains have been implemented; both chains could feasibly be generalized to match the moves of the BUD walk.} Section~\ref{sec:BUD-exact-balance} characterizes BUD moves in the case of {\it exact balance}, where each partition must have exactly equal weight. Several structural lemmas regarding how BUD steps (which are defined on trees) act on partitions are presented in this subsection. 

In Section~\ref{sec:irreducibility}, we consider the irreducibility of the BUD walk. 
In Section~\ref{sub:irreducibility-simple}, we show the strong result that the BUD walk is irreducible for all simple lattices (grid subgraphs for which the boundary is a closed loop) when $k = 2$. It is notable that this result holds for irregularly-shaped grid subgraphs, as previous irreducibility results for ReCom (albeit for $k = 3$) required regularly shaped (rectangular or triangular) grid subgraphs~\cite{cannon2024irreducibility, akitaya2026redistricting}.  
In Section~\ref{sub:irreducibility-triominoes}, we show BUD is irreducible in a key case where ReCom is not~\cite{LockedPolyominoTilings}, that of rectangular grid graphs partitioned into pieces of size three (i.e., triominoes). Notably, both these irreducibility results hold under \emph{exact} balance, whereas previous irreducibility results for ReCom require balance deviations of $\pm 1$ vertex. On the other hand, in Section~\ref{sub:negative-results}, we show a contrived grid subgraph on which BUD is not irreducible. 

Section~\ref{sec:complexity} considers the computational complexity of a key subroutine of the BUD walk: determining whether a particular spanning tree is splittable into $k$ balanced parts. While this is simple to answer in unweighted graphs when considering exact balance, it becomes much more difficult when dealing with weighted graphs and approximate balance, the most relevant case for our motivating redistricting application. Determining whether a given weighted tree is splittable into approximately balanced parts had previously been shown to be solvable in time $O(k^4 n)$~\cite{Ito2008Partitioning}. In Section~\ref{subRuntimeAnalysis}, we give an improved amortized runtime analysis of this algorithm, showing this problem is solvable in time $O(k^3 n)$ in general, and time $O(kn)$ in a special case that's particularly relevant to our motivating application.

In the approximate balance case, there may be many ways to split a given tree to produce an approximately balanced partition. If one wants to implement the BUD walk to sample from a probability distribution that depends on the induced partition, rather than simply the underlying tree, it is necessary to have a method for randomly sampling from the possible partitions that could arise by splitting a given tree. In Section~\ref{sub:hardness}, we show it is \#P-hard to compute the necessary conditional probabilities for the natural iterative approach to uniform sampling. 

Despite this hardness, in Section~\ref{sec:practical} we demonstrate that computationally-feasible probabilistic methods can be employed with Metropolis-Hastings to use BUD to sample in practice. 
In Section~\ref{sub:proposing_marked_edges}, we present a method for sampling (non-uniformly) a valid way of splitting a tree to produce an approximately balanced partition.  This method has two nice features: Every valid partition has a nonzero probability of occurring, and the probability with which a given partition arises can be easily computed. 
In Section~\ref{sub:restricting_bud_for_computational_efficiency}, we present a slight modification of BUD that mildly restricts certain transitions; this restriction is necessary to make implementing a single step of BUD computationally feasible. In particular, rather than needing to resample an entire set of $k-1$ split edges in each step, we only recompute those split edges within distance $d$ of the cycle created in the BUD step for some value of $d$, assuming all split edges that are farther away remain the same. 
In Section~\ref{sub:numerical}, we present empirical validations of BUD and comparisons between BUD and Cycle Walk on grid graphs. We also investigate the performance of BUD in several cases, including a graph representing the state of North Carolina and compare it to previous methods used in this same setting.

\section{Technical overview and proof preview}\label{sec:overview}

In this section, we formally define the BUD walk and then give a high-level overview of the proof techniques used to prove our various results.  We begin with a few definitions that will be necessary for this exposition. 

\subsection{Preliminaries and the BUD walk}\label{sub:tech:BUD}

Let $G$ be a connected graph with $n$ vertices, whose vertex and edge sets will be denoted by $V(G)$ and $E(G)$, respectively. Let vertex weights be given by a function $p: V(G) \rightarrow \R^+$; we call $p(v)$ the {\it population} at vertex $v$.  The total population of $G$ is $p(G) := \sum_{v \in V(G)} p(v)$.  The following definition formalizes what we mean for a tree to be splittable into $k$ approximately-balanced components. 

\begin{definition}\label{def:k-balanced-trees}
    For any tolerance $\eps \geq 0$ and integer $k \geq 2$, a spanning tree $T$ of $G$ is  \emph{$\eps$-balanced $k$-splittable}
    if there exists a set of $k - 1$ edges of $T$ whose removal partitions $G$ into $k$ connected components (called \emph{districts}) such that the population of each district is in the range
    \begin{align}
        \left[\frac{p(G)}{k} - \frac{\eps}{2}, \frac{p(G)}{k} + \frac{\eps}{2}\right].
    \end{align}
    Let $B^\varepsilon_k(G)$ be the set of $\eps$-balanced $k$-splittable trees of $G$.
\end{definition}

\noindent Now that we know what it means for a tree to be splittable into balanced pieces, we define the BUD walk. 

\begin{definition}\label{def:BUD}
    The \emph{Balanced Up-Down (BUD) walk on $B^\varepsilon_k(G)$} is the Markov chain whose transition step from a tree $T \in B^\varepsilon_k(G)$ is defined as follows.
\begin{enumerate}
    \item Add to $T$ a uniform random edge from $E(G) \setminus E(T)$, which creates a unique cycle $C$.
    \item Remove an edge from $C$, uniformly at random from the set of edges whose removal would yield a tree $T' \in B^\varepsilon_k(G)$.
\end{enumerate}
\end{definition}

Recall a Markov chain is {\it irreducible} if it can move from any state to any other state, and it is {\it aperiodic} if the possible times at which it could return to a given state have a greatest common divisor of one. The BUD walk is clearly aperiodic; it can return to the same state in a single step if the newly added edge is subsequently removed. Much of the focus of this paper is on the irreducibility of the BUD walk, in various settings.  A Markov chain that is both irreducible and aperiodic is {\it ergodic}.

Let $P$ be the $|B^\varepsilon_k(G)| \times |B^\varepsilon_k(G)|$ transition matrix of the BUD walk, where $P(T, T')$ is the probability, if at tree $T$, of transitioning to tree $T'$ in one BUD step. A {\it stationary distribution} is any row vector $\pi$ such that $\pi P = \pi$. Any aperiodic Markov chain has at least one stationary distribution, and any ergodic Markov chain has exactly one stationary distribution.  For ergodic chains, this unique stationary distribution $\pi$ is also the limiting distribution, that is, it is the probability distribution over states as the number of steps of the chain goes to $\infty$: $\pi = \lim_{t \rightarrow \infty} x_0 P^t$, where $x_0$ is a vector describing any starting state and $P^t$ is the $t$-step transition matrix. 

The following lemma shows that the uniform distribution on $B^\varepsilon_k(G)$ is a stationary distribution of BUD. 

\begin{proposition}\label{prop:stationary-uniform}
    The uniform distribution on $B^\varepsilon_k(G)$
    is a stationary distribution for the BUD walk. If BUD is irreducible, this stationary distribution is unique. 
\end{proposition}
\begin{proof} For any adjacent pair $T, T'\in B^\varepsilon_k(G)$, we have:
\begin{align*}
    P(T, T') = \frac{1}{|E(G)|-|E(T)|}\cdot\frac{1}{L} = P(T', T),
\end{align*}
where $L$ denotes the number of edges of the unique cycle $C$ of $T\cup T'$ whose removal would yield a balanced splittable tree.  Thus, the detailed balance equation holds and so the uniform distribution is a stationary distribution for BUD. If BUD is irreducible, then it is ergodic, and it has a unique stationary distribution, which as just shown must be the uniform distribution. 
\end{proof}

\noindent If BUD is not irreducible, it does not have a unique stationary distribution.  In this case, the limiting distribution it converges to depends on the starting state. If a random walk starts at a tree $T$, the limiting distribution will necessarily be uniform on the set of all trees in $B^\varepsilon_k(G)$ reachable by BUD moves from $T$.

\subsection{Irreducibility}\label{sub:tech:irred}

A key question is whether BUD moves suffice to transform any splittable tree in $B^\eps_k(G)$ into any other splittable tree. This is essential to consider before we can ask questions about mixing time and the rate of convergence of the BUD walk to the uniform distribution on $B^\eps_k(G)$. We prove that BUD is irreducible in two important cases, both in the tightest $0$-balanced setting, on opposite sides of the spectrum of values for $k$.

As with previous irreducibility results~\cite{akitaya2026redistricting,cannon2024irreducibility}, we focus on unweighted graphs, or equivalently, graphs where $p(v)$ is the same for every vertex $v$. Without loss of generality, we assume $p(v) = 1$ for all $v$, so $p(G) = |V| = n$ and the target size for each district is $n/k$.  While the broadest irreducibility results for ReCom allow districts sizes of $(n/k) \pm 1$, our irreducibility results hold when each district is required to have size exactly $n/k$. 

In Section~\ref{sec:BUD-exact-balance}, we consider how BUD acts on partitions. In the $0$-balanced case, each $T \in B^0_k(G)$ has a unique set of edges whose removal produces a balanced partition (Lemma~\ref{L:unique_cuts}), so each tree can be associated with exactly one partition.  We prove that if two trees $T, T' \in B^0_k(P)$ yield the same partition~$\mP$, there exists a sequence of BUD moves transforming $T$ into $T'$ where each intermediate tree also forms partition $\mP$ when split (Lemma~\ref{lemInternalStep}).  As a result, our argument for irreducibility can focus only on BUD moves that change the partition, and we can always choose the most convenient tree to use for each partition.

\subsubsection{Simple lattices with $k = 2$} 

In Section~\ref{sub:irreducibility-simple}, we show that for $k = 2$, BUD is irreducible on all {\it simple lattices}: grid subgraphs whose boundary (vertices in the subgraph that are adjacent to vertices not in the subgraph) is a simple closed loop $\alpha$. Equivalently, this is all simply connected, 2-connected grid subgraphs, or grid subgraphs with no holes and no cut vertices.

Without loss of generality, we refer to one part of the 2-partition as red and the other as blue. We also show that it suffices to assume that both the red part and the blue part are themselves simply connected, by giving a method by which to achieve this if it is not initially true (Claim~\ref{lem:reduce-to-simple}). Therefore there are both red vertices and blue vertices in $\alpha$. 

The key insight for this proof is to focus on {\it column components}.  For the tree consisting of the red vertices, we choose the tree that has as many vertical edges as possible, plus additional horizontal edges as necessary to connect these, and the same for the blue tree (Figure~\ref{F:column_graph}). A {\it column component} is a connected set of vertices of the same color in the same column, and a column component is a {\it leaf component} if it is only adjacent to one other column component of the same color. A leaf component is {\it sandwiched} if it has vertices of the opposite color both above and below it (i.e., it does not touch the boundary of the region). 

We are able to show how to transform any 2-partition into a canonical 2-partition that has as few total column components as possible.  When there is both a red sandwiched leaf component and a blue sandwiched leaf component, one can use a cycle that links these components and exchange vertices until one of the leaf components is eliminated, reducing the total number of column components (Lemma~\ref{L:interior_reduction}, Figure~\ref{F:Bud_grid_change}). One can continue doing this until there are no sandwiched leaf components (Lemma~\ref{L:no_sandwich}), at which point all leaf components include a vertex of the boundary cycle $\alpha$. The boundary cycle $\alpha$ must consist of a sequence of red vertices followed by a sequence of blue vertices, meaning there are two places in $\alpha$ where a red vertex is adjacent to a blue vertex (Figure~\ref{F:rotate}). By working near these transition points and gradually moving them around the boundary, we are able to transform this configuration as needed to make the number of column components as small as possible.  It is this step, of gradually working around the boundary cycle $\alpha$, that requires our assumption that the boundary of our grid graph is a simple closed loop. 

For any BUD move from $T$ to $T'$, it is possible to move from $T'$ to $T$. Therefore, showing any partition can be transformed into a particular canonical partition suffices to prove that any partition can be transformed into any other partition, thus showing irreducibility. A complete proof is given in Section~\ref{sub:irreducibility-simple}. 

It would be interesting to show that BUD is irreducible for larger values of $k$. A neat approach would be to use the proof of irreducibility for $k = 2$ as a part of a recursive argument for $k > 2$. To do so, we would need our result to hold for any connected grid subgraph, not just those bounded by a simple closed loop.  This is an aim of ongoing and future work.

\subsubsection{Triominoes in rectangular grids}

In Section~\ref{sub:irreducibility-triominoes}, we show BUD is irreducible on rectangular grids when $k = n/3$, that is, when each district has exactly 3 vertices.  In analogy with the tiling literature, we call these size-three districts {\it triominoes}.  If a rectangular grid has at least one valid partition into triominoes, one of its dimensions must be divisible by three; without loss of generality, we assume it is the horizontal dimension. We show that BUD can transform any triomino partition into a canonical partition where all triominoes are horizontal. This suffices to prove irreducibility. 

For any non-canonical partition, we consider the first vertex in a left-to-right, top-to-bottom ordering that is not in a horizontal triomino (labeled $c$ in Figure~\ref{figTriominoProgram}). We can identify all possible partial partitions including the 13 triominoes closest to this location, and computationally check that for each one, there exists a sequence of BUD moves ultimately producing the next horizontal triomino in the canonical partition (13 is the smallest integer for which this claim is true).  Additional care is taken when $c$ is close to a boundary of the rectangular grid subgraph to be sure that only feasible tilings respecting the boundary are considered. 

\subsubsection{Counterexample} 

An example of a grid subgraph on which BUD is not irreducible is given in Figure~\ref{fig:locked}; this graph was specifically created so that no BUD moves would be possible from the given partition. 
The grid subgraph in question has two holes as well as multiple cut vertices. We do not know of an example of a simply connected grid subgraph on which BUD is not irreducible. Following our results above, for $k = 2$ any such example would need to have cut vertices.

\subsection{Approximately balanced partitions of trees}

A key subroutine of BUD requires us to determine whether a proposed tree is $\eps$-balanced $k$-splittable, and if it is, to choose one such split at random. Sections~\ref{sec:complexity} and \ref{sec:bud_as_a_practical_algorithm} address several aspects of this problem in detail. 

\subsubsection{Determining splittability: improved runtime bounds}

For a vertex-weighted tree $T$, a positive integer $k$, and a tolerance $\eps$, the objective of the {\it Tree Approximate Partition Problem (TAPP)} is to determine whether there exists an $\eps$-balanced $k$-partition of $T$. For simplicity, we assume the vertex weights sum to one (this can easily be accomplished by replacing each weight $p(v)$ with $p(v)/p(G)$, for instance); this means the total weight of each district must be in the range $[\frac1k - \frac\eps2, \frac1k + \frac\eps2]$. 

Ito, Uno, Zhou, and Nishizeki~\cite{Ito2008Partitioning} consider a variant of this problem.  The algorithm they develop can be applied to solve TAPP, with a runtime bound of $O(k^4 n)$ for a tree with $n$ vertices. We describe an essentially equivalent algorithm, using our language and notation, in Section~\ref{subApproximateBalanceDp}. At a high level, this algorithm arbitrarily roots the tree $T$, and is essentially a dynamic program traversing $T$ from the leaves to the root.  At each node $v$, the algorithm tracks all the possible `left-over' populations resulting from the ways the subtree rooted at $v$ could be partitioned into some number of districts plus one partial district. Denote by $S$ the set of possible weights for this partial district.  Unfortunately, $S$ could be exponentially large, making it infeasible to track in its entirety.  Instead, the algorithm tracks the \emph{$\eps$-closure} of $S$, which is $S$ along with all intervals between consecutive values of $S$ that are at most $\eps$-far from each other. 
By merging nearby values of $S$ into one interval, this enables a polynomially-sized representation of the essential information needed to determine $\eps$-balanced $k$-splittability. 

Our main contribution is an improved amortized runtime analysis of this algorithm. At each vertex $v$, the algorithm performs a sequence of Minkowski sums over unions of intervals (specifically, of the $\eps$-closures for each child of $v$). The bound of $O(k^4 n)$ from~\cite{Ito2008Partitioning} follows by noting there are $n-1$ merges, each merging the set of intervals for a child with those for a parent. Each such merge between a vertex $v$ and a vertex $w$ involves $k$ choices for the number of completed districts in the subtree rooted at $v$ and $k$ choices for the subtree rooted at $w$. For each such option there are at most $k$ intervals to be merged from each side of the Minkowski sum, which takes $O(k^2)$ steps.

However, this analysis ignores the fact that costly merges can only appear near the root of the tree. By amortizing this analysis to account for the fact that merges near leaves are much less costly (the key idea is Lemma~\ref{L:few_components_2}), we improve this analysis to $O(n k^3)$ in general~(Theorem~\ref{T:deadline}). Additionally, much of the difficulty in the algorithm comes from needing to track the number of possible completed districts in a subtree rooted at $v$, and handle each possibility separately.  If the number of possible completed districts in any subtree were unique, rather than having to consider $k$ possible such values for $v$ and $k$ possible such values for $w$ and perform $O(k^2)$ merges, we only need to consider 1 or 2 possible values for each, reducing the runtime by an additional factor of $k^2$.  A simple calculation shows that $\eps < \frac{2}{k(k+1)}$ suffices to ensure that it's not possible to partition the graph into $k-1$ or $k+1$ districts, implying the number of possible completed districts in any subtree is unique. In this way, we obtain a runtime bound of $O(nk)$ when $\eps = O(1/k^2)$ (Theorem~\ref{T:deadline}). In many practical redistricting cases, $\eps$ is small enough for this hypothesis to hold, making this bound of $O(kn)$ the most relevant for the practical applications that motivate our work.

\subsubsection{Hardness of exact sampling}
Once we have determined that a tree is $\eps$-balanced $k$-splittable, a natural next goal would be to sample one of the possible splits uniformly at random.  However, we show in Section~\ref{sub:hardness} that one natural approach for such sampling will not succeed. The approach in question is to sample one edge at a time, each with the correct conditional probability. That is, one would consider an edge $e$ and randomly (with the appropriate probability) either make $e$ a cut edge and remove it from the tree, or determine $e$ will not be a cut edge and contract it. Unfortunately, computing the correct probabilities for removing vs. contracting edge $e$ turns out to be intractable: For any edge $e$ in a tree $T$, it is \#P-complete to determine the probability that $e$ is split in a uniformly random $\eps$-balanced $k$-partition of $T$ (Theorem~\ref{thmApproximateBalanceExactSamplingHardness}). 

The proof of this result uses a reduction from a particular 0/1 Knapsack counting problem, where one seeks to determine how many subsets of items have total weight at most some capacity $C$, a problem that is known to be \#P-complete. One can build a tree where each possible item corresponds to one branch. If $s$ is the number of knapsack solutions, there is a particular edge $e$ of this tree where the probability $e$ is split in a uniformly random $\eps$-balanced $k$-partition of $T$ is exactly $s/(s+1)$.  Therefore, if we could estimate this probability, we could compute $s$, and so the problem of estimating this probability must be \#P-hard. (That the problem is also in \#P is straightforward to show.)

\subsubsection{An alternate sampling method with theoretical guarantees}

This hardness result necessitates an alternate approach to randomly sample an $\eps$-balanced $k$-partition of a given tree $T$.  This is necessary for BUD to be able to sample from any distribution that depends on the induced partition rather than just the tree structure, as many practicably relevant distributions do.  Any such sampling method should satisfy (1) any valid $\eps$-balanced partition of $T$ must be sampled with non-zero probability, and (2) given the sampled partition, we can efficiently compute the probability with which it was sampled. Property (1) is important to support irreducibility: we don't want our sampling method to invalidate any transitions that are otherwise possible.  Property (2) is necessary to compute values necessary for our probability filters (Metropolis filters) when using BUD to target a particular distribution.

In Section~\ref{sub:proposing_marked_edges}, we propose an algorithm $\Alg$ (Definition~\ref{def:select-marked-tree}) to randomly sample an $\eps$-balanced $k$-partition of $T$. We do so by sampling a set of edges that, if cut, induce an $\eps$-balanced forest whose $k$ connected components correspond to the desired partition. At a high level, our idea is to recursively choose a leaf $v$ of $T$ and sample a single edge that, if cut, would break $v$ into its own district. How should we choose $v$? A natural approach that satisfies Property (1) would be to sample a random ordering on the vertices. But this would run afoul of Property (2), as computing the probability of sampling a given partition would require summing over all $n!$ vertex orderings. On the other hand, fixing an ordering on which to visit the leaves of $T$ would easily satisfy Property (2), but may deterministically avoid certain partitions.\footnote{
    For instance, this approach would never sample partitions for which the district containing the first leaf vertex requires \emph{more than one cut} to be split off from $T$.
} Our solution is to use a fixed ordering, but also occasionally contract $v$ into the tree instead of cutting it off. Under a mild bound on the population tolerance --- which ensures that our probabilistic contraction still allows us to satisfy Property (1) --- we show in Lemma~\ref{lem:alg-has-properties} that $\Alg$ satisfies both of our desired properties.

\subsubsection{Restricting BUD for computational efficiency}

As described, BUD allows transitions between pairs of splittable trees, but it does not specify the split. When defining measures on partitions rather than splittable trees, one must extend the state space to ``mark'' edges of a tree such that removing the edges will lead to a balanced forest (and hence partition). 
Naively, one could re-mark \emph{all} edges at every transition, which would require running $\Alg$ on the entire tree, for every step of the BUD walk. This is computationally expensive.

In Section~\ref{sub:restricting_bud_for_computational_efficiency}, we present a slightly restricted version of the BUD walk on partitions. At core, the idea is to fix marked edges that are ``far'' from the cycle involved in the BUD step, so we only need to run $\Alg$ on a restricted subtree, which can lead to significant speedups in practice. The distance at which to fix marked edges is parameterized by a distance $d$; our experiments use $d = 0$, but our approach works for any $d \in \N$.

\subsection{Experimental results}
We present our experimental results in Section~\ref{sub:numerical}. We first validate our code on a $4 \times 4$ grid divided into 4 districts, and observe that we are sampling from the correct distribution. We then consider the $4 \times 4$ grid divided into 5 districts, each of size 3 or 4.  The Cycle Walk fails on this example, but we are able to show that the BUD walk does not (Figure~\ref{fig:cyclevsbud4x4d5}). We then compare the BUD walk and the Cycle Walk on an $8 \times 8$ grid divided into $5$ districts, each required to be size 12 or 13. Despite some challenges determining the fairest comparison between BUD and the Cycle Walk, when looking at effective sample rates and autocorrelations for the distributions each is best suited to sample from, we see BUD outperforms the Cycle Walk over a fixed number of iterations (Figure~\ref{fig:cyclevsbud8x8d5autocor}), though BUD's walk clock time per iteration is significantly larger. We also test BUD's performance on a larger graph by examining the problem of congressional redistricting in the state of North Carolina. We find that it has promising mixing properties on observables of interest.

When adding marked edges to BUD, a certain fraction of the proposed moves will select edges with nodes that belong to the same region, and some that will create loops in the graph in which we quotient out the regions. Previous algorithms have used parameters to tune between selecting edges internal to a region and edges that span it. We find the BUD walk proposes internal and external moves that is at a ratio that is near optimal.

We conclude by comparing the BUD walk with the Up-Down walk: The BUD walk has restricted moves, but the Up-Down walk has a large state space. We find that despite concentrating on different regions of state space, the two walks demonstrate remarkably similar auto-correlations and effective sample rates.

\section{The BUD walk}\label{sec:BUD}

In this section we begin to explore the properties and structure of the BUD walk in more detail, beginning by comparing it to previous chains. 

\subsection{Generalizing other walks}\label{sec:generalizing}

There are two recent sampling methods that are (i) similar to BUD and (ii) have shown promise in generalizing the types of measures that can be efficiently sampled (empirically), that is, in sampling measures that are farther away from the spanning tree measure. In this section we will discuss how BUD can be considered to be a generalization of each of these methods.

\subsubsection{The cycle walk}

The first method we discuss is called the Cycle Walk \cite{cyclewalk}. In it, the state space is considered to be all balanced forests (with no interlinking edges between the trees of the forest). For such a balanced forest, the {\it quotient graph} contracts each tree in the forest to a single vertex, keeping multiple edges and loops. One then probabilistically selects a $\ell$-tree Cycle Walk step which is defined by (i) selecting a length $\ell$ cycle in the quotient graph; (ii) adding the $\ell$ edges of this cycle to the (unquotiented) forest, which creates a unique cycle $C$ involving $\ell$ of the forest's trees; and (iii) finding all $\ell$-tuples of edges of $C$ to remove from $C$ that result in a balanced forest. In the paper, the authors considered 1-tree and 2-tree Cycle Walk steps: A 1-tree Cycle Walk step selects a loop in step (i) and alters one of the trees representing a partition,
and a 2-tree Cycle Walk step draws two edges between two adjacent trees to create a cycle. 

In moving to $\ell>2$, one must establish a protocol for selecting cycles of size $\ell$ and then for removing edges in the resulting cycle to recover a balanced forest. The former step is non-trivial and was not addressed in original exposition; the latter step grows in complexity with the length of the cycle due to the combinatorics of being able to select multiple edges per cut when handling approximate balance conditions.
BUD provides a natural framework for the former step: by expanding the state space to trees splittable into balanced pieces (rather than balanced forests), adding a single edge generates a cycle involving some arbitrary number of (possible) partitions. Furthermore, when traversing the cycle, we only need to consider removing any single edge along the cycle to check if the remaining object is a balanced tree.

Under exact balance conditions, BUD also naturally generalizes to handle the second issue of selecting the $k-1$ remaining edges to move since there is a unique forest associated with any given splittable tree. For approximate balance, we can expand the BUD walk to operate on partitions/forests (rather than trees) by (i) first adding and removing an edge on the splittable tree, and (ii) randomly selecting from the set of $k-1$ edge-tuples which, if removed, would result in a balanced forest with $k$ trees. We can call the chosen $k-1$ edges `marked,' in that we do not remove them from the splittable tree, but selecting them defines a corresponding balanced forest. What is left to do is to then develop an algorithm to propose such a marked edge set given a balanced tree. We develop such an algorithm below in Section~\ref{sub:proposing_marked_edges}.
 The algorithm, called \Alg, is capable of proposing any possible set of marked edges. However, it is also capable of fixing a subset of marked edges and selecting the remainder.  A state space of $k$-splittable trees with $k-1$ marked split edges can be viewed as a `lifted' version of partitions or forests, where we project from splittable trees to balanced forests by removing the split edges to produce $k$ components. A random walk among marked $k$-splittable trees can thus be projected down to a random walk among balanced partitions or forests. 

Therefore, we can think of BUD (combined with marking edges via \Alg) as a generalization of the Cycle Walk in that (i) we can naturally propose cycles that span an arbitrary number of partitions and (ii) new marked edges are not necessarily limited to the cycle that we generate since we first find a $k$-balanced tree and only then find marked edges. In practice, however, it may be convenient to restrict the marked-edge domain when proposing a new marked-edge set as well as fix distant marked edges. We explore possible choices below in Section~\ref{sub:restricting_bud_for_computational_efficiency}; these ideas create a trade-off between algorithmic simplicity and generalized moves. Cycle Walk may be thought of one type of possible restriction in which marked edges on a cycle (formed within a BUD step) must re-appear on cycle, and marked-edges off the cycle are fixed.

\subsubsection{The marked edge walk}
Another recent and promising algorithm is called the Marked-Edge Walk, or MEW~\cite{mcwhorter2025MEW}. The state space for MEW is a tree with marked edges, just as in the method above. The evolution of the tree involves (A) adding a random edge to create a cycle, (B) removing a non-marked edge from the cycle, and then (C) moving a single marked edge (chosen indepdently of the cycle location) so that it still shares a node with the previous marked edge. 

The first step (A) of the MEW is similar to the first step in BUD. However, in step (B), BUD ensures the resulting tree is $k$-splittable whereas MEW does not. Additionally, as described in Section~\ref{sec:bud_as_a_practical_algorithm}, BUD doesn't avoid currently marked edges when picking an edge to remove. Finally, if a MEW step creates a non-balanced tree in this step, it handles this by simply rejecting the new proposed tree and trying again. 
For step (C), the MEW walk can be thought of as employing $\Alg$ in which we (i) restrict all but a single marked edge from changing, (ii) limit the possible edges we can cut (via the adjacency condition), and (iii) ignore the balance condition when marking a new edge and instead address it by rejecting (marked) non-balanced trees, that is, trees where the removal of the marked edges does not produce a balanced partition. 

Conceiving of the rejection conditions as added self-loops in the Markov chain, all of the moves of MEW are possible with a combination of BUD moves and restrictions on applying $\Alg$ to a subset of the marked-edges (in this case a single random choice with a restricted new edge selection). 
Thus BUD provides both a unified and general framework for discussing and analyzing both of these new and empirically promising algorithms, Cycle Walk and Marked Edge Walk.

\subsection{Characterizing BUD moves}\label{sec:BUD-exact-balance}

In this section we consider the case $\varepsilon = 0$, so each district must have population exactly $p(G)/k$. This is the setting in which our irreducibility results hold. As necessary preliminaries for our irreducibility results, the following sequence of results characterizes what the BUD walk looks like from the perspective of partitions rather than splittable trees, and gives conditions on partitions guaranteeing that one partition can be transformed into another by a BUD move or sequence of BUD moves. 
To begin, we note that in this case the set of marked edges (the edges whose removal produces $k$ connected components of the same size) is unique. 

\begin{lemma}\label{L:unique_cuts} If $T\in B^0_k(G)$, then $T$ has a unique set of $k - 1$ edges whose removal partitions $G$ into equipopulous pieces.  We denote this set as $T_0$.
\end{lemma} 
\begin{proof}
Because $T \in B^0_k(G)$, $T$ is $0$-balanced $k$-splittable, meaning there exists at least one set of $k-1$ edges of $T$ whose removal partitions $G$ into $k$ connected components with population exactly $p(G)/k$. To show this set of edges is unique, let $T_0$ be one such set of edges, which we can regard as a spanning tree on the set of components of $T\setminus T_0$.  This spanning tree has at least one leaf (one component $P$ that is connected to only one other component by some edge $e\in T_0$).  The removal from $T$ of any edge within $P$ would result in a component that is underpopulated because it would comprise only a proper subgraph of $P$, so $e$ must be included in any valid set of marked edges.  This argument can be repeated after removing $P$ and $e$ to yield a smaller tree with at least one leaf, and so on until only one district remains.  Therefore all of the edges of $T_0$ must be included in any valid set of balanced cut edges.  
\end{proof}

\noindent There is a natural correspondence between balanced trees and balanced partitions of $G$.

\begin{definition}\label{def:k-balanced-partitions}
    Given a connected graph $G$, we say $\mP$ is a \emph{balanced partition of $G$} if $\mP$ is a partition of $V(G)$ where each part induces a connected subgraph of equal population. For $T \in B^0_k(G)$, the balanced partition given by $T$ is denoted by $\mP(T, k) = P_1, \dots, P_k$.
\end{definition}

Note that for any $T \in B^0_k(G)$, Lemma~\ref{L:unique_cuts} implies $\mP(T, k)$ is unique. For any balanced partition $\mP$ of $G$, let $G/\mP$ denote the quotient graph, which has one node for each district. Formally, $G/\mP$ can be obtained from $G$ by contracting each edge between vertices in the same part in $\mP$, keeping multiple edges. The number of edges in $G/\mP$ between the nodes corresponding to $P_i$ and $P_j$ equals the number of edges of $G$ connecting a vertex of $P_i$ to a vertex of $P_j$. Notice that $T$ can be determined from $\mP$ together with the following additional information: a spanning tree of each district, and a spanning tree of $G/\mP$ (which corresponds to the inter-district edges).  This added information can be freely altered via steps of the BUD walk, as we now show.

\begin{lemma}\label{lemInternalStep}
    For any pair of trees $T, \tilde{T} \in B_k^0(G)$ such that $\mP(T, k) = \mP(\tilde{T}, k)$, it is possible for the BUD walk to transition from $T$ to $\tilde{T}$ in at most $n - 1$ steps, where $n=|V(G)|$.
\end{lemma}

\begin{proof}
    Write $\mP(T, k) = \mP(\tilde{T}, k) = (P_1, P_2, \dots, P_k)$. Any tree in $T \in B_k^0(G)$ decomposes into a set of $k$ spanning trees $T_1, T_2, \dots, T_k$, one on each part in $\mP(T, k)$, together with a spanning tree $T_0$ of $G / P(T, k)$. For $T$ and $\tilde{T}$, we write these decompositions as $(T_0, T_1, \dots, T_k)$ and $(\tilde{T}_0, \tilde{T}_1, \dots, \tilde{T}_k)$, respectively, so that for each $1\leq i\leq k$, $T_i$ and $\tilde{T}_i$ are spanning trees on $P_i$, whereas $T_0$ and $\tilde{T}_0$ are spanning trees of $G/\mP(T,k)$. For each $i \in \{1, 2, \dots, k\}$, we can transition from $T_i$ to $\tilde{T}_i$ in $\abs{P_k} - 1$ steps: at each step, we add any edge in $T_i$ that is not in the current tree, then remove a different edge in the cycle that is not in $T_i$ (which must exist since there are no cycles in $T_i$). None of these moves change the underlying partition, and so all intermediate trees are in $B_k^0(G)$. We then apply the same method to transition from $T_0$ to $\tilde{T}_0$ in a further $k - 1$ steps. This changes which edges are the split edges, but the underlying partition remains the same. The total number of steps is thus at most
    \begin{equation*}
        \sum_{i = 1}^k (\abs{P_i} - 1) + (k - 1) = \sum_{i = 1}^k \abs{P_i} - 1 = n - 1.  \qedhere
    \end{equation*}
\end{proof}

To decide whether the BUD walk is irreducible on $B^0_k(G)$, it therefore suffices to simply ask whether any balanced partition $\mP$ can be converted into any other balanced partition $\tilde{\mP}$, where we can choose whichever trees $T$ and $\tilde{T}$ satisfying $\mP(T,k) = \mP$ and $\mP(\tilde{T},k) = \tilde{\mP}$ are most convenient. Lemma~\ref{lemPartitionCharacterization} characterizes which changes to the underlying partitions are possible in a single BUD step.

\begin{lemma}\label{lemPartitionCharacterization}
    Let $\mP$ and $\tilde{\mP}$ be a pair of distinct $0$-balanced $k$-partitions of $G$. The following statements are equivalent:
    \begin{enumerate}
        \item\label{itmPCharTreesStatement} There exists a pair of trees $T, \tilde{T} \in B_k^0(G)$ such that $P(T, k) = \mP$, $P(\tilde{T}, k) = \tilde{\mP}$, and it is possible to transition from $T$ to $\tilde{T}$ in a single BUD walk step.
        \item\label{itmPCharPartitionsStatement} There exists a cycle in $G/\mP$ (which after re-indexing $\mP = (P_1, \dots, P_k)$ for notational convenience looks like $P_1\rightarrow P_2 \rightarrow \cdots\rightarrow P_\ell\rightarrow P_1$) and subgraphs $H_i\subset P_i$ for all $i\in[\ell]$ such that:
        \begin{enumerate}
            \item\label{itmPCharConnected} For each $i \in [\ell]$, $H_i$ is connected and $P_i \setminus H_i$ is connected.
            \item\label{itmPCharSwitch} $\tilde{\mP}$ is obtained from $\mP$ by transferring $H_i$ from $P_i$ to $P_{i+1}$ for all $i\in[\ell]$ (with indexes interpreted cyclically, so $\ell+1$ means $1$).
        \end{enumerate}
    \end{enumerate}
\end{lemma}

\begin{proof}
(\ref{itmPCharPartitionsStatement}) $\implies$ (\ref{itmPCharTreesStatement}): Let $i\in[\ell]$.  By (\ref{itmPCharConnected}), there exists a spanning tree $T^H_i$ of $H_i$ and a spanning tree $T^K_i$ of $K_i := P_i\setminus H_i$.  Since $\mP_i$ is connected, there exists an edge $e_i$ linking these two trees.  Furthermore, since $\tilde{\mP}_{i+1}$ is connected, there exists an edge $f_i$ linking $T^H_i$ with $T^K_{i+1}$ (with indexes interpreted cyclically, so $\ell+1$ means $1$); see Figure~\ref{F:BUD_step}.  This is because $\tilde{\mP}_{i+1}$ has vertex set $V(H_i)\cup V(K_{i+1})$ according to (\ref{itmPCharSwitch}).

Let $T$ be the spanning tree of $G$ that includes all of the edges of all of the trees $\{T^H_i\mid i\in[\ell]\}$ and $\{T^K_i\mid i\in[\ell]\}$, together with all of the edges $\{e_i\mid i\in[\ell]\}$ and 
$\{f_i\mid i\in[\ell-1]\}$.  Then $P(T,k)=\mP$, and the BUD step that alters $T$ by adding $f_\ell$ and removing $e_\ell$  yields a tree $\tilde{T}$ for which $P(\tilde{T}, k) = \tilde{\mP}$.

(\ref{itmPCharTreesStatement}) $\implies$ (\ref{itmPCharPartitionsStatement}):  Decompose $T$ and $\tilde{T}$ respectively as $(T_0, T_1, \dots, T_k)$ and $(\tilde{T}_0, \tilde{T}_1, \dots, \tilde{T}_k)$, as in the proof of Lemma~\ref{lemInternalStep}.  Let $f$ be the added edge, and let $e$ be the removed edge in the BUD step from $T$ to $\tilde{T}$, and let $C$ denote the corresponding cycle in $G$.  The edge $f$ connects some pair of trees $T_i,T_j$ with $1\leq i,j\leq k$, and we know that $i\neq j$ because of the assumption that $\mP \neq\tilde{\mP}$. Thus, $f$ completes a cycle $C_0$ in $T_0$.  After re-indexing $\mP = (P_1, \dots, P_k)$, we can assume this cycle looks like
$P_1\rightarrow P_2 \rightarrow \cdots\rightarrow P_{\ell}\rightarrow P_1$ and that $f$ connects $P_\ell$ and $P_1$, so we will denote $f$ also as $f_\ell$.  Denote by $f_1,...,f_{l-1}$ the remaining edges of $C_0$, so that $f_i\in T_0$ connects $P_i$ with $P_{i+1}$.  Note that $C$ is a cycle in $G$, while $C_0$ is a cycle in $G/\mP$; both include the edges $f_1,...,f_\ell$ in that order, but $C$ also includes a path in $P_{i+1}$ between $f_i$ and $f_{i+1}$ for each $i\in[\ell]$.

We now claim that each edge in $\tilde{T}_0$ either lies in $T_0\setminus \{f_1,...,f_\ell\}$ or lies in $T_i$ for some $i\in[\ell]$.  For this, we first claim that $T_0$ and $\tilde{T}_0$ agree about the inclusion of any edge of $G$ which neither lies in one of the pieces $P_1,...,P_\ell$ nor connects a pair of these pieces.  This is because the proof of Lemma~\ref{L:unique_cuts} will progress through the same steps when applied to $T$ and $\tilde{T}$, until all districts have been removed except the ones along the cycle $C_0$.  Moreover, no edge of $\tilde{T}_0$ could be in $\{f_1,...,f_\ell\}$ because $\mP\neq\tilde{\mP}$.  In summary, the set $\tilde{T}_0\setminus T_0$ contains exactly $\ell-1$ edges, each of which lies in $P_j$ for some $j\in[\ell]$.  The same is true of the removed edge $e$, so the set $(\tilde{T}_0\setminus T_0)\cup\{e\}$ contains exactly $\ell$ edges, each of which lies in  $P_j$ for some $j\in[\ell]$.  No two such edges could lie in the same part because of population balance, so we can enumerate them as $e_1,...,e_{\ell}$, where $e_j\in T_j$ for each $j\in[\ell]$.

Let $j\in[\ell]$.  For the two components of the partition of $P_j$ obtained by removing $e_j$ from $T_j$, let $H_j$ be the one adjacent to $f_j$ and let $K_j$ be the other one, which is adjacent to $f_{j-1}$.  It is now straightforward to show that (\ref{itmPCharConnected}) and (\ref{itmPCharSwitch}) hold.

\end{proof}

\begin{figure}[bht!]\centering
\includegraphics[width=5.5in]{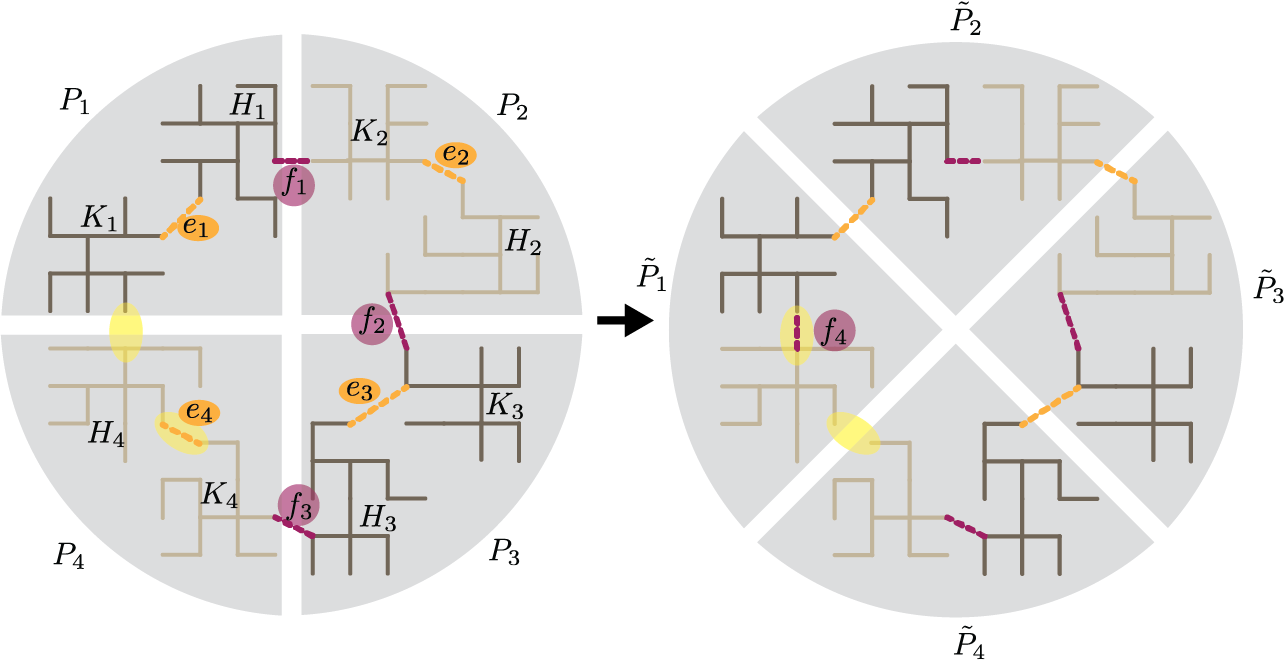}
\caption{The BUD step in the proof of Lemma~\ref{lemPartitionCharacterization} with $k=4$.}\label{F:BUD_step}
\end{figure}

If $\mP$ and $\tilde{\mP}$ satisfy the conditions of Lemma~\ref{lemPartitionCharacterization}, then we will abuse language by saying that they differ by a single \emph{BUD step on partitions}.

Notice that Lemma~\ref{lemPartitionCharacterization} assumes that both $\mP$ and $\tilde{\mP}$ are balanced partitions, which means by definition that they have connected parts.  If instead we only assume that $\mP$ has connected parts, and if there exists a cycle satisfying conditions (2a) and (2b), then the following additional assumption is required to ensure that $\tilde{\mP}$ also has connected parts:
\begin{enumerate}
\item[(2c)] For each $i \in [\ell]$, $H_i$ is adjacent to $P_{i+1} \setminus H_{i+1}$.
\end{enumerate}

Lemma~\ref{lem:counting-characterization} provides a counting-based characterization of when $\mP$ and $\tilde{\mP}$ are connected by a set of disjoint BUD steps on partitions.
\begin{lemma}\label{lem:counting-characterization}
  Let $\mP$ and $\tilde{\mP}$ be a pair of distinct balanced $k$-partitions of $G$.  The following are equivalent:
  \begin{enumerate}
      \item $\mP$ and $\tilde{\mP}$ differ by a set of disjoint BUD steps on partitions.
      \item There exists a spanning tree on each piece of $\mP$ and on each piece of $\tilde{\mP}$ such that the number of pieces in $\mP$ but not in $\tilde{\mP}$ equals the number of edges in the union of the trees of $\mP$ but not in union of the trees of $\tilde{\mP}$.
  \end{enumerate}
\end{lemma}

\begin{proof}

   $(1) \implies (2)$: If $\mP$ and $\tilde{\mP}$ differ by a single BUD step on partitions, then there exists a cycle  $P_1\rightarrow P_2 \rightarrow \cdots\rightarrow P_\ell\rightarrow P_1$ in $G/\mP$ as in Lemma~\ref{lemPartitionCharacterization}, and these $\ell$ pieces are the only pieces on which $\mP$ and $\tilde{\mP}$ differ.  Furthermore, the proof of the lemma constructs trees $T^H_1,...,T^H_l, T^K_1,...,T^K_l$ that can be connected by edges $e_1,...,e_\ell$ to become trees on the pieces of $\mP$, and also can be connected by edges $f_1,...,f_\ell$ to become trees on the pieces of $\tilde{\mP}$.  In particular, $e_1,...,e_l$ are the only edges that are in the union of the trees on $\mP$ but not in the union of the trees on $\tilde{\mP}$.  Thus, the number of pieces in $\mP$ but not in $\tilde{\mP}$ equals the number of edges in the union of the trees of $\mP$ but not in union of the trees of $\tilde{\mP}$.

   If $\mP$ and $\tilde{\mP}$ differ by a set of of disjoint BUD steps on partitions, then there is a set of disjoint cycles as above, and we arrive at the same conclusion.

    $(2) \implies (1)$:  Assume there exist spanning trees on the pieces of $\mP$ and $\tilde{\mP}$ such that the number of pieces in $\mP$ but not in $\tilde{\mP}$ equals the number of edges in $\mP$ but not in $\tilde{\mP}$, and call this number $\ell$.  Here we're abusing language by saying that an edge is in $\mP$ if it is in the spanning tree on one of the pieces of $\mP$.  Note that $\ell$ is also the number of pieces in $\tilde{\mP}$ but not in $\mP$ and the number of edges in $\tilde{\mP}$ but not in $\mP$.

    Let $P_1,...,P_\ell$ be the pieces of $\mP$ that are not in $\tilde{\mP}$, and let $\tilde{P}_1,...,\tilde{P}_\ell$ be the pieces of $\tilde{\mP}$ that are not in $\mP$.  Let $e_1,...,e_\ell$ be the edges of $\mP$ that are not in $\tilde{\mP}$, and let $f_1,...,f_\ell$ be the edges of $\tilde{\mP}$ that are not in $\mP$.  First notice that each piece $P_i$ contains at least one and therefore exactly one of $\{e_1,...,e_\ell\}$, so we can re-index so that $P_i$ contains $e_i$ for each $i\in[\ell]$.  Similarly, each piece $\tilde{P}_i$ contains exactly one of $\{f_1,...,f_\ell\}$.  
    
    For each $i\in[\ell]$, the removal of $e_i$ splits the tree on $P_i$ into two trees, which we denote $T^H_i$ on $H_i$ and $T^K_i$ on $K_i$.  Since $e_i$ is the only edge in $P_i$ that's not in $\tilde{P}$, each of $K_i$ and $H_i$ is contained in a single piece of $\tilde{\mP}$.
    
    After re-indexing and swapping the labels of $K_1$ and $H_1$ if necessary, we can assume that $\tilde{P}_1$ contains $K_1$ and $\tilde{P}_2$ contains $H_1$, and that $f_1$ connects $H_1$ to $K_2$.  Similarly, we can assume that $\tilde{P}_2$ contains $K_2$ and $\tilde{P}_3$ contains $H_2$, and that $f_2$ connects $H_2$ to $K_3$, and so on until an index is repeated.  If all indices have been used, then $\mP$ and $\tilde{\mP}$ differ by a single BUD step on partitions.  If not, then repeating the argument with the remaining pieces shows that they differ by a set of disjoint BUD steps on partitions.
\end{proof}

In Section~\ref{sub:irreducibility-triominoes} we will consider triomino tilings of a rectangular grid $G$, so $|V(G)| = 3k$.  In this case, there is a \emph{unique} spanning tree on each piece of any partition, so we can unambiguously say that an edge is in a partition if it is in the spanning tree of one of the pieces of the partition.  The previous lemma gives the following.
\begin{corollary}
  Let $\mP$ and $\tilde{\mP}$ be a pair of different balanced $k$-partitions of $G$ such that each piece of each partition has a unique spanning tree.  The following are equivalent:
  \begin{enumerate}
      \item $\mP$ and $\tilde{\mP}$ differ by a set of disjoint BUD steps on partitions.
      \item The number of pieces in $\mP$ but not in $\tilde{\mP}$ equals the number of edges in $\mP$ but not in $\tilde{\mP}$.
  \end{enumerate}
\end{corollary}

\section{Irreducibility of BUD}\label{sec:irreducibility}

In this section we focus on the BUD walk over \emph{exactly balanced} trees, so the population tolerance is $\varepsilon = 0$. Despite this strict balance requirement, we show that BUD is irreducible in two natural settings. In Section~\ref{sub:irreducibility-simple}, we consider subgraphs $G$ of the infinite square lattice whose boundary is a simple closed loop, and show that BUD is irreducible over $B^0_2(G)$. In Section~\ref{sub:irreducibility-triominoes}, we consider rectangular subgraphs $G$ of the infinite square lattice, and show that BUD is irreducible whenever $k = |V(G)| / 3$, i.e., each district is of size $3$.

\subsection{Irreducibility for simple grid graphs}\label{sub:irreducibility-simple}

We start by defining a large class of subgraphs of the infinite grid.

\begin{definition}\label{def:simple-grid-graph}
    A subgraph $G$ of the infinite square lattice $\Z^2$ is called a \emph{simple grid graph} if $G$ is comprised of all the vertices and edges that are on and interior to a simple closed loop $\alpha$ in $\Z^2$. We call the set of vertices of $\alpha$ the \emph{boundary} of $G$ and denote it by $\partial G$. The edges of $\alpha$ are called \emph{boundary edges} of $G$.
\end{definition}

Recall that $B^0_2(G)$ is the set of spanning trees $T$ of $G$ that contain an edge whose removal partitions $T$ into two trees with equal populations. More generally, if $a, b \in \N$ with $a + b = |V(G)|$, let $B^{a,b}_2(G)$ denote the set of spanning trees $T$ of $G$ that contain an edge whose removal partitions $T$ into two trees with sizes $a$ and $b$, respectively. Our goal is to prove the following theorem.

\begin{theorem}\label{thm:irreducible-simple}
    For any $n, a, b \in \N$ such that $a + b = n$, if $G$ is a simple grid graph with $n$ vertices, then  the BUD walk is irreducible on $B_2^{a,b}(G)$ and has diameter $O(n)$.
\end{theorem}

\noindent Each ``tentacle'' of a simple grid graph has width at least two; more precisely, 
\begin{lemma}\label{L:tentacle}
Each vertex of a simple grid graph $G$ has a neighbor in $G$ to its left or right. 
\end{lemma}
\begin{proof}
Let $\alpha:\{0,...,l\}\rightarrow\partial G$ be a parameterization of $\partial G$ (with $\alpha(0)=\alpha(l)$).  Let $\bar\alpha:[0,l]\rightarrow \R^2$ be the corresponding piecewise linear loop in $\R^2$.  According to the Jordan Curve Theorem, $\bar\alpha$ divides $\R^2$ into a bounded ``interior'' region and an unbounded ``exterior'' region. Furthermore, we can choose $\alpha$ and $\bar\alpha$ to be \emph{positively oriented} parameterizations, which intuitively means that the interior is on one's left as one traverses $\partial G$ this way.  More precisely, let $R_{90}$ rotate vectors in $\R^2$ counterclockwise by 90 degrees. Then $R_{90}(\bar\alpha'(t))$ points to the interior for all smooth points $t\in(0,l)$, which means that $\bar\alpha(t)+s\cdot R_{90}(\bar\alpha'(t))$ lies in the interior for all sufficiently small positive values of $s$ (and hence for all values of $s$ up to the first intersection with the integer grid).

It follows that if $v$ is a node of $G$ that is not in $\partial G$, then all four of its neighbors in $\Z^2$ must lie in $G$.  
Moreover, for an ``up'' edge (which means $\alpha(t+1)$ is above $\alpha(t)$) the left $\Z^2$-neighbors of $\alpha(t)$ and $\alpha(t+1)$ must lie in $G$.  Similarly, the right neighbors of a ``down'' edge lie in $G$, the top neighbors of a ``right'' edge lie in $G$, and the bottom neighbors of a ``left'' edge lie in $G$.  In all cases, the left or right $\Z^2$-neighbor of a boundary vertex lies in $G$.
 \end{proof}

We first require the following, which is essentially the $k=2$ version of Lemma~\ref{lemPartitionCharacterization}, but is generalized to $B_2^{a,b}(G)$ and framed in terms of BUD steps on trees rather than on partitions.
\begin{lemma}\label{L:k2}
    Let $G$ be a connected graph.  Let $T\in B_2^{a,b}(G)$ and let $T_1,T_2$ be subtrees (obtained by removing an edge of $T$) with $|V(T_1)|=a$ and $|V(T_2)|=b$.  For each $i\in\{1,2\}$, let $e_i$ be an edge of $T_i$ whose removal partitions $T_i$ into subtrees called $K_i$ and $H_i$.  Assume that $|V(H_1)|=|V(H_2)|$.  Assume that some vertex of $H_1$ is adjacent in $G$ to some vertex of $K_2$, and that some vertex of $H_2$ is adjacent in $G$ to some vertex of $K_1$.  Then $1$ or $2$ BUD steps suffice to transform $T$ into a tree $\tilde{T}\in B_2^{a,b}(G)$ whose corresponding subtrees $\tilde{T}_1,\tilde{T}_2$ have vertex sets $V(\tilde{T}_1)=V(K_1)\cup V(H_2)$ and $V(\tilde{T}_2)=V(K_2)\cup V(H_1)$.  
\end{lemma}
Lemma~\ref{L:k2} is pictured in Figure~\ref{F:Kequals2}. The initial configuration on the left colors $T_1$ red and $T_2$ blue.  The ending configuration on the right colors $\tilde{T}_1$ red and $\tilde{T}_2$ blue.  The lemma can be phrased in terms of colors as follows.  To turn a red subtree $H_1$ to blue and simultaneously turn a same-sized blue subtree $H_2$ to red, we only require that each $H_i$ is adjacent to a node (called its \emph{anchor}) that is and will remain the color it wants to turn.  Note that the term ``subtree'' implies that  removing $H_1$ would leave a \emph{connected} red piece (and similarly removing $H_2$ would leave a connected blue piece).

\begin{proof}
The lemma statement doesn't name the edge of $T$ whose removal partitions $T$ into $T_1$ (red) and $T_2$ (blue).  Let's call this edge $f_1$.  Using just one BUD step, we can make $f_1$ be any edge in $G$ that connects $T_1$ to $T_2$, so we can assume $f_1$ connects $H_1$ to $K_2$.  Let $f_2$ denote an edge of $G$ that connects $H_2$ to $K_1$.  Exactly as in the proof of Lemma~\ref{lemPartitionCharacterization}, the BUD step that adds $f_2$ and removes $e_2$ has the desired effect.
\end{proof}

\begin{figure}[bht!]\centering
\includegraphics[width=4in]{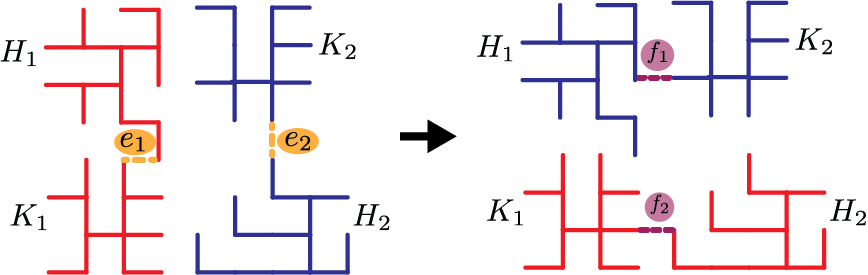}
\caption{The $k=2$ version of Figure~\ref{F:BUD_step}, where $T_1, \tilde{T}_1$ are red and $T_2, \tilde{T}_2$ are blue.}\label{F:Kequals2}
\end{figure}

The remainder of this section is devoted to the proof of Theorem~\ref{thm:irreducible-simple}.  Let $G$ denote a simple grid graph.  Let $T\in B_2^{a,b}(G)$, and let $T_1,T_2$ be two subtrees (obtained by removing an edge of $T$) with $|V(T_1)|=a$ and $|V(T_2)|=b$.  Let $G_1$ and $G_2$ denote the induced subgraphs of $G$ that are spanned by $T_1$ and $T_2$ respectively.

\begin{claim}\label{clm:not-cut-means-leaf}
    Let $T$ be a spanning tree of a graph $G$ with maximum degree $d$. If $v \in V(T)$ is not a cut vertex of $G$, then we can make $O(d)$-many BUD steps to make a spanning tree $T'$ of $G$ where $v$ is a leaf.
\end{claim}
\begin{proof}
    Let $u$ be a neighbor of $v$ in $T$, and denote by $e$ the edge connecting $v$ to $u$. If $v$ is not a leaf, let $T_1, T_2$ be two of the subtrees created by removing $e$. Since $v$ is not a cut vertex of $G$, the graph $G - v$ is connected, so there exists an edge $f \in E(G - v)$ that connects $T_1$ to $T_2$. Therefore the BUD step that adds $f$ and removes $e$ creates a new spanning tree $T'$ where $v$ does not neighbor $u$. Repeating this process a total of $\deg(v) - 1 \leq d$ times will turn $v$ into a leaf.
\end{proof}

\begin{claim}\label{clm:exists-ripe}
    Let $G_1, G_2$ be induced subgraphs of a simple grid graph $G$ as described above. There exists a vertex $v \in G_1$ that is not a cut vertex of $G_1$ and is adjacent to a vertex $u \in G_2$. The analogous statement holds if we swap the role of each subgraph.
\end{claim}
\begin{proof}
    Color the vertices in $G_1$ and $G_2$ red and blue, respectively, and say that a vertex is \emph{on the interface} if it has a differently-colored neighbor. Let $R_I$ be the set of red vertices that are on the interface. Choose an arbitrary $v_0 \in R_I$, which exists since $G$ is connected. If $v_0$ is not a cut vertex of $G_1$ we are done. Otherwise, let $v^*$ be the farthest vertex in $R_I$ from $v_0$, using the path metric through $G_1$. We claim $v^*$ is not cut. Otherwise, removing $v^*$ would induce a red connected component $R^+$ that does not interface with $G_2$. Consider the $3 \times 3$ ``box" of vertices around $v^*$. There is no way to color the vertices of the box without contradicting either the fact that $v^*$ is cut, that $R^+$ does not interface with $G_2$, or that $G$ is simple. One can see this by choosing any $v^+ \in R^+$ that neighbors $v^*$ and attempting to color the remaining vertices in the box, see Figure~\ref{fig:simple-cut-gives-ripe}.
\end{proof}

\begin{figure}[h]
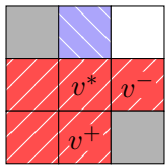

    \centering
    \begin{tikzgrid}[rows=3, cols=3]
        \colorCell{2}{2}{myred}
        \colorCell{2}{3}{myred}
        \colorCell{3}{2}{myred}
        \colorCell{3}{3}{mygray}
        \colorCell{3}{1}{myred}
        \colorCell{2}{1}{myred}
        \colorCell{1}{1}{mygray}
        \colorCell{1}{2}{myblue}

        \patternCell{1}{2}{sparse north west lines}{white}
        \patternCell{2}{2}{sparse north east lines}{white}
        \patternCell{2}{3}{sparse north east lines}{white}
        \patternCell{3}{2}{sparse north east lines}{white}
        \patternCell{3}{1}{sparse north east lines}{white}
        \patternCell{2}{1}{sparse north east lines}{white}

        \labelCell{2}{2}{$v^*$}
        \labelCell{2}{3}{$v^-$}
        \labelCell{3}{2}{$v^+$}
   
    \end{tikzgrid}
    \drawGrid[rows=3, cols=3, show indices=false, cell size=0.7]
    \caption{The configuration (modulo symmetry) were $v^*$ to be a cut vertex in the proof of Claim~\ref{clm:exists-ripe}. Vertices are represented by square cells, with northeastern strokes denoting red vertices and northwestern strokes denoting blue vertices. The vertex $v^-$ is closer to $v_0$ than $v^*$ is, via the path metric through $G_1$. Then there exists another red connected component $R^+$ (induced by deleting $v^*$) which is farther from $v_0$ than $v^*$ is, and therefore cannot interface with any blue vertices. To avoid contradictions, one must color three vertices of $R^+$, in this case on the bottom-left, forcing the top neighbor of $v^*$ to be blue. But then there is no coloring of the top-left and bottom-right vertices that avoids contradiction.}
    \label{fig:simple-cut-gives-ripe}
\end{figure}

\begin{lemma}\label{lem:reduce-to-simple}
$O(n)$ BUD steps suffice to transform $T$ to a configuration in which $G_1$ and $G_2$ are both simply connected; that is, both colors intersect $\partial G$ (neither color fully surrounds the other).    
\end{lemma}
\begin{proof}
    Suppose $T \in B_2^{a,b}(G)$ is such that only one of $G_1, G_2$ intersects $\partial G$. If $\min\{a, b\} = 1$, all we need to do is make any $v \in V(T) \cap \partial G$ a leaf. Since $G$ is a simple grid graph, the boundary is a single-color cycle, so no boundary vertex is a cut vertex. Therefore we can use Claim~\ref{clm:not-cut-means-leaf}, along with the fact that grid graphs have constant degree, to make $O(1)$-many BUD steps and make $v$ a leaf. 
    
    Otherwise, suppose that $G_1$ consists of at least two red vertices and does not intersect the boundary. Given any blue vertex $u_1 \in V(G_2)$, let $P_{u_1} = \{u_1, u_2, \dots, u_s\} \subseteq V(G_2)$ be the vertices of a minimal-length path from $u_1$ to the boundary $\partial G$, so $u_s$ is the only vertex of $P_{u_1}$ in $\partial G$. Consider the set $\mathcal{S}$ of all blue vertices $u$ such that $u$ is not a cut vertex of $G_2$ and is adjacent to a red vertex. $S$ is nonempty by Claim~\ref{clm:exists-ripe}. 

    We say that $u_1 \in S$ induces a \emph{bad path $P_{u_1}$} if there exists a \emph{bad vertex} $u_i \in P_{u_1}$ that is a cut vertex of the blue subgraph $G_2 - \{u_1, \dots, u_{i-1}\}$. We claim there exists a $u_1 \in S$ that does not induce a bad path. Choose any $u_1 \in S$ which minimizes $|P_{u_1}|$. If there exists a bad $u_i \in P_{u_1}$, it must be that $2 \leq i \leq s - 1$, so $s \geq 3$ and all four neighbors of $u_i$ must be blue. Consider the $3 \times 3$ box of vertices around $u_i$, see Figure~\ref{fig:bad-paths}. If $u_{i-1}$ and $u_{i+1}$ are collinear, the minimality of $P_{u_1}$ forces the two neighbors of $u_{i+1}$ (in the box) to be blue, which means that $u_i$ and hence $P_{u_1}$ cannot be bad. If $u_{i-1}$ and $u_{i+1}$ are not collinear and $u_i$ is bad, then there must be a red vertex in the box. If $s > 3$ this contradicts the minimality of $P_{u_1}$, and otherwise we can re-pick $u_1 = u_{i-1}$ to be collinear with $u_{i+1}$, reducing to the first case.

    \begin{figure}[h]
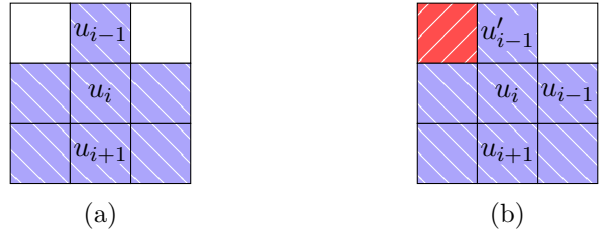

    \centering
        \begin{subfigure}{0.32\linewidth}
            \centering
            \begin{tikzgrid}[rows=3, cols=3]

                \colorCell{1}{2}{myblue}
                \colorCell{2}{1}{myblue}
                \colorCell{2}{2}{myblue}
                \colorCell{2}{3}{myblue}
                \colorCell{3}{1}{myblue}
                \colorCell{3}{2}{myblue}
                \colorCell{3}{3}{myblue}

                \patternCell{1}{2}{sparse north west lines}{white}
                \patternCell{2}{1}{sparse north west lines}{white}
                \patternCell{2}{2}{sparse north west lines}{white}
                \patternCell{2}{3}{sparse north west lines}{white}
                \patternCell{3}{1}{sparse north west lines}{white}
                \patternCell{3}{2}{sparse north west lines}{white}
                \patternCell{3}{3}{sparse north west lines}{white}

                \labelCell{2}{2}{$u_i$}
                \labelCell{3}{2}{$u_{i+1}$}
                \labelCell{1}{2}{$u_{i-1}$}
           
            \end{tikzgrid}
            \drawGrid[rows=3, cols=3, show indices=false, cell size=0.8]
            \caption{}
            \label{fig:bad-paths-a}
        \end{subfigure}
        \begin{subfigure}{0.32\linewidth}
            \centering
            \begin{tikzgrid}[rows=3, cols=3]
                \foreach \i in {1,2,3}{
                    \foreach \j in {1,2,3}{ 
                        \labelCell{\i}{\j}{}
                        \colorCell{\i}{\j}{white}
                    }
                } 

                \colorCell{1}{1}{myred}
                \colorCell{1}{2}{myblue}
                \colorCell{2}{1}{myblue}
                \colorCell{2}{2}{myblue}
                \colorCell{2}{3}{myblue}
                \colorCell{3}{1}{myblue}
                \colorCell{3}{2}{myblue}
                \colorCell{3}{3}{myblue}

                \patternCell{1}{1}{sparse north east lines}{white}

                \patternCell{1}{2}{sparse north west lines}{white}
                \patternCell{2}{1}{sparse north west lines}{white}
                \patternCell{2}{2}{sparse north west lines}{white}
                \patternCell{2}{3}{sparse north west lines}{white}
                \patternCell{3}{1}{sparse north west lines}{white}
                \patternCell{3}{2}{sparse north west lines}{white}
                \patternCell{3}{3}{sparse north west lines}{white}

                \labelCell{2}{2}{$u_i$}
                \labelCell{3}{2}{$u_{i+1}$}
                \labelCell{2}{3}{$u_{i-1}$}
                \labelCell{1}{2}{$u'_{i-1}$}
           
            \end{tikzgrid}
            \drawGrid[rows=3, cols=3, show indices=false, cell size=0.8]
            \caption{}
            \label{fig:bad-paths-b}
        \end{subfigure}
    \caption{Helpful diagrams for Lemma~\ref{lem:reduce-to-simple}. Vertices are represented by square cells, with northwestern strokes denoting blue vertices and northeastern strokes denoting red vertices. In Case (a), $u_i$ is not a bad vertex, since its removal cannot increase the number of connected components of $G_2 - \{u_1, \dots, u_{i-1}\}$. In Case (b), the path $u_{i-1}, u_i, u_{i+1}$ is bad, but the path $u'_{i-1}, u_i, u_{i+1}$ reduces to Case (a) so is not bad.}
    \label{fig:bad-paths}
\end{figure}

    By the argument above, we can choose a $u_1 \in S$ that induces a path $P_{u_1}$ that is not bad. Let $v \in V(G_1)$ be a red neighbor of $u_1$. We claim there exists vertices $v' \neq v, u' \neq u_1$ (red and blue, respectively) such that $v'$ is not a cut vertex of $G_1$ and is adjacent to $u' \in G_2$. Consider a red block $L \subseteq G_1$ which corresponds to a leaf of the block-cut tree of $G_1$. The \emph{top-right corner} of $L$ is the the top vertex in $L$ among all of the rightmost vertices in $L$, and we can define the other corners similarly. Note that each corner will have at least two blue neighbors, one of which will not be $u_1$. Moreover, $L$ has at most one cut vertex. If $L$ has at least two non-cut corners, we can choose $v'$ to be one that is not $v$. Otherwise, there must be two such leaves in the block-cut tree, so we can still choose such a $v'$.
    
    Since both $v'$ and $u_1$  are non-cut vertices, by Claim~\ref{clm:not-cut-means-leaf}, in $O(1)$-many BUD steps we can make $v'$ a leaf of $T_1$, and $u_1$ only connected to $u_2$ in $T_2$. Now set $H_1 = \{v'\}$ and $H_2 = \{u_1\}$. By our choice of $v'$ and $u_1$, we can use Lemma~\ref{L:k2} to take $O(1)$-many BUD steps to color $v'$ blue and $u_1$ red. Recall that $u_2$ is not a cut vertex in the new blue subgraph $G_2 - u_1 + v'$. Therefore we can repeat this process a total of $s < n$ times until $u_s$ is red, so both colors intersect $\partial G$ and we are done.
\end{proof}

Thus, we can henceforth assume simple connectivity, and we will only use BUD moves that preserve simple connectivity.  

As in the proof of Lemma~\ref{lemInternalStep}, we can, in $O(n)$ steps, separately transform $T_1$ and $T_2$ into arbitrary spanning trees of $G_1$ and $G_2$.  More specifically, we will choose trees with the minimum possible number of horizontal edges.  We next explain the construction of such trees, and simultaneously define the \emph{column-component graph} of $T$, denoted $\C(T)$.

For $i\in\{1,2\}\approx \{\text{red},\text{blue}\}$, $T_i$ initially contains all vertices of $G_i$ and all vertical edges of $G_i$ between pairs of these vertices.  The nodes of $\C(T)$ are defined as the connected components of $T_1$ and $T_2$ at this point; that is, they are the maximum same-color vertical paths.  Each node of $\C(T)$ is either red or blue.  For each pair $X,Y$ of same-color nodes of $\C(T)$, we add an edge between them if some vertex $x$ of $X$ is connected by a horizontal edge of $G$ to some vertex $y$ of $Y$.  In this case, we also add a corresponding horizontal edge to $T_1$ or $T_2$, and we choose arbitrarily whenever there are more than one such adjacent $(x,y)$ pairs.   Figure~\ref{F:column_graph} (left) exemplifies two such trees in red and blue.  

In Figure~\ref{F:column_graph} (right), the nodes of $\C(T)$ are illustrated as vertical bars.  For clarity, we will henceforth use the term ``vertices'' for $G$ and ``nodes'' for $\C(T)$.  We will denote nodes with upper case letters and vertices with lower case letters. 

\begin{figure}[bht!]\centering
\includegraphics[width=4in]{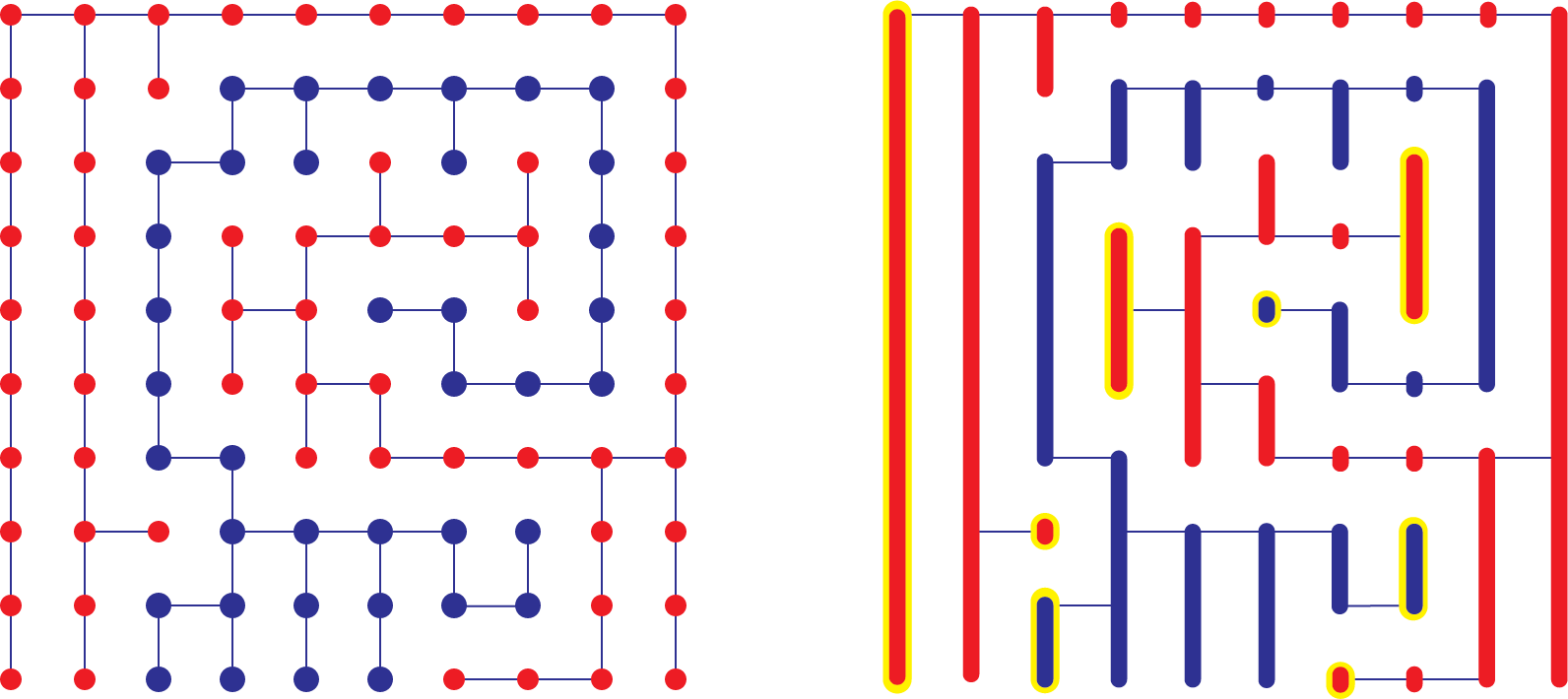}
\caption{Left: The red and blue spanning trees contain a minimal number of horizontal edges.  Right: The nodes of $\C(T)$ are the maximal same-color columns, and the leaves of $\C(T)$ are highlighted in yellow. In grayscale, the darker nodes are blue.}\label{F:column_graph}
\end{figure}

\begin{lemma}
For each $i\in\{1,2\}$, $T_i$ is a spanning tree of $G_i$.  Furthermore, $\C(T)$ is a forest with two components: a red one that we will denote $\C_1(T)$ and a blue one that we will denote $\C_2(T)$.
\end{lemma}
\begin{proof}
We first argue that $T_1$ is connected.  For this, since $G_1$ is connected, there exists a path $\beta$ in $G_1$ between an arbitrary pair of red vertices.   We can replace $\beta$ with a path in $T_1$ by replacing each horizontal edge $(v_1,v_2)$ of $\beta$ with the natural path that uses the unique horizontal edge in $T_1$ between the nodes of $\C_1(T)$ that contain $v_1$ and $v_2$.   This shows that $T_1$ (and hence also $\C_1(T)$) is connected.

We next argue that $C_1(T)$ is acyclic.  Suppose to the contrary that there was a cycle $\beta:\{0,...,x\}\rightarrow \C_1(T)$.  Let $t_0$ be such that $\beta(t_0)$ is a furthest-right node of the cycle.  Then $\beta(t_0-1)$ and $\beta(t_0+1)$ are both to the left of $\beta(t_0)$.  Since $\beta(t_0-1)\neq\beta(t_0+1)$, there must be a node $z$ of $\Z^2$ between them that is not red, so $z$ is either blue or not in $G$.  The node $z$ is encircled by the cycle, which contradicts simple connectivity.  This shows that $C_1(T)$ (and hence also $T_1$) is acyclic.

In summary,  $T_1$ is a spanning tree of $G_1$ and $C_1(T)$ is connected and acyclic.  Similarly, $T_2$ is a spanning tree of $G_2$ and $C_2(T)$ is connected and acyclic.
\end{proof}

Every tree has at least two leaves, so $\C(T)$ has at least two red leaves and at least two blue leaves.  In Figure~\ref{F:column_graph}, the leaves are highlighted in yellow.  

Let $l=|\partial G|$ and let $\alpha:\{0,1,...,l\}\rightarrow \partial G$ be a parameterization of the boundary of $G$ that is positively oriented (as defined in the proof of Lemma~\ref{L:tentacle}), with $\alpha(0)=\alpha(l)$.  Define $\hat{\alpha}:\{0,1,...,l\}\rightarrow \C(T)$ so that $\hat\alpha(t)$ is the node of $\C(T)$ containing $\alpha(t)$ for every $t\in\{0,1,...,l\}$.  Note that $\hat\alpha$ is not necessarily one-to-one.  A node of $\C(T)$ is called a \emph{boundary node} if it equals $\hat\alpha(t)$ for some $t\in[l]$; that is, it intersects $\partial G$.  

There are exactly two positions along the boundary of $G$ at which the vertex color changes between red and blue.  This is because more than two would contradict the connectivity of the colors, while fewer would contradict the simply connected assumption.  The blue and red boundary vertices of $G$ at each of these two positions will be called \emph{transition} vertices, and the corresponding boundary nodes of $\C(T)$ will be called \emph{transition} nodes. 

A node $V$ of $\C(T)$ is called a \emph{sandwiched} node if the top vertex comprising $V$ has an opposite-color $G$-neighbor directly above it, and the bottom vertex comprising $V$ has an opposite-color $G$-neighbor directly below it.  Note that a sandwiched node could also be a boundary node, and a non-boundary node must be sandwiched.

\begin{lemma}\label{L:interior_reduction}
    If $\C(T)$ has a red sandwiched leaf $L_1$ and a blue sandwiched leaf $L_2$, then $O(1)$ BUD steps suffice to reduce the number of nodes of $\C(T)$.
\end{lemma}
\begin{proof}

First consider the case in which neither of $L_1$ and $L_2$ is directly above the other.  For $i\in\{1,2\}\approx\{\text{red},\text{blue}\}$, let $x_i,y_i$ be the size (number of vertices) of the top and bottom portion of $L_i$ respectively (which means the portions above/below the horizontal edge of $T$ connected to $L_i$).  For example in the initial configuration shown in Figure~\ref{F:Bud_grid_change}, the red leaf has $x_1=1$ top vertices and $y_1=1$ bottom vertices.  The blue leaf has $x_2=2$ top vertices and $y_2=0$ bottom vertices.  Our counting convention here is that the ``middle'' vertex of a leaf (the one adjacent to the horizontal edge) goes uncounted when the leaf has both top and bottom parts, but is included in the count for a leaf that has only a bottom or only a top part.  

Among the non-zero members of $\{x_1,y_1,x_2,y_2\}$, we claim that one or two BUD steps suffice to reduce the smallest to zero.  Therefore, removing half a leaf at a time like this, it takes at most 8 BUD steps to eliminate at least one of the leaves.  This does not necessarily reduce the number of \emph{leaves} of $\C(T)$ (since new ones might be created, as happens in Figure~\ref{F:Bud_grid_change}), but it does reduce the number of \emph{nodes} of $\C(T)$. 

Since the roles of top and bottom are symmetric, and the roles of blue and red are symmetric, we can assume without loss of generality that $x_1=\min\{z > 0: z \in \{x_1,y_1,x_2,y_2\} \}$ and $x_1,x_2\neq 0$, which is the case in Figure~\ref{F:Bud_grid_change}.  That is, we wish to reduce $x_1$ to $0$ (remove the top of the red leaf) while also reducing the $x_2$ (shrinking the top of the blue leaf).  The required (one or two) BUD steps are given by Lemma~\ref{L:k2}.  Figure~\ref{F:Bud_grid_change} shows the effect of applying Lemma~\ref{L:k2} to a specific pair of leaves from Figure~\ref{F:column_graph}.  

\begin{figure}[bht!]\centering
\includegraphics[width=5in]{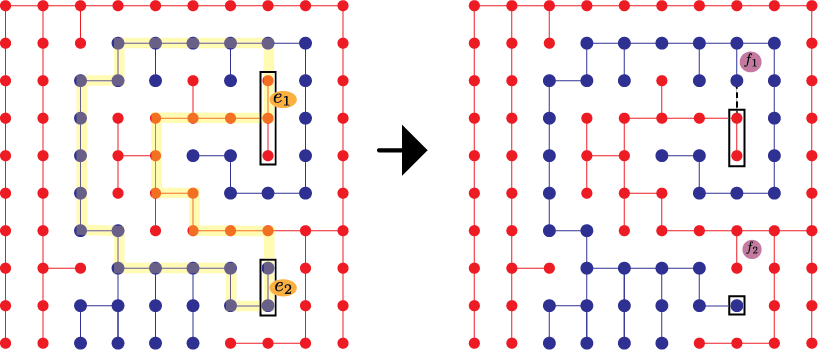}
\caption{Starting from the configuration of Figure~\ref{F:BUD_step}, Lemma~\ref{L:k2} allows us to eliminate the top of the red leaf.  The active leaves are shown with black boxes.  The up-down step's cycle is highlighted yellow.  The edge labels $\{e_1,e_2,f_1,f_2\}$ match the previous figures.}\label{F:Bud_grid_change}
\end{figure}

Next, consider the case in which one of the leaves is directly above the other; let's say that $L_1$ is directly above $L_2$.  We need to ensure that each half-leaf that we will adjust is adjacent to an opposite-color anchor. We begin by moving the horizontal edge connected to each leaf as close as possible to the other leaf.  For example, let's say for specificity that $L_1$ is connected in $\C(T)$ to a red node $N_1$ on its left.  Using a single BUD step, we can replace the horizontal edge of $T$ that connects vertices of $L_1$ and $N_1$ with the \emph{lowest} horizontal edge connecting a pair of vertices from these two nodes.  If $L_1$ no longer has a bottom portion after this change, then we need only remove the top portion of $L_1$, which has a blue anchor directly above it, so we are done.  So assume that $L_1$ still has a bottom portion after this change, and let $a,b$ be the nodes of $\Z^2$ to the left and right of the bottom vertex $v$ of $L_1$.  According to Lemma~\ref{L:tentacle}, $a$ or $b$ is a neighbor of $v$ in $G$, so it must be blue because $L_1$ is a leaf.   This neighbor serves as an anchor for $L_1$.

Similarly, after replacing $L_2$'s horizontal connector with its \emph{highest} possible counterpart, if $L_2$ still has a top part, then its topmost vertex is adjacent (on its left or right) to a red vertex that can serve as its anchor.  We can therefore proceed like in the previous argument. 
\end{proof}

\begin{lemma}\label{L:all_boundary}
If $\C(T)$ has no blue sandwiched leaves, then every blue node of $\C(T)$ is a boundary node.  Furthermore, the analogous statement holds for red. 
\end{lemma}

\begin{proof}
Suppose, to the contrary, that there is a blue node, $B$, that does not intersect $\partial G$.  There must be red nodes $R_1$ and $R_2$ directly above and below $B$.  Since $B$ is not a leaf, it must have at least two blue neighbors.

First consider the case in which $B$ has blue neighbors on both sides: a blue neighbor $B_1$ on its left and a blue neighbor $B_2$ on its right.  By the simple connectivity assumption, one of these -- let's say $B_1$ --  is connected to $\partial G$ via a blue path in $\C(T)$ that doesn't include $B$.  Deleting $B$ would split $\C_2(T)$ into at least two components; let $C$ denote the component containing $B_2$.  Since $C$ is a tree, it has at least two leaves, so it has at least one leaf, $L$, that's different from $B_2$.  We claim that $L$ does not intersect $\partial G$, so in particular it must be a sandwiched node, contradicting the hypothesis that there are no blue sandwiched leaves.  This is because, if we suppose to the contrary that $L$ intersects $\partial G$, then we would have a blue path $\partial G\rightarrow B_1\rightarrow B\rightarrow B_2\rightarrow L\rightarrow\partial G$ that would separate the $\C_1(T)$ into two components: one containing $R_1$ and the other containing $R_2$, contradicting the connectedness of $C_1(T)$.

Next consider the case in which $B$ only has blue neighbors on one side -- let's say its left.  Let $B_1,B_2$ be two blue neighbors of $B$ on its left (if there are more than two, then choose a consecutive pair, which means that they don't have any other blue nodes between them).  If there were a red node between $B_1$ and $B_2$, then the argument from the previous paragraph would contradict the connectedness of $\C_1(T)$.  Therefore the nodes of $\Z^2$ between $B_1$ and $B_2$ do not belong to $G$, which implies that the bottom node, $b_1$, of $B_1$ and the top node, $b_2$, of $B_2$ both belong to $\partial G$.  But this forces $B$ to be a boundary node, contradicting our assumption that it is not.  
\end{proof}

Recall that, for each adjacency in $\C_2(T)$, the choice of a corresponding horizontal edge in $T_2$ was made arbitrarily.  A BUD step suffices to re-select such a horizontal edge.

\begin{corollary}\label{C:all_boundary}
If $\C(T)$ has no blue sandwiched leaves, then $O(n)$ BUD steps suffice to re-select the horizontal edges of $T_2$ such that they are all boundary edges.  Furthermore, the analogous statement holds for red.
\end{corollary}

\begin{proof}
The blue component of $\partial G$ is connected; that is, $\alpha$ can be re-parameterized so that its restriction to $\{1,2,...,l_2\}$ corresponds to the blue.  Remove all horizontal edges of $T_2$, and then add all edges of the form $(\alpha(t),\alpha(t+1))$ that are horizontal for $1\leq t\leq l_2-1$.  The previous lemma implies that $T_2$ is now connected, so it can be turned into a tree by removing horizontal edges until it is cycle-free. 
\end{proof}

\begin{lemma}\label{L:no_sandwich}
$O(n)$ BUD steps suffice to transform $T$ into a configuration with no red or blue sandwiched leaves.
\end{lemma}

\begin{proof}
Whenever there are sandwiched leaves of both colors, Lemma~\ref{L:interior_reduction} allows us in $O(1)$ steps to reduce the number of nodes of $\C(T)$.  Since $\C(T)$ starts with $O(n)$ nodes, repeatedly applying this process must terminate after $O(n)$ steps in a configuration that has no sandwiched leaves of one color.  Let's say that $\C(T)$ has no blue sandwiched leaves.  Corollary~\ref{C:all_boundary} allows us in $O(n)$ BUD steps to ensure that all horizontal edges of $T_2$ are boundary edges.

It remains to show that $O(n)$ BUD steps suffice to eliminate all red sandwiched leaves.  For this, our measurement of progress will be the following. 
$$\Phi =  \text{the number of nodes of $\C(T)$ plus three times the number of blue vertices of $\partial G$}.$$ 
It will suffice to prove the following: 
$$\text{\emph{If there is a red sandwiched leaf, then $O(1)$ BUD steps suffice to reduce $\Phi$.}}$$  Since the initial value of $\Phi$ is $O(n)$, repeatedly applying this process must terminate after $O(n)$ steps in a configuration that has no red or blue sandwiched leaves.  

For this, let $R$ denote a red sandwiched leaf.  Let $b=\alpha(t_0)\in\partial G$ denote a blue transition vertex, and let $B=\hat\alpha(t_0)$ denote the corresponding blue transition node.  Similar to the proof of Lemma~\ref{L:interior_reduction}, we will apply $O(1)$ applications of Lemma~\ref{L:k2}.  Each of these BUD moves will turn part of $R$ blue (all or part of its top or bottom) and simultaneously turn part of $B$ red.  After these moves, either $R$ will be entirely blue (which decreases the number of nodes of $\C(T)$ and therefore decreases $\Phi$) or a connected portion of $B$ including $b$ will be red (which might increase the number of nodes of $\C(T)$ by up to $2$, but also decreases the number of blue nodes of $\partial G$, yielding a net decrease in $\Phi$). 

We assume for now that $R$ is not directly above or below $B$, which will simplify the proof by alleviating the need to find alternative anchors for them.  Recall that $\alpha$ is positively oriented.  We can assume that $B$ represents the transition from red to blue; that is, $r=\alpha(t_0-1)$ is red; the other case is essentially the same.  We can also assume that $b_2=\alpha(t_0+1)$ is blue, for if it were not, there would be only one blue boundary node and hence only one blue node at all, which would imply that there are no red sandwiched leaves.

We proceed by cases depending on the cardinal directions of the edges $(r,b)$ and $(b,b_2)$.  For example, the case $\leftarrow \uparrow$ means that the edge $(r,b)$ goes left and the edge $(b,b_2)$ goes up.  There are $4^2=16$ ways to choose these two directions, but the choices $\uparrow\downarrow, \downarrow\uparrow, \leftarrow\rightarrow, \rightarrow\leftarrow$ are not possible, so there are $12$ valid possibilities.  Pairs of them are equivalent modulo $180^\circ$ rotation, so we must consider $6$ cases.   

\begin{itemize}
\item (CASE $\leftarrow\uparrow$): Two examples of this case are illustrated in Figure~\ref{F:nonconvex}.  Let $B^-$ denote the bottom portion of $B$; that is, the portion below and including $b$.  

If $B^-$  intersects $\partial G$  only at $b$ (like in the top figure),  then a BUD move can be made to turn $B^-$ red.  To see that the BUD move is valid, notice $B^-$ is a subtree because it does not have any blue $T_2$-neighbors on its right or left (because all horizontal edges of $T_2$ are boundary edges), so the removal of the edge $(b,b_2)$ separates $B^-$ from the rest of $T_2$.  Since $B^-$  intersects $\partial G$  only at $b$, there is a red vertex below $B^-$ to serve as an anchor.

On the other hand, suppose there is a vertex $b'$ of $B^-$ other than $b$ that lies in $\partial G$ (like in the bottom figure).  We will perform a preliminary BUD step to change $C_2(T)$ by removing the vertical edge below $b$.  For this, we must find an edge to add that connects the resulting two components of $\C_2(T)$.  There are two paths in $\partial G$ from $b$ to $b'$; namely, the clockwise and counterclockwise paths.  Since the blue portion of $\partial G$ is connected, one of these two paths must be entirely blue; let's call this path $\beta$. Any missing horizontal edge of this blue path will fit the bill.  After this BUD step on trees, Lemma~\ref{L:k2} provides a BUD step that turns $b$ red (and turns part of $R$ blue).  The existence of the blue boundary segment $\beta$ implies that this BUD step does not create any blue sandwiched leaves.
\begin{figure}[bht!]\centering
\includegraphics[width=5in]{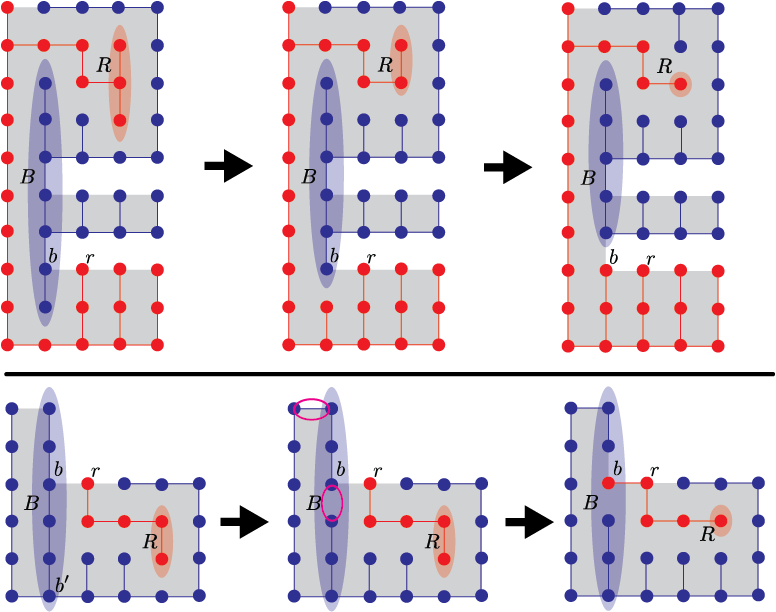}
\caption{A red sandwiched leaf $R$ is turned blue while a blue transition node $B$ is turned red.  In the bottom configuration, a preliminary BUD step is needed to internally change the blue spanning tree.}\label{F:nonconvex}
\end{figure}

\item (CASE $\rightarrow\rightarrow$): Note that $b$ is the bottom vertex of $B$.   If $B$  intersects $\partial G$  only at $b$, then $B$ is a subtree, and there is a red anchor above $B$, so we can turn $B$ red.

If $B$ contains a vertex $b'$ other than $b$ that lies in $\partial G$, then, as in the previous case, we perform a preliminary BUD step to change $\C_2(T)$ by removing the vertical edge above $b$, and then perform a  BUD step that turns $b$ red.

\item (CASE $\uparrow\leftarrow$): Here $B=\{b\}$ and the only $G$-neighbors of $b$ are $r$ and $b_2$.  A single BUD step turns $b$ red.

\item (CASE $\leftarrow\downarrow$): Here the only $G$-neighbors of $b$ are $r$ and $b_2$.  As before, we perform a BUD step to turn $b$ red.

\item (CASE $\uparrow\uparrow$): Here the only $T_2$-neighbor of $b$ is $b_2$.  Note that the node to the left of $b$ might be a blue $G$-neighbor of $b$, but is not a $T_2$-neighbor because every horizontal edge of $T_2$ is a boundary edge.  We can therefore perform a BUD step that turns $b$ red.

\item (CASE $\downarrow\leftarrow$): Note that $b$ is the top vertex of $B$.   If $B$  intersects $\partial G$  only at $b$, then $B$ is a subtree, and there is a red anchor below $B$, so we can turn $B$ red.

If $B$ contains a vertex $b'$ other than $b$ that lies in $\partial G$, then we perform a preliminary BUD step to remove the vertical edge below $b$, and then perform a  BUD step that turns $b$ red.

\end{itemize}

We've been assuming that $R$ is not directly above or below $B$.  If it were, then we could use the other blue transition node in place of $B$.  It therefore remains to consider the case in which switching to the other blue transition node does not work.  

That is, we must now consider the case in which one blue transition node $B_1$ is directly above $R$ and the other $B_2$ is directly below $R$.  Let $R'$ be the unique red $T_1$-neighbor of $R$.  Assume without loss of generality that $R'$ is to the left of $R$.  Let $b_1$ be the transition vertex contained in $B_1$ and let $b_2$ be the transition vertex contained in $B_2$.  There are two paths in $\partial G$ from $b_1$ to $b_2$, namely the clockwise and counterclockwise ones.  One of these paths, which we call $\beta$, is entirely blue.

If $\beta$ follows the positive orientation of the boundary, as in Figure~\ref{F:flanked} (left), then there must be another sandwiched red leaf, so we can use that one instead.  In fact, the other red leaf works because $\beta$ encircles $T_1$ and therefore all red nodes in $\C_1(T)$ are sandwiched.

If $\beta$ follows the negative orientation of the boundary, as in Figure~\ref{F:flanked} (right), the all vertices to the left of $R$ are red.  In this situation, we can turn $R$ blue using an anchor from $B_2$ while turning the bottom vertex of $B_1$ red using the red anchor to its left.
\begin{figure}[bht!]\centering
\includegraphics[width=2.5in]{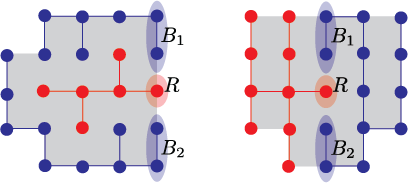}
\caption{Configurations in which both blue transition nodes are above/below the red sandwiched leaf $R$.}\label{F:flanked}
\end{figure}
\end{proof}

\begin{proof}[Proof of Theorem~\ref{thm:irreducible-simple}]
We are given a target member of $B_2^{a,b}(G)$.  Let $T$ be an arbitrary member of $B_2^{a,b}(G)$ which we wish to transform into the target.  It will suffice to transform $T$ so that all vertices have the target color, since we can use the argument from the proof of Lemma~\ref{lemInternalStep} to transform the individual trees $T_1$ and $T_2$ to their corresponding targets in $O(n)$ steps.

By Lemma~\ref{L:no_sandwich}, $O(n)$ BUD steps suffice to transform $T$ into a configuration with no red or blue sandwiched leaves.  At this point, according to Lemma~\ref{L:all_boundary}, all red and blue nodes of $\C(T)$ are boundary nodes.  Furthermore Corollary~\ref{C:all_boundary} allows us, in $O(n)$ BUD steps, to ensure that all horizontal edges of $T$ are boundary edges.  We apply the same process to the target configuration, so it similarly has no sandwiched leaves.

Let $b_1,r_1,r_2,b_2$ be the transition nodes along $\partial G$ in the order in which they are encountered using the positively oriented parameterization (with $b_i$ blue and $r_i$ red for each $i\in\{1,2\}$).  Let $B_1,R_1,R_2,B_2$ be the corresponding transition nodes in $\C(T)$.

As in Figure~\ref{F:rotate}, the color of almost every vertex is determined by the locations of these transition vertices.  More precisely, a \emph{stripe} means a connected segment of a vertical column of $\Z^2$ for which the top and bottom vertices are boundary vertices, but the other vertices are not boundary vertices.  Every non-boundary vertex of $G$ lies in a stripe.  If the top and bottom vertices of a stripe have the same color, then the whole stripe must have that color (otherwise there would be a node of $\C(T)$ that wasn't a boundary node).  For a stripe whose top vertex is red and bottom is blue (a \emph{red-over-blue stripe}), the color changes exactly once along this stripe, although the location of this change is arbitrary.  Similarly, for a blue-over-red stripe, the color changes exactly once at an arbitrary position.  Thus, the color of every vertex is determined by the locations of $r_1$ and $r_2$, except for the positions along the bi-colored stripes at which the color changes.

\begin{figure}[bht!]\centering
\includegraphics[width=3in]{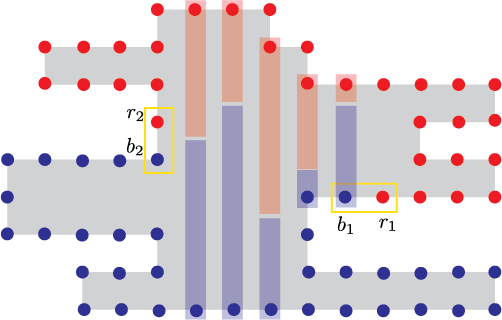}
\caption{All colors are determined by $r_1$ and $r_2$, except for the positions along the bi-colored stripes at which the color switches.}\label{F:rotate}
\end{figure}

In previous proofs, we made $O(1)$ BUD steps to turn all or part of a red set blue and simultaneously turn all or part of a blue set red.  Let's call these sets the \emph{red changer} and the \emph{blue changer}.  For example, in the proof of Lemma~\ref{L:interior_reduction}, the red changer was a red sandwiched leaf, while the blue changer was a blue sandwiched leaf.  Then, in the proof of Lemma~\ref{L:no_sandwich}, the red changer was a red sandwiched leaf, while the blue changer was a blue transition node.  For the following proof, each changer will be a transition node or a \emph{half-stripe} (which means the single-color portion of a bi-color stripe, not including the boundary vertex).

\begin{itemize}
\item STEP 1: ROTATE THE BOUNDARY COUNTERCLOCKWISE.  We use $R_1$ as the red changer and $B_2$ as the blue changer, so that one or both of the color transition points along the boundary rotates counterclockwise.  The proof of Lemma~\ref{L:no_sandwich} shows that it is always possible to make such a move.  We repeat this step until $r_1$ agrees with the target.  This requires $O(n)$ total BUD steps because the size of the boundary is $O(n)$.
\item STEP 2: ALIGN THE OTHER BOUNDARY TRANSITION.  Now $r_1$ is in its target position.  Suppose that $r_2$ needs to rotate counterclockwise to reach its target position.  In this case, we use $B_2$ as the blue changer, and we use any red half-stripe as the red changer.  Similarly, if $r_2$ needs to rotate clockwise to reach its target, then we use $R_2$ as the red changer, and we use any blue half-stripe as the blue changer.  We repeat this step until $r_2$ agrees with the target.  This requires $O(n)$ total BUD steps because there are $O(n)$ vertices that are part of bi-colored stripes, and we only ever change them in one direction.
\item STEP 3:  MATCH THE STRIPES.  Now $r_1$ and $r_2$ are in their target positions, so it only remains to match the color transition points of all bi-colored stripes.  For this, both changers will be half-stripes.  We repeat this step until all bi-colored stripes match their target.  This requires $O(n)$ total BUD steps because there are $O(n)$ vertices that are part of bi-colored stripes, and we only ever change them in one direction.
\end{itemize}
\end{proof}

\subsection{Irreducibility for triominoes on rectangular grids}\label{sub:irreducibility-triominoes}

In this section we prove that BUD is irreducible on rectangular grids when $k = n/3$, that is, when each part of the partition has exactly three vertices. In analogy with the tiling literature, we find it helpful to visualize the {\it dual} of the grid graph, where each vertex is a square and a partition into $k = n/3$ pieces is equivalent to a tiling of this dual with {\it triominoes}, polygons in the plane made of three unit squares connected edge-to-edge.  
Our proof uses a computational case analysis.

\begin{theorem}\label{thmTriomino}
    On an $N$-vertex grid graph, any pair of 3-omino tilings are connected by a series of at most $10N$ BUD moves, and any pair of spanning trees that are splittable into 3-ominoes are also connected by $\frac{61}{3}N$ Up-Down moves by intermediate trees that are also splittable into 3-ominoes.
\end{theorem}

\begin{proof}
    We assume without loss of generality that the horizontal dimension of the grid is divisible by 3. (If there are any 3-omino tilings at all, one dimension must be divisible since 3 is prime.) It suffices to show that any tiling is connected in $O(N)$ BUD moves to the canonical tiling where all tiles are horizontal. We order the cells by row top to bottom, and left to right within each row. Given any tiling $P$ other than the all-horizontal tiling, let $c$ be the first cell in this ordering that is not in a horizontal tile. For convenience, we  adopt a Cartesian coordinate system with the boundary of the grid at at integer coordinates such that cell $c$ is above and to the right of coordinate $(10, 10)$; see Figure~\ref{figTriominoProgram}). It suffices to show that we can make $c$ part of a horizontal tile in $O(1)$ BUD moves.

    \ipncm{.5}{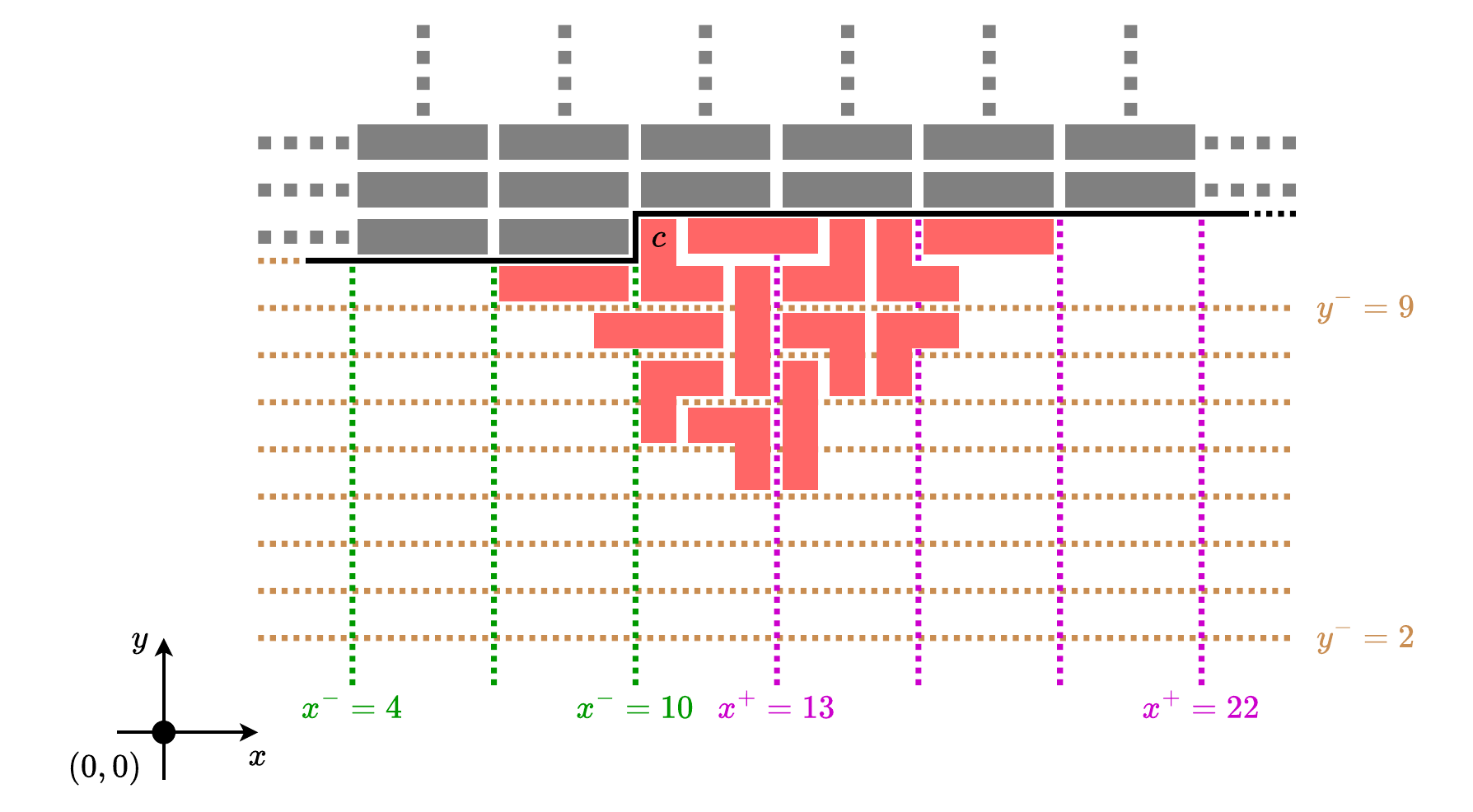}{\label{figTriominoProgram}Illustration of the cases in each step of Theorem~\ref{thmTriomino}. In each case, we select one of the 3 vertical lines on the left side as a left boundary, one of the 4 vertical lines on the right side as a right boundary, and one of the 8 horizontal lines as a bottom boundary.}

    Let $x^-$ be the $x$-coordinate of the left boundary, let $x^+$ be the $x$-coordinate of the right boundary, and let $y^-$ be the $y$-coordinate of the bottom boundary. We break into disjoint cases, depending on the choices of
    \begin{align*}
        x^- \in \{(\leq 4), 7, 10\}, && x^+ \in \{13, 16, 19, (\geq 22)\}, && y^- \in \{(\leq 2), 3, 4, 5, 6, 7, 8, 9\},
    \end{align*}
    where, for instance, $(\leq 4)$ denotes any case where the true value of $x^-$ is at most 4 in the chosen coordinate system. These encapsulate all possible cases, since the $x$-dimension is assumed to be divisible by 3 and $y^- = 10$ would imply that the tiling is already all-horizontal. In total, there are 3 cases for $x^-$, 4 cases for $x^+$, and 8 choices for $y^-$, so there are 96 cases overall.
    
    In each case, we computationally verify that, no matter what the local neighborhood around $c$ looks like, we can cover $c$ with a horizontal cell in $O(1)$ BUD moves. To accomplish this, we iteratively enumerate partial tilings of $c$ and all later cells in the ordering with $s$ cells that are as close to the cell two cells to the right of $c$ as possible, for $s = 1, 2, 3, \dots$, until we find that we can prove \emph{any} such tiling can be transformed into one with a horizontal tile covering $c$. Figure~\ref{figTriominoProgram} shows one such partial tiling, with $s = 13$. Checking the BUD connectivity condition is expensive, so we employ memoization, keeping track of sets of cells on which the BUD metagraph of all tilings contains a horizontal cell covering $c$ in every connected component.
    
    Note that the maximum number of tiles that ever needed to be placed was $s = 13$. Thus, we are able to handle the ``infinite'' cases where some values of $x^-$, $x^+$ or $y^-$ are arbitrarily low/high by recording the most extreme cells that were ever considered in the algorithm. For instance, in the most general case, where $x^- \leq 4$, $x^+ \geq 22$, and $y^- \leq 2$, we only needed to consider $x$ coordinates in the range $[4, 21]$ and $y$ coordinates above 2. Thus, we may actually consider the finite problem where $x^- = 4$, $x^+ = 22$, and $y^- = 2$, as the computation is equivalent. So in fact, all cases are finite, as long as we check that the maximum bounds of cells considered did not exceed those that were the farthest away. And in no instance where the bounds farther away than in this case.

    To get concrete upper bounds, we also recorded the largest number edges that needed to be added/removed in any of the cases, which was 30. This must have involved at least 15 BUD moves per tile being made horizontal, of which there are $\frac{N}{3}$. Thus, it is possible to make all tiles horizontal in at most $\frac{15N}{3} = 5N$ BUD moves, so it is possible to join an arbitrary pair in $10N$ moves. Counting tree moves rather than partition moves, the number of new edges added is precisely the worst-case number of tree moves needed to reconfigure the induced tiling appropriately, so it takes 30 moves per tile rather than 15. Hence, it takes $\frac{60N}{3}$ to transition from one tiling to any other. Finally, an additional $\frac{N}{3} - 1 \leq \frac{N}{3}$ external moves may be needed at the very end to reconfigure the tree while leaving the induced tiling the same, for a total of $\frac{61N}{3}$.
\end{proof}

We do not know whether it is possible to connect 3-omino tilings with 3-way ReCom moves.\footnote{It is known that there are no locked configurations, as noted by \cite{LockedPolyominoTilings}. But this does not necessarily imply full connectivity of 3-way ReCom.} In the example from Figure~\ref{figTriominoProgram}, there is a single BUD move through the lower part of the tiling making c covered by a horizontal cell, but there are no valid 3-way ReCom moves. Thus, if 3-way ReCom moves are sufficient, we are only going to see it if we expand to consider partial tilings with $s > 13$ tiles.

\subsection{Negative results: locked configurations for BUD}\label{sub:negative-results}

\begin{figure}[h]
    \centering
    \begin{subfigure}[b]{0.25\textwidth}
        \centering        \includegraphics[scale=.3,page=2]{Media/locked.pdf}
        \caption{}
        \label{fig:locked-simple}
    \end{subfigure}
    ~
    \begin{subfigure}[b]{0.7\textwidth}
        \centering        \includegraphics[scale=.3,page=1]{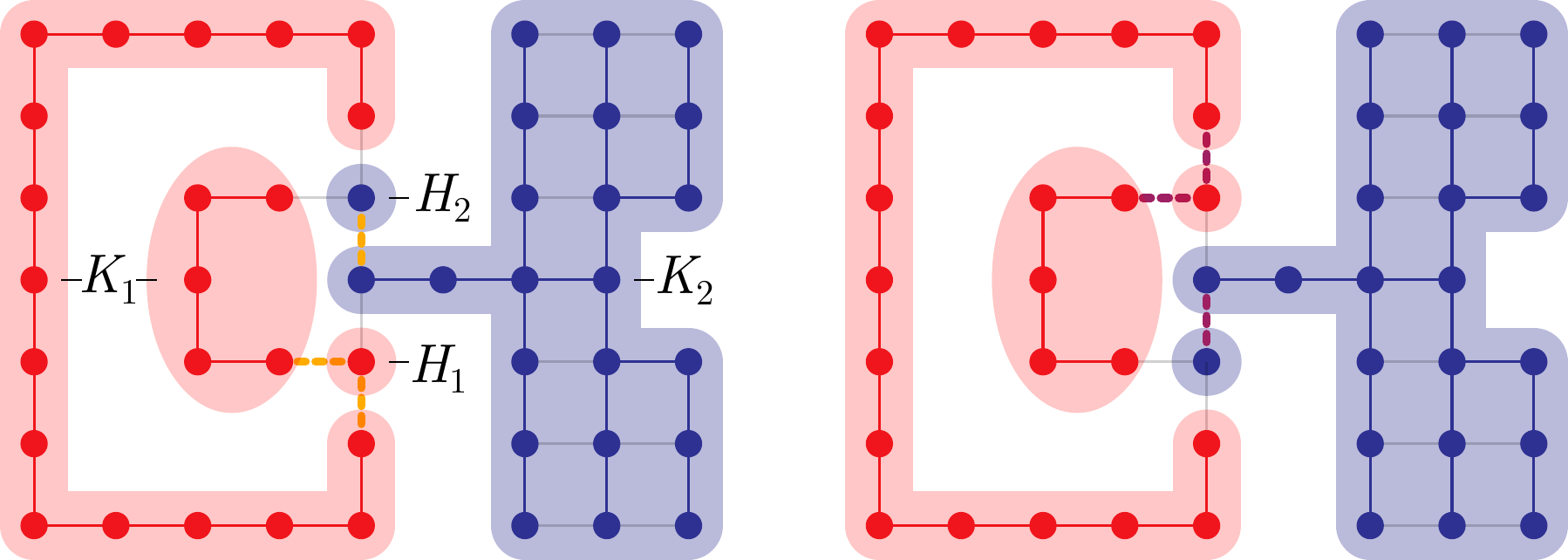}
        \caption{}        
        \label{fig:locked-grid}
    \end{subfigure}
    \caption{Examples in which the BUD walk is not irreducible.}
    \label{fig:locked}
\end{figure}

\begin{lemma}
    \label{lem:locked}
    For any $\varepsilon\ge 0$, there are graphs $G$ such that the BUD walk is not irreducible on $B^\varepsilon_2(G)$, even when $G$ is an induced grid graph.
\end{lemma}

\begin{proof}
    Figure~\ref{fig:locked-simple} shows an example of such a graph $G$ for $\varepsilon =0$.
    $B^0_2(G)$ contains only two partitions $\mathcal{P}=(P_1,P_2)$ and $\tilde{\mathcal{P}}=(\tilde{P}_1,\tilde{P}_2)$ shown in the left and right of Figure~\ref{fig:locked-simple} respectively.
    Let $P_1$ and $\tilde{P}_1$ be the districts shown in red and  $P_2$ and $\tilde{P}_2$ be the districts shown in blue.
    $\tilde{\mathcal{P}}$ is obtained from $\mathcal{P}$ by swapping the subgraphs $H_1$ and $H_2$ shown in the figure.
    Note that $K_1 =P_1\setminus H_1$ is not connected and thus the two partitions are not connected by a BUD step by Lemma~\ref{lemPartitionCharacterization}.
    The argument holds for the graph in Figure~\ref{fig:locked-grid}, which is an induced grid graph.
    
    We can generalize this counterexample for any $\varepsilon >0$ 
    by increasing the size of the two components of $K_1$ while maintaining $K_1$ two induced paths $\tau_1$ and $\tau_2$.
    Let $v_1$ and $v_2$ be the vertices in $H_1$ and $H_2$.
    For simplicity, we refer to $P_2$ as the district in a partition that contains the common neighbor of $v_1$ and $v_2$, and $P_1$ as the other district.
    In a balanced partition, $P_2$ must now include some vertices of $K_1$.
    As long as this number is greater than $\varepsilon$, and smaller than the sizes of $\tau_1$ and $\tau_2$, in every partition $P_2$ must include exactly one vertex in $\{v_1,v_2\}$.
    Thus, in any move between a partition where $P_2$ includes $v_1$ and a partition where $P_2$ includes $v_2$, deleting the subgraph of $P_1$ that is swapped in the move disconnects $P_1$.
    Thus, no such BUD step exists and $B^\varepsilon_2(G)$ has two components under the BUD walk.
\end{proof}

\section{Computing approximately balanced partitions of trees}\label{sec:complexity}

Implementing each step of the BUD walk requires checking whether a given tree is splittable into balanced subtrees. If exact balance is desired, it is a simple exercise to prove that, if such a partition exists, it is unique (Lemma~\ref{L:unique_cuts}). Moreover, finding such a partition (or determining that none exists) can be accomplished in polynomial time by a simple recursive algorithm that roots the tree at an arbitrary node. However, if parts are only required to be balanced to within some small error term (as is the case in political redistricting), it is not immediately clear how to extend this approach.

Formally, we define the \emph{Tree Approximate Partition Problem (TAPP)} as follows. The input is a vertex-weighted tree $T$, a positive integer $k$, and a value $\varepsilon$. We assume that vertex weights, $k$, and $\varepsilon$ are encoded in binary as rational numbers, and that the weights sum to one (this can easily be accomplished by replacing each weight $p(v)$ with $p(v)/p(G)$, for instance). The objective is to determine whether there exists an $\varepsilon$-balanced $k$-partition of $T$; recall this means that the total weight of each district must be in the range $[\frac1k - \frac{\varepsilon}{2}, \frac1k + \frac{\varepsilon}{2}]$.

If weights are (scaled-down versions of) polynomially-bounded integers, there is a straightforward dynamic programming solution to this problem. However, this is very inefficient in practice, for population weights can number in the millions. Thus, we seek to find a truly polynomial-time algorithm, rather than pseudo-polynomial-time.

Ito, Uno, Zhou, and Nishizeki \cite{Ito2008Partitioning} consider a variant of this problem where, instead of a single parameter $\varepsilon$, there are arbitrary lower and upper bounds on the size of each part, $l$ and $u$. In Section~\ref{subApproximateBalanceDp} we describe an essentially equivalent algorithm in the context where the population bounds are specified by $k$ and $\varepsilon$ as described above. Indeed, their setting can be thought of as a generalization of ours in which the sum of populations is not necessarily 1, but our runtime analysis does not depend on this assumption anyway. Ito et al.\txt{}~\cite{Ito2008Partitioning} give a runtime bound of $O(k^4 n)$. In Section~\ref{subRuntimeAnalysis} we give a tighter analysis, show that TAPP is decidable in time $O(k^3 n)$. In the special case where $\varepsilon = O(\frac{1}{k^2})$, which roughly corresponds to the practical regime where the population tolerance is stringent enough that it's impossible to construct an $\varepsilon$-balanced $(k - 1)$-partition or $(k + 1)$-partition, we obtain a runtime of $O(kn)$. Finally, in Section~\ref{sub:hardness}, we show that despite the tractability of the decision problem, the associated counting problem is \#P-complete. This means there is no obvious uniform sampling algorithm to output a random partition from a valid tree, an obstacle that we overcome using a heuristic approach in Section~\ref{sec:practical}.

\subsection{Efficiently deciding existence}\label{subApproximateBalanceDp}

We denote the set of all subtrees of $T$ by $h(T)$. For any $T' \in h(T)$, we denote the total weight of $T'$ by $\abs{T'}$. Thus $\abs{T} = 1$ by assumption. We arbitrarily choose a root to orient $T$. For any vertex $v \in T$, we denote by $T(v) \in h(T)$ the subtree of $T$ rooted at $v$, including $v$ itself. We denote a partition of any subtree $T' \in h(T)$ as a map $P: T' \to h(T')$ assigning each vertex to its subtree.

For any set $S \subseteq \rr$ and $\varepsilon > 0$, we define the \emph{$\varepsilon$-closure} of $S$ to be the set $S$ along with all intervals between consecutive values of $S$ that are at most $\varepsilon$ apart from each other,
$$C_\varepsilon(S) := \{y \in \rr \suchthat \txt{there exist $x, z \in S$ such that } x \leq y \leq z \leq x + \varepsilon\}.$$
If $C_\varepsilon(S) = S$, we say that $S$ is \emph{$\varepsilon$-closed}. It is straightforward to see that the $\varepsilon$-closure of any set is $\varepsilon$-closed.  

We'll use ``$+$'' to denote the Minkowski sum of sets $A,B\subset\rr$; that is,  $A+B = \{a+b\mid a\in A, b\in B\}$.  If $A_1,...,A_d\subset\R$, then $C_\varepsilon(A_1)+\cdots+C_\varepsilon(A_d)$ is not necessarily $\varepsilon$-closed;  However, the following is true:
\begin{lemma}\label{L:pickleball}
     If $A_1,...,A_d\subset\R$, then
     $ C_\varepsilon\left(C_\varepsilon(A_1)+\cdots+C_\varepsilon(A_d)\right)
     =C_\varepsilon(A_1+\cdots+A_d)$.
\end{lemma}

\begin{proof}
The inclusion ``$\supseteq$'' is obvious because $C_\varepsilon(A_1)+\cdots+C_\varepsilon(A_d)\supseteq A_1+\cdots+A_d$.  For the other inclusion, it will suffice to prove that $$C_\varepsilon(A_1)+\cdots+C_\varepsilon(A_d)\subseteq C_\varepsilon(A_1+\cdots+A_d);$$  
this will suffice because, if $C_\varepsilon(A_1)+\cdots+C_\varepsilon(A_d)$ is contained in an $\varepsilon$-closed set, then its $\varepsilon$-closure must also be contained in that set.

For this, let $y\in C_\varepsilon(A_1)+\cdots+C_\varepsilon(A_d)$, which means we can write $y=y_1+\cdots+y_d$ such that for each $i\in[d]$ we have $a_i^-\leq y_i\leq a_i^+$ for some $a_i^-,a_i^+\in A_i$ with $a_i^+-a_i^-\leq\varepsilon$.  This means that the set 
$$\{p_1+\cdots +p_d\mid p_i\in\{a_i^-,a_i^+\}\text{ for all }i\in[d]\}$$ is an $\varepsilon$-net of values in $A_1+\cdots+A_d$ on the interval $[a_1^-+\cdots+ a_d^-,a_1^++\cdots+a_d^+]$ in which $y$ is contained, so $y\in C_\varepsilon(A_1+\cdots+A_d)$.
\end{proof}

For a fixed upper limit $M>0$, consider the ``cropping'' operator $\Psi$ defined as \begin{equation}\label{E:define_crop}\Psi(A)=A\cap(-\infty,M].\end{equation}
In general $C_\varepsilon\circ\Psi\neq \Psi\circ C_\varepsilon$; to see why, notice for example that $(C_\varepsilon\circ\Psi)(A)$ only contains $M$ when $M\in A$, whereas $(\Psi\circ C_\varepsilon)(A)$ contains $M$ more generally whenever $A$ contains a pair of $\varepsilon$-close points above and below $M$.  However, the two sets only differ near $M$.  To express this more precisely, we define the following.

\begin{definition}\label{D:sim}
If $S_1,S_2\subset\R$ are closed sets, define $S_1\sim S_2$ to mean that 
\begin{enumerate}
    \item $S_1\cap(-\infty,M-\varepsilon]=S_2\cap(-\infty,M-\varepsilon]$, and 
    \item $\inf(S_1\cap[M-\varepsilon,\infty)) = \inf(S_2\cap[M-\varepsilon,\infty))$ (or both of these sets are empty).
\end{enumerate}

\end{definition}

\begin{lemma}\label{L:commute}
   If $A\subset\R$ is a closed set, then $(\Psi\circ C_\varepsilon)(A)\sim (C_\varepsilon\circ\Psi)(A)$. 
\end{lemma}

\begin{proof}
Let $x = \sup\{a\in A\mid a\leq M\}$ and let $y=\inf\{a\in A \mid a\geq M\}$.
If $y$ is undefined (because there are no such points) or $y-x>\varepsilon$, then $(\Psi\circ C_\varepsilon)(A)= (C_\varepsilon\circ\Psi)(A)$.  On the other hand, if $y-x\leq\varepsilon$, then $(\Psi\circ C_\varepsilon)(A)= (C_\varepsilon\circ\Psi)(A)\cup(x,M]$.  Note that $y>M$ (because otherwise we'd have $x=y=M$), so $M-x<y-x=\varepsilon$.  In all cases, $(\Psi\circ C_\varepsilon)(A)\sim (C_\varepsilon\circ\Psi)(A)$.
\end{proof}
We leave to the reader the proof of the following Lemma, which requires a straightforward argument similar to the last proof. 
\begin{lemma}\label{L:smalltool}
If $A,B\subset\R$ are closed sets with $A\sim B$, then $C_\varepsilon(A)\sim C_\varepsilon(B)$, and $\Psi(A)\sim\Psi(B)$, and $(\Psi\circ C_\varepsilon)(A)\sim(\Psi\circ C_\varepsilon)(B)$.
\end{lemma}

\begin{lemma}\label{L:sums}
If $A_1,...,A_d, B_1,...,B_d\subset\R$ are closed sets, and $A_i\sim B_i$ for every $i\in[d]$, then
$$A_1+\cdots+A_d\sim B_1+\cdots+ B_d.$$
\end{lemma}

\begin{proof} Set $A=A_1+\cdots+A_d$ and $B=B_1+\cdots+B_d$.  It is clear that $A\cap(-\infty,M-\varepsilon] = B\cap(-\infty,M-\varepsilon]$ because only terms of the individual sets $A_i$ and $B_i$ that are $\leq M-\varepsilon$ can contribute to sums that are $\leq M-\varepsilon$.  

Next, suppose that $A$ intersects $[M-\varepsilon,\infty),$ and set $y=\inf(A\cap[M-\varepsilon,\infty))$.  We can write $y=y_1+\cdots+y_d$, where $y_i\in A_i$ for all $i\in[d]$.  For each $i\in[d]$, if $y_i\leq M-\varepsilon$, then $y_i\in B_i$.   On the other hand, if $y_i\geq M-\varepsilon$, then the assumption that $y$ in the infimum implies that $y_i=\inf(A_i\cap[M-\varepsilon,\infty))=\inf(B_i\cap[M-\varepsilon,\infty))$, so $y_i\in B_i$.  It follows that $y\in B$.  In summary, if $A$ intersects $[M-\varepsilon,\infty)$, then so does $B$ and  $\inf(A\cap[M-\varepsilon,\infty))\in B$.  The roles of $A$ and $B$ are symmetric, so the reverse is also true, which completes the proof.  
\end{proof}

Lemma~\ref{L:pickleball} can be generalized using Lemmas~\ref{L:commute},~\ref{L:smalltool} and~\ref{L:sums}.

\begin{lemma}\label{L:Gracie}
If $A_1,...,A_d, \tilde{A}_1,...,\tilde{A}_d\subset\R$ are closed sets, and $\tilde{A}_i\sim C_\varepsilon(A_i)$ for every $i\in[d]$, then
$$(C_\varepsilon\circ\Psi)\left(\tilde{A}_1+\cdots+\tilde{A}_d\right)\sim (C_\varepsilon\circ\Psi)\left(A_1+\cdots+ A_d\right).$$
\end{lemma}

\begin{proof}
\begin{align*}
 (C_\varepsilon\circ\Psi)(\tilde{A}_1+\cdots+\tilde{A}_d)
 & \sim  (\Psi\circ C_\varepsilon)(\tilde{A}_1+\cdots+\tilde{A}_d) \\
 & \sim (\Psi\circ C_\varepsilon)(C_\varepsilon(A_1)+\cdots+C_\varepsilon(A_d)) \\
 & =  (\Psi\circ C_\varepsilon)\left(A_1+\cdots+ A_d\right)
 \sim (C_\varepsilon\circ\Psi)\left(A_1+\cdots+ A_d\right).
\end{align*}
\end{proof}

\begin{lemma}\label{L:Boston}
If $A_1,...,A_d, B_1,...,B_d\subset\R$ are closed sets, and $(\Psi\circ C_\varepsilon)(A_i)\sim (\Psi\circ C_\varepsilon)(B_i)$ for every $i\in[d]$, then
\begin{equation}\label{E:Boston}
(\Psi\circ C_\varepsilon)\left(A_1\cup\cdots\cup A_d\right)
\sim (\Psi\circ C_\varepsilon)\left(B_1\cup\cdots\cup B_d\right).
\end{equation}
\end{lemma}

\begin{proof}
Denote the left and right sides of Equation~\ref{E:Boston} as $L$ and $R$ respectively.  We will first prove that $$L\cap(-\infty,M-\varepsilon]= R\cap(-\infty,M-\varepsilon].$$
We will only prove the inclusion ``$\subseteq$'', since the other is identical with the roles of the $A$s and $B$s swapped.  For this, let   $y\in L\cap(-\infty,M-\varepsilon]$.  This means that $x\leq y\leq z$ for some $x,z\in A_1\cup\cdots\cup A_d$ with $z-x\leq\varepsilon$.  Thus, $x\in A_i$ and $z\in A_j$ for some $i,j\in[d]$.   Since $x\leq y\leq M-\varepsilon$ and $x\in A_i\subseteq C_\varepsilon(A_i)$, we know that $x\in C_\varepsilon(B_i)$, so $x^-\leq x\leq x^+$ for some $x^-,x^+\in B_i$ with $x^+-x^-\leq\varepsilon$.   

CASE 1: If $z\leq M-\varepsilon$, then since $z\in A_j\subseteq C_\varepsilon(A_j)$, we have $z\in C_\varepsilon(B_j)$, so $z^-\leq z\leq z^+$ for some $z^-,z^+\in B_j$ with $z^+-z^-\leq\varepsilon$.  In this case, $\{x^-,x^+,z^-,z^+\}$ is an $\varepsilon$-net of points of $B_i\cup B_j$ that surrounds $y$, so $y\in R$.

CASE 2: If $z>M-\varepsilon$, then $z\geq w$, where  $w=\inf(C_\varepsilon(A_j)\cap[M-\varepsilon,\infty))=\inf(C_\varepsilon(B_j)\cap[M-\varepsilon,\infty))$.  As before, we have that $w^-\leq w\leq w^+$ for some $w^-,w^+\in B_j$ with $w^+-w^-\leq\varepsilon$.  In this case, $\{x^-,x^+,w^-,w^+\}$ is an $\varepsilon$-net of points of $B_i\cup B_j$ that surrounds $y$, so $y\in R$.

It remains to prove that $$\inf(L\cap[M-\varepsilon,\infty)) = \inf(R\cap[M-\varepsilon,\infty))$$
or that both sets are empty.  For this, assume that $L$ intersects $[M-\varepsilon,\infty)$ and let $y=\inf(L\cap[M-\varepsilon,\infty))$.  Since $y\in L$, we have that $x\leq y\leq z$ for some $x,z\in A_1\cup\cdots\cup A_d$ with $z-x\leq\varepsilon$.  Thus, $x\in A_i$ and $z\in A_j$ for some $i,j\in[d]$.  We know that $x<M-\varepsilon$ (because otherwise it would contradict the assumption that $y$ is the infimum), so $x\in C_\epsilon(B_i)$, which means that $x^-\leq x\leq x^+$ for some $x^-,x^+\in B_i$ with $x^+-x^-\leq\varepsilon$.   Set $w=\inf(C_\varepsilon(A_j)\cap[M-\varepsilon,\infty))=\inf(C_\varepsilon(B_j)\cap[M-\varepsilon,\infty))$, so $w\leq z$.  We know that $y\leq w$ because otherwise we would contradict the assumption that $y$ is the infimum.  In summary, $y\leq w\leq z$.   As before, we have that $w^-\leq w\leq w^+$ for some $w^-,w^+\in B_j$ with $w^+-w^-\leq\varepsilon$.  Therefore, $\{x^-,x^+,w^-,w^+\}$ is an $\varepsilon$-net of points of $B_i\cup B_j$ that surrounds $y$, so $y\in R$.  

In summary, we proved that if $L$ intersects $[M-\varepsilon,\infty)$, then $\inf(L\cap[M-\varepsilon,\infty))\in R$. 
The vice-versa statement (with $L$ and $R$ exchanged) is proven symmetrically, which completes the proof.

\end{proof}

We next set up our application of $\varepsilon$-closures to the TAPP problem.  For this, we call a set of vertices \emph{valid} if its total weight lies in $[1/k - \varepsilon/2, 1/k + \varepsilon/2]$.  For any vertex $v$ and any $l\in\{0,...,k\}$, we consider all ways to partition $T(v)$ into $l$ valid pieces, plus a possible ``surplus'' piece; that is, the piece containing $v$ is allowed to be one of the $l$ valid districts (surplus $=0$), or to be underpopulated, in which case we think of this under-population as a ``surplus'' to which more upstream nodes will be added to form a valid piece.  We define $S_l(v)$ as the set of possible such surpluses, and $f_l(v)$ as its $\varepsilon$-closure.
\begin{align*}
    S_l(v) & :=  X\cup  \{x\in[0,1/k + \varepsilon/2] \suchthat  \txt{there is a partition $P$ of $T(v)$ into $l+1$ pieces}\\
    & \hspace{.6in}\txt{where $\abs{P(v)} = x$ and the other $l$ pieces are valid}\},\\
    & \text{ where } X = \begin{cases} \{0\} & \text{if there is a partition of $T(v)$ into $l$ valid pieces,} \\
    \emptyset &\text{otherwise.}\end{cases} \\
    f_l(v) & := C_\varepsilon(S_l(v)).
\end{align*}

The advantage of using $\varepsilon$-closures is that, even though $|S_l(v)|$ might be exponential in the number of vertices, the following is true:
\begin{lemma}\label{L:few_components}
For every vertex $v$ and every $l\in\{1,...,k\}$, $f_l(v)$ has at most $l$ connected components.     
\end{lemma}
\begin{proof}
The total weight of the union of $l$ valid districts must lie in the interval
$\left[l\cdot(1/k-\epsilon/2), l\cdot(1/k+\epsilon/2)\right],$ from which it follows that
$$ S_l(v)\subseteq \left[|T(v)|-l\cdot(1/k+\epsilon/2), |T(v)|-l\cdot(1/k-\epsilon/2)\right],$$
which is an interval of length $l\cdot \varepsilon$.  It is straightforward to show that the $\varepsilon$-closure of a subset of an interval of length $l\cdot\varepsilon$ has at most $l$ components. 
\end{proof}
Regarding the $l=0$ case, observe that $f_0(v)$ has at most $1$ connected component because $S_0(v)$ is empty or equals the singleton set $\{|T(v)|\}$.  

As noted, the following theorem is implied by Ito et al.\txt{}~\cite{Ito2008Partitioning}. In Section~\ref{subRuntimeAnalysis} we will give a tighter analysis of the algorithm's runtime, but for now, we will describe the algorithm and simply remark that it is polynomial-time.
\begin{theorem}\label{thmApproximateBalanceAlgorithm}
    There is a polynomial-time algorithm to solve TAPP.
\end{theorem}
\begin{proof}
Set $M= 1/k+\varepsilon/2$, define ``$\Psi$'' as the operator that crops at $M$ (Equation~\ref{E:define_crop}), and define ``$\sim$'' as in Definition~\ref{D:sim}.  Define $S=[1/k - \varepsilon/2, 1/k+ \varepsilon/2]=[M-\varepsilon,M]$, which is the range of populations for a valid district.

Working from the leaves towards the root, we recursively approximate $f_l(v)$ for each vertex $v$ and each $l\in\{0,...,k-1\}$ as the value $\hat{f}_l(v)$ defined as follows.  If $v$ is a leaf with weight $p$, then define $\hat{f}_0(v)=f_0(v)\in
    \{\{p\},\emptyset\}$ depending on whether $p\leq M$. Similarly, define $\hat{f}_1(v)=f_1(v)\in\{\{0\},\emptyset\}$ depending on whether $p\in S$, and define $\hat{f}_l(v)=f_l(v)=\emptyset$ for all $l>1$.  

Next consider a fixed vertex $v$ that is not a leaf, let $p$ denote its weight, and let $\{u_1,...,u_d\}$ denote its children.    We define $\hat{f}_l(v)$  for each $l\in\{0,...,k-1\}$ via the following recursive formula:
\begin{equation}\label{E:tofu}
\hat{f}_l(v)=(C_\varepsilon\circ\Psi)\left(\hat{X}\cup \left\{ p+\sum x_i \,\big|\, x_i\in \hat{f}_{l_i}(u_i)\text{ and } \sum l_i = l\right\}\right),
\end{equation}
where $\hat{X}=\emptyset$ or $\hat{X}=\{0\}$, depending on whether 
$S$ intersects with the following set:
\begin{equation}\label{E:nap}
\hat{Z} = \left\{ p+\sum x_i \,\big|\, x_i\in \hat{f}_{l_i}(u_i)\text{ and } \sum l_i = l-1\right\}.
\end{equation}
Here, all sums are over $i\in\{1,...,d\}$.

Note that the exact value of $f_l(v)$ is obtained by replacing ``$\hat{f}_{l_i}(u_i)$'' with ``$S_{l_i}(u_i)$'' in Equations~\ref{E:tofu} and~\ref{E:nap}; that is,
\begin{equation}\label{E:faketofu}
f_l(v)=(C_\varepsilon\circ\Psi)\left(X\cup \left\{ p+\sum x_i \,\big|\, x_i\in S_{l_i}(u_i)\text{ and } \sum l_i = l\right\}\right),
\end{equation}
where $X=\emptyset$ or $X=\{0\}$, depending on whether the following set overlaps with 
$S$:
\begin{equation}
Z = \left\{ p+\sum x_i \,\big|\, x_i\in S_{l_i}(u_i)\text{ and } \sum l_i = l-1\right\}.
\end{equation}

Note that $\hat{f}_l(v)$ does not necessarily equal $f_l(v)$.  To understand why, notice for example that $f_l(v)$ would only contain $M$ in the unlikely event there are children surplus values that sum exactly to $M$, whereas $\hat{f}_l(v)$ conatins $M$ under more general circumstances.   Nevertheless, we claim that our approximation is accurate in the sense that for every vertex $v$ and every $l\in\{0,...,k-1\}$, 
$$\hat{f}_l(v)\sim f_l(v).$$ 
We will prove this claim by induction on the path distance from $v$ to the nearest leaf.  The formula is true in the base case because if $v$ is a leaf, then $\hat{f}_l(v)=f_l(v)$ for every $l$.  Next, take a non-leaf vertex $v$ and assume that the formula is true for all children of $v$.  Let $l\in\{0,...,k-1\}$.  We must prove that $\hat{f}_l(v)\sim f_l(v)$.  

We begin by showing $0\in\hat{f}_l(v)$ if and only if $0\in f_l(v)$; in other words, $\hat{X}=X$.  For this, consider a single ``index vector'' $I=(l_1,...,l_d)$ with $\sum l_i=l-1$, and define $\hat{Z}_I$ (respectively $Z_I$) as the corresponding subset of $\hat{Z}$ (respectively $Z$); that is,
\begin{equation}\label{E:defz}
\hat{Z}_I = \left\{ p+\sum x_i \,\big|\, x_i\in \hat{f}_{I(i)}(u_i) \right\},\,\,
Z_I = \left\{ p+\sum x_i \,\big|\, x_i\in S_{I(i)}(u_i) \right\}.
\end{equation}
Thus $\hat{Z}=\cup \hat{Z}_I$ and $Z=\cup Z_I$, where the unions are over all index vectors $I$ that sum to $l-1$.  We can determine whether $Z$ (respectively $\hat{Z}$) intersects with $S$ by checking separately for each such $I$ whether $Z_I$ (respectively $\hat{Z}_I$) intersects $S$.  Note that $Z_I$ (respectively $\hat{Z}_I$)  with $S$ if and only if $(C_\varepsilon\circ\Psi)(Z_I)$ (respectively $(C_\varepsilon\circ\Psi)(\hat{Z}_I)$) intersects $S$.  
By Lemma~\ref{L:Gracie}, $(C_\varepsilon\circ\Psi)(Z_I)\sim (C_\varepsilon\circ\Psi)(\hat{Z}_I)$, so in particular,  $(C_\varepsilon\circ\Psi)(Z_I)$ intersects with $S$ if and only if $(C_\varepsilon\circ\Psi)(\hat{Z}_I)$ intersects with $S$.  In summary, $\hat{X}=X$, so $0\in\hat{f}_l(v)$ if and only if $0\in f_l(v)$.

Next, to compare Equations~\ref{E:tofu} and~\ref{E:faketofu}, let $\mathcal{I}$ denote the set of index vectors $I=(l_1,...,l_d)$ with $\sum l_i=l$ (rather than summing to $l-1$, as before).  For $I\in \mathcal{I}$, define $\hat{Z}_I, Z_I$ as in Equation~\ref{E:defz} so that 
$$\hat{f}_l(v) = (C_\varepsilon\circ\Psi)(\hat{X}\cup(\cup_{I\in\mathcal{I}}\hat{Z}_I )), \text{ and }f_l(v) = (C_\varepsilon\circ\Psi)\left(X\cup\left(\cup_{I\in\mathcal{I}}Z_I \right)\right).$$ 
According to Lemma~\ref{L:Gracie}, $(C_\varepsilon\circ \Psi)(\hat{Z}_I)\sim (C_\epsilon\circ\Psi)(Z_I)$ for each $I\in\mathcal{I}$.  According to Lemma~\ref{L:Boston}, this in turn implies that  $\hat{f}_l(v)\sim f_l(v)$.

When we reach the root $v_0$, we need only compute $\hat{f}_{k}(v_0)$ via the same recursive formula and observe that there exists a valid partition of  $T$ if and only if $0\in \hat{f}_{k}(v_0)$.

To show that the algorithm has polynomial time, we must we must consider the time required to compute the Minkowski sums in Equations~\ref{E:tofu} and~\ref{E:nap}.  Both equations require us to compute an expression of the form
\begin{equation}
F_{j_0} = (C_\varepsilon\circ\Psi)\left(\left\{ p+\sum x_i \,\big|\, x_i\in \hat{f}_{l_i}(u_i)\text{ and } \sum l_i = j_0\right\}\right),
\end{equation}
where $j_0=l$ for Equation~\ref{E:tofu}, or $j_0=l-1$ for Equation~\ref{E:nap}.  According to Lemma~\ref{L:few_components} , each set $\hat{f}_{l_i}(u_i)$ has at most $l_i\leq k$ connected components, and it is clear from construction that these components are disjoint intervals.   In general, if $A_1,\dots, A_n$ each equals a union of disjoint intervals, then the Minkowski sum $A_1+\cdots+A_n$ equals the union of all intervals of the form $[x_1+\cdots+x_n,y_1+\cdots y_n]$ where $[x_i,y_i]$ is one of the intervals of $A_i$ for each $i\in [n]$.  The graphs of interest for applications typically have bounded degree, so we expect this process to be the fastest way to compute $F_{j_0}$ in practice.  But in theory, the number of steps of this computation is exponential in the degree, so to complete the proof, we must describe an alternative method of doing the computation one child vertex at a time, with a polynomial number of steps required for each child vertex.  

For this, for each $j\in\{0,...,j_0\}$, we define $F^i_j$ to be the computation of $F_j$ that only factors in the child vertices $u_1,...,u_i$, and we iteratively obtain $\{F_j^{i}\mid j\in\{1,...,j_0\}\}$ from $\{F_j^{i-1}\mid j\in\{1,...,j_0\}\}$.  For this, first initialize $F^0_0=\{p\}$ and $F^0_j=\emptyset$ for each $j\in\{1,...,j_0\}$.  Next, iteratively compute $F^{i}$ from $F^{i-1}$ by factoring in the $i^\text{th}$ child vertex as follows.  Start by setting $F_j^i=\emptyset$ for each $j\in\{0,...,j_0\}$.  Then for each $a,b\in\{1,\dots,j_0\}$ with $a+b\leq j_0$, we set 
\begin{equation}\label{E:george}
F^{i}_{a+b}=(C_\varepsilon\circ\Psi)(F^{i}_{a+b}\cup(F^{i-1}_a+\hat{f}_{b}(u_i))).
\end{equation}
Note that ``$F^i_{a+b}$'' on the right side of Equation~\ref{E:george} might be non-empty due to previously processed values $a',b'$ with $a'+b'=a+b$.  According to the lemmas above, the extra occurrences of the $(C_\varepsilon\circ\Psi)$ at intermediate steps of this algorithm do not change the final result.

\end{proof}
The algorithm in the previous proof can be simplified if we add the hypothesis that 
\begin{equation}\varepsilon < \frac{2}{k(k + 1)},\label{E:smallep}\end{equation} which ensures that it's not possible to partition the graph into more or less than $k$ valid districts. (As we will see in Section~\ref{subRuntimeAnalysis}, this also allows implies a faster running time.) This hypothesis allows us to simply focus on the task of making districts of the correct sizes without worrying about how many districts we have made so far.  In particular, we don't need to separately track $\hat{f}_l(v)$ for each $l\in[k]$, but we can instead just track their union.  

That is, we need only define
$$S(v) = \cup S_l(v) \text{ and }f(v)=C_\varepsilon(S(v)).$$
We can approximate $f(v)$ as the value $\hat{f}(v)$ defined recursively as
$$\hat{Z}=\left\{p+\sum x_i\mid x_i\in\hat{f}(u_i)\right\},\text{ and }
\hat{f}(v) = (C_\varepsilon\circ\Psi)(\hat{X}\cup\hat{Z}),
$$
where $\hat{X}=\emptyset$ or $\hat{X}=\{0\}$, depending on whether $\hat{Z}$ intersects with $S$.  This simplified structure obviates the need for Lemma~\ref{L:Boston}.

\subsection{Runtime analysis}\label{subRuntimeAnalysis}

The algorithm presented in Section~\ref{subApproximateBalanceDp} performs a sequence of Minkowski sums over unions of intervals. We can obtain a straightforward bound of $O(k^4 n)$ by noting that there are $n - 1$ merges, each involving $k$ choices for the number $l$ of completed parts below each of the two vertices being merged, and at most $k$ components per choice of $l$ by Lemma~\ref{L:few_components}. This is the bound obtained by~\cite{Ito2008Partitioning}. However, this analysis ignores the fact that costly merges can only occur near the root of the tree, where the number of completed parts can be high. Leveraging this observation, we are able to obtain improved runtime bounds.

\begin{theorem}\label{T:deadline}
    Our algorithm for TAPP runs in time $O(k^3 n)$ in general, and $O(kn)$ in the case where $\varepsilon = O(1/k^2)$.
\end{theorem}

The proof will require the following modification of Lemma~\ref{L:few_components}, in which we think of $|T(v)|\cdot k$ as representing the number of ``districts worth of weight'' in $T(v)$. 
\begin{lemma}\label{L:few_components_2}
For every vertex $v$ and every $l\in\{1,...,k\}$, $f_l(v)$ has at most $1+|T(v)|\cdot k$ connected components.     
\end{lemma}

\begin{proof}
Set $w=|T(v)|$ and set $p=w\cdot k$.  We can think of $p$ as the number of districts worth of weight in $T(v)$.  If $p\geq l$, then the result follows immediately from Lemma~\ref{L:few_components}, so we will assume that $p< l$; that is, $w< \frac lk$.
As in the proof of Lemma~\ref{L:few_components}, the total weight of the union of $l$ valid districts must lie in the interval
$$I=\left[\max\{0, l\cdot(1/k-\varepsilon/2)\}, l\cdot(1/k+\varepsilon/2)\right].$$  By hypothesis, $w$ is less than $\frac l k \in I$.  If $w\leq\min(I)$, then the surplus set $S_l(v)$ is empty; otherwise $S_l(v)$ lies in 
$I'=[0,w -\min(I)]$, which is an interval of length
$$\text{length}(I') = w-\min(I) = w-\max\{0,l\cdot(1/k-\varepsilon/2)\}.$$
It will suffice to prove that $\text{length}(I')\leq p\cdot\varepsilon$, because the $\varepsilon$-closure of a subset of an interval of length $p\cdot \varepsilon$ can have at most $\lceil p \rceil \leq 1+p$ components, where $\lceil p\rceil$ means $p$ rounded up to the nearest integer. 
\begin{itemize}
    \item CASE 1: Assume $\varepsilon\leq \frac 1k$.  In this case, we have
    $$\text{length}(I') = \frac pk - \frac lk + \frac{l\varepsilon}{2} 
    = \frac{l\varepsilon}{2}-\frac{l-p}{k} 
    \leq  \frac{l\varepsilon}{2}-\varepsilon(l-p)
    = \varepsilon\cdot\left(p-\frac l2\right)\leq p\varepsilon,$$
    as desired.
    
    \item CASE 2: Assume $\varepsilon\geq\frac 1k$.  In this case, we have $\text{length}(I')\leq w=\frac pk \leq p\varepsilon$, as desired.
\end{itemize}

\end{proof}

\begin{proof}[Proof of Theorem~\ref{T:deadline}]
    Let $B \geq 2$ be a bound such that, for any subtree $T(v)$, the number of distinct values $d$ for which $T(v)$ contains contain $d$ complete districts is $\leq B$.  For example, $B=k$ works in general, while $B=2$ works under the assumption of Equation~\ref{E:smallep}. 
    
    Consider a merge operation as shown in Figure~\ref{figMergeStep}, merging a child node $v_2$ with $n_2$ vertices of collective weight $p_2/k$ (e.g., ``$p_2$ districts worth of weight'') into a parent node $v_1$ with $n_1$ vertices of collective weight $p_1/k$. Lemma~\ref{L:few_components_2} implies that this operation takes time $O(B(p_1 + 1) \cdot B(p_2 + 1))$. (The $+1$ are necessary because $p_1$ and $p_2$ can be arbitrarily small.) The same bound obviously holds if we replace $p_1$ and $p_2$ with $n_1$ and $n_2$. Thus, letting $x_i := \min\{p_i, n_i\}$, we know that there is some constant $c_1 \geq 1$ such that the runtime of a merge operation is bounded by $c_1(B(x_1 + \tfrac12))(B(x_2 + \tfrac12))$.
    
    \ipncm{.6}{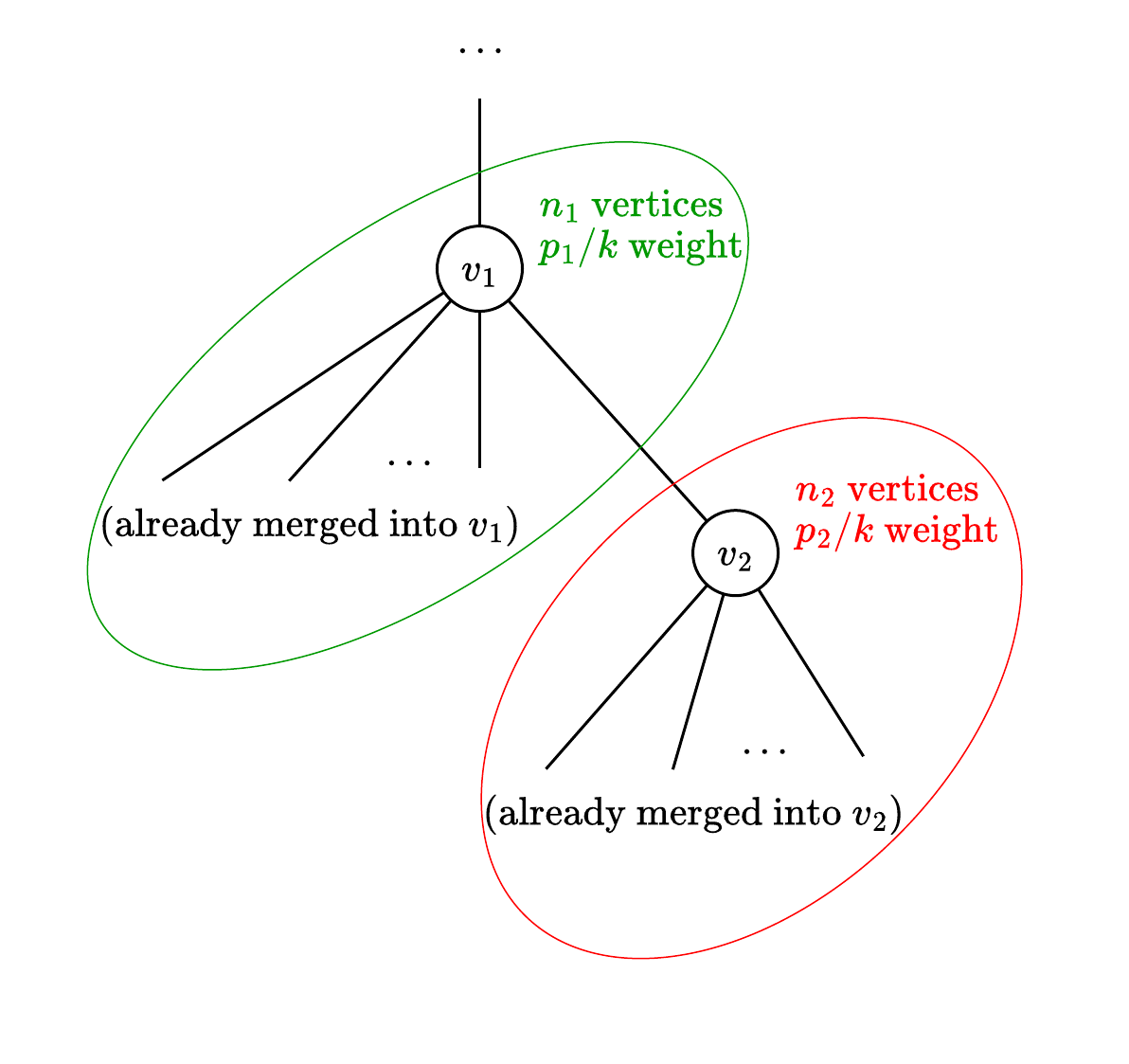}{\label{figMergeStep}An arbitrary merging step in the polynomial-time algorithm for TAPP. Note that the $n_1$ and $p_1$ bounds for $v_1$ do not include the subtree rooted at $v_2$; i.e., they count the subtree that has already been merged into $v_1$.}

    We claim that the total runtime spent within a recursive call to a subtree containing $n$ vertices and $p/k$ weight is at most $c_2(c_1B^2(n - \frac12)(p + 1) + 1)$, where $c_2$ is the time it takes to process a leaf. We proceed by induction on $n$; the base case ($n = 1$) follows from the definition of $c_2$. For the inductive case, consider the subtree depicted in Figure~\ref{figMergeStep}, rooted at $v_1$. Note that
    \begin{align}
        \frac12 x_1 x_2 + \frac12x_1 &= x_1 \left(\frac{x_2 + 1}{2}\right) \leq p_1\left(\frac{n_2 + 1}{2}\right) \leq p_1n_2\label{equMerge1},\\
        \frac12 x_1 x_2 + \frac12x_2 &= x_2 \left(\frac{x_1 + 1}{2}\right) \leq p_2\left(\frac{n_1 + 1}{2}\right) \leq p_2n_1\label{equMerge2}.
    \end{align}
    
    Applying the inductive hypothesis to the two smaller subtrees, we may bound the total running time by
    \begin{align*}
        T(v_1) &\leq c_2 \left( c_1B^2(n_1 - \tfrac{1}{2})(p_1 + 1) + 1 + c_1B^2(n_2 - \tfrac{1}{2})(p_2 + 1) + 1\right) + c_1(B(x_1 + \tfrac12))(B(x_2 + \tfrac12)) \\
        &\leq c_2 \left( c_1B^2(n_1 - \tfrac{1}{2})(p_1 + 1) + 1 + c_1B^2(n_2 - \tfrac{1}{2})(p_2 + 1) + 1 + c_1 B^2 (x_1 + \tfrac{1}{2})(x_2 + \tfrac{1}{2}) \right) \\
        &= c_2 \left( c_1 B^2 \left[ n_1 p_1 + n_1 - \tfrac{1}{2}p_1 - \tfrac{1}{2} + n_2 p_2 + n_2 - \tfrac{1}{2}p_2 - \tfrac{1}{2} + x_1 x_2 + \tfrac{1}{2}x_1 + \tfrac{1}{2}x_2 + \tfrac{1}{4} \right] + 2 \right) \\
        &\leq c_2 \left( c_1 B^2 \left[ n_1 p_1 + n_2 p_2 + n - 1 - \tfrac{1}{2}p + (p_1 n_2 + p_2 n_1) + \tfrac{1}{4} \right] + 2 \right) \stext{by \eqref{equMerge1} and \eqref{equMerge2}}\\
        &= c_2 \left( c_1 B^2 \left[(n_1 + n_2)(p_1 + p_2) + n - \tfrac{1}{2}p - \tfrac{3}{4} \right] + 2 \right) \\
        &= c_2 \left( c_1 B^2 \left[ np + n - \tfrac{1}{2}p - \tfrac{3}{4} \right] + 2 \right) \\
        &= c_2 \left( c_1 B^2 \left[ (n - \tfrac{1}{2})(p + 1) - \tfrac{1}{4} \right] + 2 \right) \\
        &= c_2 \left( c_1 B^2 (n - \tfrac{1}{2})(p + 1) - \tfrac{1}{4} c_1 B^2 + 2 \right) \\
        &\leq c_2 \left( c_1 B^2 (n - \tfrac{1}{2})(p + 1) + 1 \right)
    \end{align*}
    where in the final line, we have used the bounds $B \geq 2$ and $c_1 \geq 1$.
       
    By induction, the bound holds for all subtrees. Applying it to the root node, this becomes $O(B^2 k n)$. Since $B \leq k$, this is always $O(k^3 n)$. When $\varepsilon = O(\frac{1}{k^2})$, we have $B = O(1)$, so this becomes $O(kn)$.
\end{proof}

\subsection{Hardness of uniform sampling via conditional probabilities}\label{sub:hardness}

Theorem~\ref{thmApproximateBalanceAlgorithm} gives us a way to efficiently implement the BUD walk on the space of approximately splittable trees. However, neither the statement of the theorem nor its proof give any way to reconstruct the space of underlying partitions that may arise from splitting such trees. For a given tree, is it possible to obtain a uniform sample from the set of all possible $\eps$-balanced $k$-partitions? Our next result suggests that this problem may be considerably more challenging.

To determine a $k$-partition, we need to select a set of $k - 1$ split edges. The canonical way to sample a uniformly random set of valid edges is to sample one edge at a time, each with the correct conditional probability. To do this, we can pick an arbitrary edge $e$ and randomly pick one of two options:
\begin{enumerate}[label=(\alph*)]
    \item\label{itmCut} Make $e$ a cut edge, removing it from the tree.
    \item\label{itmContract} Determine that $e$ will not be a cut edge, contracting it.
\end{enumerate}
In either case, we recurse on the remaining subtree(s) until all trees consist of a single valid part. Thus, if we can choose between \ref{itmCut} and \ref{itmContract} with the correct probabilities, then we can get a uniform sample. The converse trivially holds as well. Unfortunately, computing the correct probabilities of \ref{itmCut} versus \ref{itmContract} turns out to be intractable.

\begin{theorem}\label{thmApproximateBalanceExactSamplingHardness}
    Given an instance $(T, k, \varepsilon)$ of TAPP and a specific edge $e \in E(T)$, it is \#P-complete to determine the probability that $e$ is split in a uniformly random $\varepsilon$-balanced $k$-partition of $T$.
\end{theorem}

\begin{proof}
    This problem is in \#P because we may compute the probability by taking ratio of the number of $\varepsilon$-balanced $(k - 1)$ partitions that split $e$ and the total number of $\varepsilon$-balanced $(k - 1)$ partitions. To show hardness, we reduce from the following 0/1 Knapsack counting problem. The input is a list of item sizes $w_1, w_2, \dots, w_n$ and a capacity $C$, all positive integers. The goal is to determine how many subsets of items have total weight at most $C$. We additionally assume that all weights are at most $\frac{C}{2}$; this problem is still \#P-complete.\footnote{It appears to be a folk theorem that the 0/1 Knapsack counting problem is \#P-complete; see for instance Morris and Sinclair~\cite{MorrisSinclair}. To see that the general problem reduces to the variant where no weight is more than half the capacity $C$, consider an arbitrary Knapsack instance with total weight $B \geq C$. If we add two more items with sizes $B$ and increase the capacity from $C$ to $2B + C$, then the number of solutions to the new instance is equal to $3 \cdot 2^n$ plus the number of solutions to the original instance.} Given an instance $w_1, w_2, \dots, w_n, C$, we let $k := n + 2$ and pick any sufficiently small $\varepsilon < \frac1k$. For each $i \in [n]$, define
    \[x_i := w_i \cdot \frac{\varepsilon}{C} \leq \frac{\varepsilon}{2}.\]
    Finally, we define the tree $T$ as in Figure~\ref{figSharpPReduction}, with the special edge $e$ denoted in bold.

    \ipncm{.6}{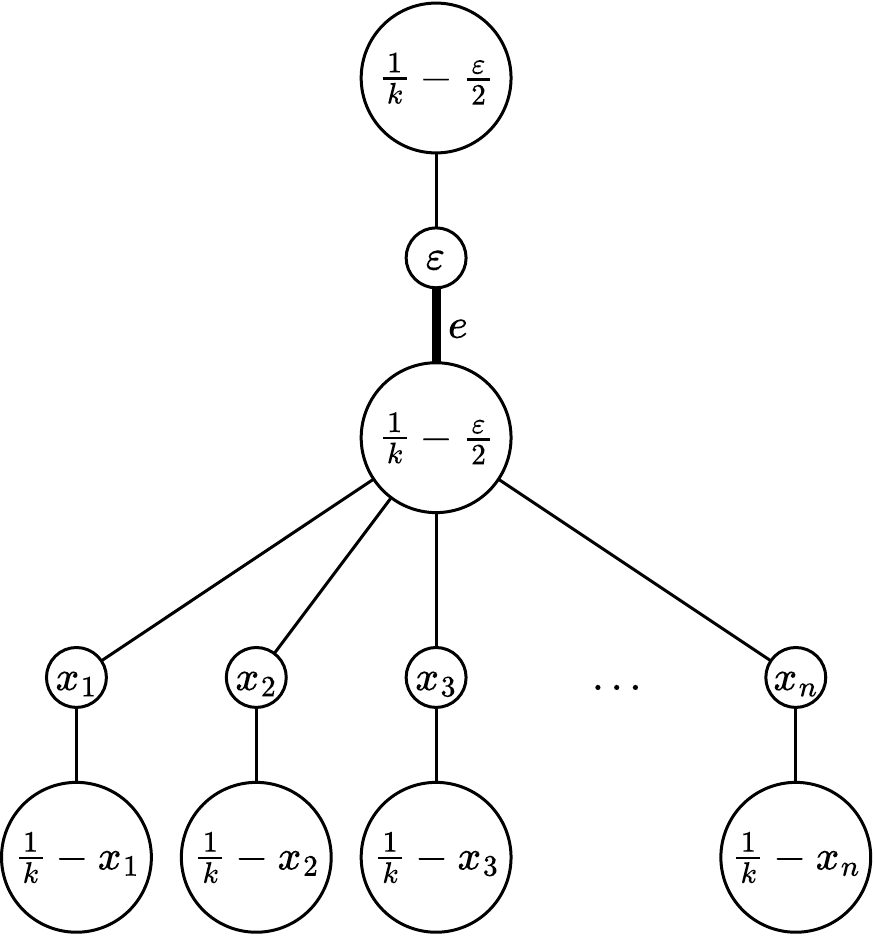}{\label{figSharpPReduction}The reduction from the 0/1 Knapsack counting problem to the problem of computing the probability a given edge $e$ is split in a uniformly random TAPP solution.}

    We claim that, if the number of knapsack solutions is $s$, then the probability $e$ is split in a uniformly random $\varepsilon$-balanced $k$-partition of $T$ is precisely $\frac{s}{s + 1}$. To see this, first observe that there is only ever one way to \emph{not} split $e$. Since preserving $e$ makes that district have maximum size $\frac{1}{k} + \frac{\varepsilon}{2}$, all of the $x_i$ nodes have to be matched with the larger parts below, with the top node in its own part. There are no other ways to get an $\varepsilon$-balanced $k$-partition. Now consider the partitions that split $e$. This entails the top two vertices getting merged together into one complete part. Each of the $x_i$ nodes can independently be matched above or below, provided that the weight is not too large above. Specifically, a set $\{x_i \suchthat i \in S\}$ for $S \subseteq [n]$ can be matched above if and only if
    \[\sum_{i \in S} x_i \leq \varepsilon \iff \sum_{i \in S} w_i \leq C.\]
    Thus, the number of valid partitions where $e$ is split is precisely the number of Knapsack solutions $s$. Thus, $e$ is split with probability $\frac{s}{s + 1}$.
\end{proof}

\section{BUD as a practical algorithm}\label{sec:practical}
\label{sec:bud_as_a_practical_algorithm}
BUD can be run as an independent algorithm with an invariant measure that is uniform on $\eps$-balanced $k$-splittable trees. When we have exact balance (i.e., $\eps = 0$), each such tree has a unique set of $k-1$ \emph{marked edges} that can be removed to induce a balanced forest. Although the measure is uniform on splittable trees, a given partition (without the connective information) may be turned into such a tree by constructing a spanning tree on the graph in which we quotient out the partitions, $G/\mP$. Thus, the probability mass associated with each balanced forest will be proportional to $\tau(G/\mP)$, where we  use $\tau(H)$ to denote the number of spanning trees on a (multi-)graph $H$.

As a first step, we implement and test the tree-level BUD algorithm on a $4\times 4$ square lattice with $k=4$ partitions and confirm that (i) all possible partitions are visited and (ii) the cut-edge edge distribution converges to the exact answer. We derive the exact answer by enumerating all possible partitions with exact balance, counting the number of associated spanning forests that each has, and then computing the number of trees on the quotiented district multigraph so that we uniformly sample splittable trees (see \eqref{eq:splittreemeasure} below; we test on the case where $\lambda \propto 1$). \footnote{See the test cases in BUD\_tree.jl in our code repository to reproduce this result.\url{https://github.com/gjh6/BUD.jl.git}.}

However, in practical redistricting applications, one often wishes to sample from a modified measure on partitions rather than trees. To do so, we use ideas from the Marked-Edge Walk (MEW \cite{mcwhorter2025MEW}) to lift the state space to marked trees. In particular, after proposing an $\eps$-balanced $k$-splittable tree, we mark a collection of $k-1$ edges in the tree that, if removed, would induce an $\eps$-balanced forest. This modification to BUD operates on splittable trees with marked edges, which uniquely describe a balanced forest (removing the marked edges), a balanced partition (by using the trees to define the partitions but then erasing the specific trees), and also a splittable tree (by unmarking the edges).

Thus we think of three sequentially lifted state spaces: (i) the partition $\mP$ on which we will define a measure, (ii) a forest $F$ with trees corresponding to the partitions (iii) and the marked balanced tree $T_M$ with marked edge set $M$. By assigning equal weight to all forests associated with a partition, we can induce measure on a forest $F$ that fixes a partition $\mP(F)$; by assigning equal weight to all possible marked edge combinations, we can induce a measure on a corresponding marked tree $T_M$ that induces a forest $F(T_M)$ and partition $\mP(F(T_M))$. 

Thus the probability measure on partitions, $\pi(\mP)$, defines a measure on forests as 
\begin{align}
\nu(F) = \pi(\mP(F))/\tau(F), \quad \text{ where }
\tau(F) = \prod_{T_i \in F} \tau(\mP(T_i)),
\label{eq:forestmeasure}
\end{align}
where $T_i$ indexes over the trees in $F$, $\mP(T_i)$ is the subgraph on the partition induced by the tree $T_i$, and $\tau(\mP(T_i))$ is the number of spanning trees in the subgraph.  Further lifting the measure to the splittable tree, the resulting measure becomes
\begin{align}
\label{eq:splittreemeasure}
\lambda(T_M) = \frac{\pi(\mP(T_M))}{\tau(F) \tau(G/\mP)},
\end{align}
where we have used $\mP(T_M) = \mP(F(T_M))$, and $Q(\mP) \equiv G/\mP$ represents the %
multigraph in which the nodes are the partitions and the edges between nodes are any edges in the original graph that span partitions. As a note, we can set $\pi(\mP(T_M))=\tau(F) \tau(G/\mP)$, which results in sampling uniformly over splittable trees. This is the measure that BUD is naturally in detailed balance with.\footnote{There is a simple extension of these ideas to weighted trees that we do not discuss in the current work.}

\subsection{Proposing marked edges}
\label{sub:proposing_marked_edges}
For exactly-balanced splittable trees ($\eps = 0$), we can use the BUD procedure as-is, since each such tree has a unique set of marked edges. In the approximate-balance regime, however, after proposing a single splittable tree there may be several ways of marking (or splitting) the tree into a $k$-partition. 
The measure over partitions is achieved by summing over all possible marked trees that correspond to the partition space. 

In developing a Markov Chain to sample from this space, BUD provides a natural mechanism for proposing a new splittable tree, but it does not come with a proposal mechanism to mark the edges. If we wish to target the measure $\pi$ on partitions using an algorithm like Metropolis-Hastings, we need a way of proposing marked edges that allows for computing the probability of selecting each set of marked edges. As we've seen above, even if we could uniformly select from all possible marked edge sets, we wouldn't be able to recover the probability due to the challenge of computing the size of the sample space.
Furthermore, if we would like to maintain any state-space connections made by BUD (e.g. not eliminate any valid transitions), the marked-edge proposal algorithm must also have a non-zero probability of returning any possible set of marked edges for a given splittable tree.

We achieve this via the algorithm $\Alg$. The idea is to use a fixed ordering $\sigma$ of $V(T)$ to select a leaf $v$ of $T$, randomly choose an edge $e_1$ to cut that breaks $v$ into its own district, and then recurse on the remaining $(k-1)$-splittable tree. The output is the tree $T_M$ with marked edge set $M = \{e_1, \dots, e_{k-1}\}$. The use of a fixed ordering ensures that --- conditioned on outputting $M$ --- our approach will deterministically select $e_1, e_2, \dots, e_{k-1}$ in that order, which allows us to tractably compute the probability of selecting $M$, given the ordering $\sigma$. However, the approach described above may deterministically avoid certain marked trees. To get around this, $\Alg$ will (with some tunable probability $p$) contract the branch containing $v$ instead of splitting it off, before recursing.

We proceed to formally define $\Alg$ in Definition~\ref{def:select-marked-tree} and prove in Lemma~\ref{lem:alg-has-properties} that it has the properties we need. We note that our results require a modest bound on the population tolerance $\eps$ around the ideal district size. In the following, we use the parameter $\ell$ to count the number of districts we aim to create, saving $k$ to control the ideal district size. Our results hold (and should be interpreted) for $\ell = k$.

Throughout the following, we use $\deg(T) := \max_{v \in T} \deg(v)$ to denote the maximum degree of tree $T$. %

\begin{definition}%
    Fix $\eps > 0$, an integer $k \geq 2$, and any tree $T$ with vertex weights $w(v) \geq 0$, normalized so that the sum of the vertex weights is $1$. For any subtree $T' \subseteq T$, the weight of $T'$ is $w(T') = \sum_{v \in V(T')} w(v)$. We say that $T'$ is a \emph{district} if its weight is within $\pm \eps/2$ of the ideal district weight $1/k$, and say that $T'$ is \emph{$\ell$-splittable} if it contains a set of $\ell - 1$ edges that, when cut, induce a forest where each tree is a district. We call such a set of edges a \emph{marked edge set} of $T'$.
\end{definition}

\begin{definition}
    For $\eps > 0$ and $k \geq 2$, let $T$ be an $\ell$-splittable tree where every leaf is a district. For any tree $T$ with  $\deg(T) \geq 3$ and any leaf $v \in V(T)$, the \emph{branch} $\branch_v$ is the path from $v$ to the closest vertex $u \in V(T)$ with $\deg(u) \geq 3$. When $\deg(T) < 3$, we define $\branch_v = T$. An edge $e$ is a \emph{viable edge of $v$} if deleting $e$ would induce a forest with subtrees $T_v$ and $T'$, where $T_v$ is a district containing $v$ and $T'$ is $(k-1)$-splittable. We let $\viable_v$ denote the set of all viable edges of $v$.
\end{definition}

\begin{definition}[$\Alg$]\label{def:select-marked-tree}
    Fix $\eps > 0$ and $k \geq 2$, let $T$ be an $\ell$-splittable tree. For any ordering $\sigma$ of $V(T)$ and parameter $p \in (0,1)$, the algorithm $\Alg(T, \ell, \sigma, p)$ proceeds as follows.
    \begin{enumerate}
        \item Let $M = \emptyset$.
        \item If $\ell = 1$, return $M$.
        \item Contract the leaves of $T$ until each leaf is a district. (Note that a leaf that is too light to be a district will be collapsed into its neighboring node and either no longer be a leaf or become a heavier leaf.) \label{step:contract-leaves}
        \item Let $v \in V(T)$ be the first leaf in the ordering $\sigma$ (ignoring contracted vertices).
        \item If $\deg(T) \geq 3$ \emph{and} contracting $\branch_v$ still leads to a balanced tree, contract $\branch_v$ with probability $1$ if $\viable_v=\emptyset$, and probability $p$ otherwise. Let $T'$ be the new tree and recurse on $\Alg(T', \ell, \sigma', p)$ by returning to Step (2). (Note $\sigma'$ is the ordering $\sigma$ after deleting vertices not in $T'$.) \label{step:contract-branch}
        \item Otherwise, select an edge $e$ from $\viable_v$ uniformly at random, add it to $M$, and let $T'$ be the tree induced by removing $e$ such that $v \notin V(T')$. Recurse on $\Alg(T', \ell - 1, \sigma', p)$ by returning to Step (2). \label{step:sample-viable}
    \end{enumerate}
\end{definition}

\begin{lemma}[Correctness]
    For any $\eps > 0$, $k \geq 2$, let $T$ be an $\ell$-splittable tree, $\sigma$ be an ordering of $V(T)$, and $p \in (0, 1)$. Then $M \gets \Alg(T, \ell, \sigma, p)$ is a set of marked edges for $T$.
\end{lemma}
\begin{proof}
    Observe that Step (\ref{step:contract-leaves}) is possible because otherwise $T$ would not be $\ell$-splittable. Step (\ref{step:sample-viable}) will certainly occur when $\deg(T) < 3$, so it must occur at some point, Step (\ref{step:contract-branch}) necessarily reduces the degree of some vertex $u \in V(T)$. When $\Alg$ calls Step (\ref{step:sample-viable}), the definition of $\viable_v$ ensures that marking $e$ will induce a forest containing a district $T_v$ and an $(\ell-1)$-splittable tree $T'$. The result follows by induction on $\ell$.
\end{proof}

\noindent To prove \Alg\ has the properties we desire, we need to show that the set of viable edges $\viable_v \subseteq \branch_v$ for any leaf $v$ of $T$. Lemma~\ref{lem:viable-all-in-branch} shows that it suffices to impose moderate bounds on the population tolerance.

\begin{lemma}\label{lem:viable-all-in-branch}
    For any $k \geq 2$ and $0 < \eps < 2/(3k)$, let $T$ be an $\ell$-splittable tree, where we have contracted the leaves until each leaf is a district. For any leaf $v \in V(T)$, we have $\viable_v \subseteq \branch_v$.
\end{lemma}
\begin{proof}
    Consider such a tree $T$ and any leaf $v \in V(T)$. Suppose that there exists an edge $e \in \viable_v \setminus B_v$. Since $e \notin B_v$, there must be a vertex $u \in V(T)$ with $\deg(u) \geq 3$; let $u$ be the closest such vertex to~$v$. Since $e \in \viable_v$, removing it will induce a forest with a district $T_v$ containing $v$. Let $L$ be the set of leaves of $T$ that are in $V(T_v)$, so $v \in L$. Crucially, $|L| \geq 2$, since $\deg(u) \geq 3$ and $e \notin B_v$ implies that $u \in V(T_v)$. But then we have
    \begin{align}
        2\left(\frac 1 k - \frac \eps 2\right) \leq |L|\left(\frac 1 k - \frac \eps 2\right) \leq w(L) \leq w(T_1) \leq \frac 1 k + \frac \eps 2 \label{chain-of-weights}
    \end{align}
    which implies that $\eps \geq 2/(3k)$, a contradiction. 
\end{proof}

\begin{lemma}\label{lem:alg-has-properties}
     For any $k \geq 2$ and $0 < \eps < 2/(3k)$, let $T$ be an $\ell$-splittable tree, $\sigma$ be an ordering of $V(T)$, and $p \in (0, 1)$. Then for any valid marked edge set $M$ of $T$, we have
    \begin{align}
        \Pr(\Alg(T, \ell, \sigma, p) = M) > 0. \label{line:positive-probability}
    \end{align}
    Moreover, conditioned on selecting $M$, the algorithm must select the edges of $M$ in a fixed order, i.e., for some ordering $e_1, \dots, e_{\ell-1}$ of the edges in $M$, we have
    \begin{align}
        \Pr(e_i \text{ is the $i^{\text{th}}$ edge selected}|\Alg(T, \ell, \sigma, p) = M) = 1. \label{line:fixed-order}
    \end{align}
\end{lemma}
\begin{proof}
    We proceed by induction on $\ell$. For the base case, we only need to prove Inequality~(\ref{line:positive-probability}). Let $T$ be a $2$-splittable tree and fix some marked edge set $M = \{e\}$. The ordering $\sigma$ deterministically picks out some leaf $v$ of $T$. Note we must have $e \in \viable_v$, so $\Alg$ will select $e$ with probability at least
    \begin{align}
        \frac{1-p}{|\viable_v|} > 0.
    \end{align}
    For the inductive step, consider an $\ell$-splittable tree $T$ and fix a marked edge set $M$. Let $e_1$ be the random variable denoting the first edge that $\Alg(T, \ell, \sigma, p)$ selects and let $v$ be the first leaf of $T$ picked out by $\sigma$. Let $S_v = M \cap \branch_v$. We consider two cases.
    \begin{itemize}
        \item \textbf{Case 1: $S_v \neq \emptyset$.} Let $e$ be the nearest edge to $v$ that is in $S_v$. Then we must have $e \in \viable_v$ and our argument above shows that $\Pr(e_1 = e) > 0$. Then the tree $T'$ induced by deleting $e$ is $(\ell-1)$-splittable and $M' = M \setminus \{e\}$ is a marked edge set of $T'$. By our inductive hypothesis, $\Pr(\Alg(T', \ell-1, \sigma', p) = M') > 0$ for any permutation $\sigma'$ over $V(T')$. Therefore the probability that we ultimately return $M$ is at least
        \begin{align}
            \Pr(e_1 = e)\Pr(\Alg(T', \ell-1, \sigma', p) = M'|e_1 = e) > 0.
        \end{align}
        Moreover, we claim that conditioning on $\Alg(T, \ell, \sigma, p) = M$ implies that the algorithm must select $e_1 = e$ with probability $1$. To see why, first note that the algorithm cannot contract $\branch_v$, as that would forbid $e \in E(\branch_v)$ from being included in the output. Therefore the algorithm will sample $e_1 \sim \viable_v$. But our choice of $e$ implies that, if $e_1 \neq e$, $e$ would never be included in the output, since the district $T_v$ induced by deleting $e_1$ will contain $e$, and we do not recurse on $T_v$.
        \item \textbf{Case 2: $M \cap E(\branch_v) = \emptyset$.} This case implies that $\deg(T) \geq 3$ and contracting $\branch_v$ still leads to a balanced tree. Therefore with probability at least $p > 0$, $\Alg$ will contract $\branch_v$ and recurse. Since contracting $\branch_v$ reduces the degree of some $u \in V(T)$ by one, after finitely many steps the recursive call will fall into Case 1, where we have already shown we will return $M$ with positive probability.
        
        Moreover, note that conditioning on $\Alg(T, \ell, \sigma, p) = M$ implies that the algorithm must contract $\branch_v$ with probability $1$, as otherwise it will sample $e_1 \sim \viable_v \subseteq \branch_v$ (by Lemma~\ref{lem:viable-all-in-branch}) and include $e_1 \notin M$ in the solution.
    \end{itemize}
    To prove Equation~(\ref{line:fixed-order}), note that these cases together imply that the first edge $e_1$ selected by the algorithm is fixed by $M$ and the ordering $\sigma$; in particular, for the first leaf node $v$ such that $S_v \neq \emptyset$, $e_1$ is the nearest edge to $v$ that is in $S_v$. The result follows by appealing to our inductive hypothesis.
\end{proof}

\subsection{Restricting BUD for computational efficiency}
\label{sub:restricting_bud_for_computational_efficiency}
To achieve the maximize connections in the state space of the marked-BUD walk, one would need to use $\Alg$ to sample a marked edge set at every step of the chain. This is because a BUD step --- even one that involves a small, local cycle --- can result in new possible ways of marking edges. However, running $\Alg$ on the entire tree at each step introduces a significant computational cost. In this section we discuss an approach to restrict the BUD walk in exchange for more efficient sampling. In essence, our idea is to fix a subset of marked edges and only re-sample marked edges near the cycle.

As discussed above in Section~\ref{sec:generalizing}, there have been previous efforts to reduce changes to the marked edges to improve computational efficiency. The Marked Edge Walk (MEW; \cite{mcwhorter2025MEW}) is equivalent to (i) randomly fixing all but one marked edge and (ii) restricting the potential positions of re-marking the edge so that it shares at least one node with its previous choice. In general, there may be no such option for such a nearby marked edge that keeps the tree in balance. In this case, the algorithm would add self-loops where the marked tree does not change (i.e., rejection steps).

The Cycle Walk method \cite{cyclewalk} also (effectively) adds a restriction to the allowable marked edges by insisting that if we have a cycle involving $\ell$ districts, then (i) all $(k-1) - (\ell-1) = k-\ell$ marked edges \emph{not} involved in the cycle are fixed and (ii) the $\ell-1$ marked edges that begin on the cycle must still be on the cycle for a newly proposed tree.

We strike a balance between the approaches of MEW and Cycle Walk. Fix a distance parameter $d \in \N$. Consider a BUD step from a splittable tree $T$ with marked edge set $M$. Start a BUD step, i.e., add a random edge $e$ and let $C$ be the set of vertices involved in the cycle that was created.
The restricted BUD move proceeds as follows.
\begin{enumerate}
    \item Let $N$ denote the subgraph induced by $C$ and all nodes at most $d$ distance away from $C$, in the path metric on $T$. Let $\ell'$ denote the number of marked edges in $N$, i.e., $\ell' = |M \cap E(N)|$.
    \item Remove the $k-1-\ell'$ marked edges that are $d$-far from $C$ and consider the remaining subgraph $\tilde{N}$ that contains $N$. Note that $\tilde{N}$ is comprised of $\ell'+1$ districts.
    \item For each node $v$ in $\tilde{N}$ that is more than $d$-far from $C$, contract $v$ to the unique vertex $w$ on the path from $v$ to $C$ that is exactly $d$-far from $C$. Call the resulting subgraph $H$.
    \item Determine which edges in $C$ turn $H$ into an $(\ell'+1)$-splittable tree when removed, and select one of these edges to remove uniformly at random, creating a tree $H'$.
    \item Use $\Alg$ to randomly select a set of $\ell'$ marked edges in $\tilde{H'}$. The final set of marked edges are these $\ell'$ edges, along with the unchanged $k - 1 - \ell'$ edges that are $d$-far from $C$.
\end{enumerate}

Importantly, this sequence will produce a \emph{proposed} new state. We will then add rejection via the Metropolis-Hastings algorithm to ensure we sample from the fixed invariant measure $\lambda$.  The probability of making a proposal is tractable and may be written as follows, where $Q(T'_M | T_M)$ is the probability, if at marked tree $T_M$, of proposing tree $T_M'$. 
\begin{align}
Q(T_M' | T_M) &= P(e \notin T_M) P(e' \in C(e, T_M)|d) P(M' | M, d, e \cup T_M, \sigma)\\
              &=  \frac{1}{|V|-1} \frac{1}{|\mathcal{R}(T_M, e |d)|} P(M' | \Alg, M, d, e \cup T_M, \sigma).
\label{eq:markedBUDproposal}
\end{align}
In the above equation, the first term is the probability of selecting an edge $e$ to add that is not in the tree $T_M$; the second term is the probability of selecting an edge $e'$ that can be removed from the cycle formed by $e$ and $T_M$, $C(e, T_M |d)$, and results in a cuttable tree on the $d$-distance contracted subgraph described above; and the last term is the probability of selecting a new set of marked edges. The first term can be expanded by noting that there are $|V|-1$ edges in $T_M$ (and we pick uniformly), the second by letting $\mathcal{R}(T_M, e |d)$ denote the set of removable edges from $C(e, T_M)$ that leave a balanced tree in the contracted subgraph described above, and the last term from the recursive (and tractable) process described in the previous section via \Alg.

The above equation holds assuming that (i) the underlying trees of $T_M$ and $T_M'$ differ by exactly one edge, (ii) given the cycle determined by filling in the difference of the trees, all marked edges more than a distance of $d$ away from the cycle are the same. If these conditions do not hold, the transition probability is 0 if $T_M \neq T_M'$ and is equal to $1-\sum_{\bar{T}_M \neq T_M} Q(\bar{T}_M | T_M)$ if $T_M = T_M'$. In the case where \eqref{eq:markedBUDproposal} holds, the first two terms cancel in the Metropolis-Hastings acceptance ratio, and we are left with the acceptance probability
\begin{align}
A(T_M' | T_M) = \min \left( \frac{ P(M | M', d, e' \cup T_M', \sigma)}{ P(M' | M, d, e \cup T_M, \sigma)} \frac{\lambda(T_M')}{\lambda(T_M)} \right),
\end{align}
in which the second term vanishes when $\lambda \propto 1$ and $\pi \propto \tau(F) \tau(G\backslash \mP)$.

Note that when $d=0$, we simply rearrange marked edges on the cycle, and thus we recover a generalization of Cycle Walk as mentioned above in Section~\ref{sec:generalizing}. We do not explicitly seek to generalize all MEW moves; however, if we added a proposal step which fixed the tree and simply reassigned all (or a subset) of marked edges with \Alg, the algorithm would recover all possible MEW moves with the added benefit of always staying in population balance.

In the experiments below, we limit our study to the case where $d=0$.

\subsection{Numerical results}\label{sub:numerical}
We begin by validating the marked-edge version of BUD again on a $4\times 4$ grid graph with exact balance, under the measure $\lambda \propto 1$. The test can found in the test cases of the repository, and examines the distribution of the cut-edges in the graph which is validated against the exact answer.

\subsubsection{A test on a previously stuck case} We continue with a simple and small graph example of a $4\times 4$ grid graph graph with 5 districts and a balance condition between 3 and 4 nodes. For an observable, we examine the rank-ordered marginals of the isoperimetric ratios. In this case we do not have an exact enumeration of the state-space so we launch six independent runs, each with a distinct initial splittable tree. We then examine the averaged total variation accross the six emperical distributions for each pair of runs and report the largest value.   

We chose this configuration because it was the one case for which the Cycle Walk was shown to \emph{not} mix. However, BUD seems to have no such issues, as we report in Figure~\ref{fig:cyclevsbud4x4d5}. Above we proved that BUD is irreducible for triominos, this result is similar though slightly expansive since the lattice dimensions and the balance conditions are beyond the scope of that proof.

\begin{figure}
\centering
\includegraphics[width = 0.45\textwidth]{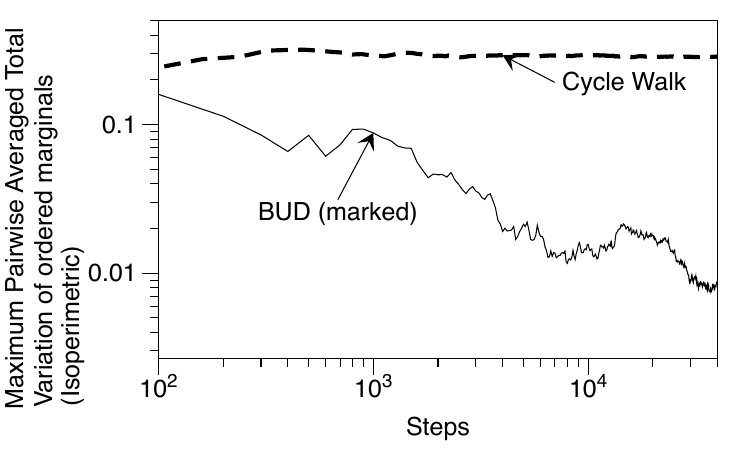}
\caption{We run BUD and the Cycle Walk on a $4\times4$ square lattice with 5 districts and each district comprised of 3 or 4 nodes. The worst case pairwise error across 6 independent runs decays for BUD but does not for Cycle Walk.}
\label{fig:cyclevsbud4x4d5}
\end{figure}

\subsubsection{Comparing BUD with the cycle walk}
Next we compare the relative sampling efficiency between BUD and Cycle Walk. We test this on an $8\times8$ square lattice with 5 districts and a balance condition of 12 to 13 nodes per tree in the corresponding balanced forest. We run BUD for 200,000 steps, gathering the number of cut edges at each step. We confirm that (i) the empirical distribution at 50,000 steps has about a 3\% deviation compared with that taken with 200,000 samples and that (ii) the deviation after 50,000 from an independent chain also has an error of roughly 3\% in total variation between the two empirical distributions. This indicates the number of steps we've considers likely suffices to address questions about limiting distributions. 

In the 200,000 steps, BUD spends roughly 55\% of them proposing edges that are entirely contained within a district, meaning $\mP$ is fixed in these steps. BUD spends roughly 32\% of the steps making cycles with two districts and the remainder with 3 or more. 

In comparing BUD with Cycle Walk, there are a few challenges. First, Cycle Walk has a tunable parameter for taking 1-tree vs 2-tree Cycle Walks and it is not entirely clear how to set a fair comparison. We opt to set the 1-tree frequency to that of BUD, i.e. 55\% of the steps will be 1-tree Cycle Walk steps. Second, it is not clear which measure we should choose. BUD is best suited to sample from $\lambda \propto 1$; Cycle-Walk is best suited to sample on $\nu \propto 1$, where $\nu$ is the measure on balanced forests defined in \eqref{eq:forestmeasure}.\footnote{This statement is based on what measure reduces variation in the Metropolis-Hastings acceptance ratio.} Comparing the methods  on one of these measures would create a biased conclusion. Furthermore, both methods should benefit from their ability to make small, medium, and large changes to the underlying forest to sample from a range of different measures. For current purposes, we compare the effective sample rates and auto-correlations for the natural measures for each method (i.e. each method samples from a separate method).

We display the autocorrelation plots in Figure~\ref{fig:cyclevsbud8x8d5autocor}. We see that, per non-internal step, the BUD walk decorrelates significantly faster than does the Cycle Walk. In terms of the effective sample rate, BUD's effective sample rate is roughly twice that of Cycle Walks.

A caveat of BUD's accelerated sampling is the wall clock time: There is computational cost in running the \emph{TAPP} procedure and in \Alg. To date, the Cycle Walk has a more highly optimized cod ebase; in their current iterations, BUD runs roughly 9 times slower per effective sample than Cycle Walk in this test case. We leave it open whether code optimization can improve the performance of BUD and whether in other types of graphs/sampling conditions BUD will be more performant than Cycle Walk.

\begin{figure}
\centering
\includegraphics[width = 0.45\textwidth]{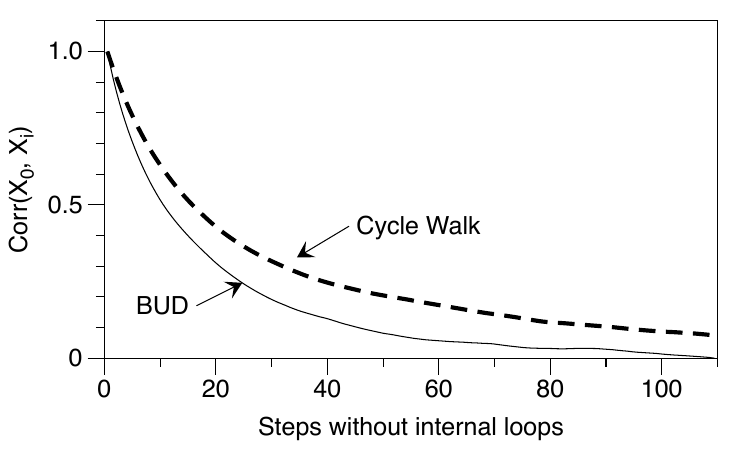}
\caption{We run BUD and Cycle-Walk on a $4\times4$ square lattice with 5 districts and each district comprised of 3 or 4 nodes. The worst case pairwise error across 6 independent runs decays for BUD but does not for Cycle Walk.}
\label{fig:cyclevsbud8x8d5autocor}
\end{figure}

\subsubsection{Sampling a large graph and understanding the uniform measure on splittable trees}
We next study the BUD walk on all of North Carolina. We run for 1 million steps on the 14 districts using the 2020 census population data and output every 100 steps, using a balance condition of 2\% deviation from exact balance. We take 4 runs of  BUD and 4 of Linked Forest Recombination \cite{autry2021metropolized}. On the former chain, we consider $\lambda \propto 1$ and on the latter, $\nu \propto 1$ (again, to provide favorable mixing conditions for each chain). 

To test for convergence we examine the Democratic vote share that each district would have received under the votes 2020 Presidential Election. We then order the the districts from most to least Republican and examine the marginal distributions of the order statistics. We examine the difference in the average total variation across the 14 marginals as a function of the number of samples and see good agreement between the 4 independent chains. We compute total variation by creating histograms with a bin width of $0.2$\% and plot our results in Figure~\ref{fig:budConvergence}.

Despite the difference in measure, the 14 observables appear to be mostly well aligned between the BUD walk and Linked Forest Recombination as shown in Figure~\ref{fig:budAndRecom}. In addition, we look at differences in cut-edges which has been proposed and used as a discretized type of compactness \cite{deford2019recombination}; we find placing importance on the district quotient graphs with many trees has the overall effect of increasing the number of cut edges. 

\begin{figure}
\centering
\includegraphics[width = 0.45\textwidth]{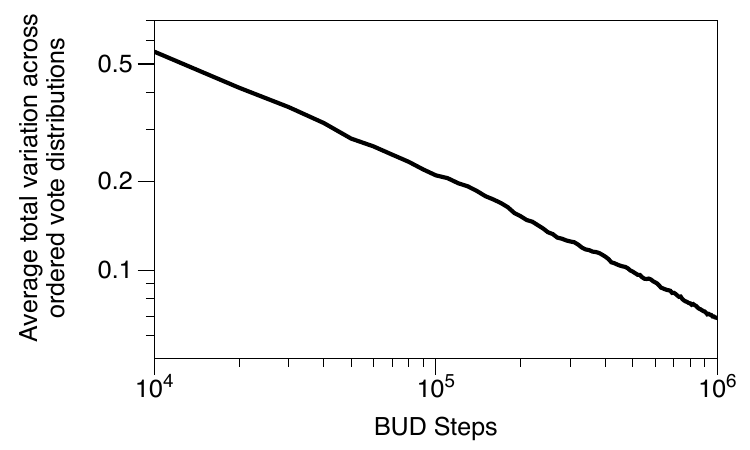}
\caption{We investigate the convergence of BUD on the North Carolina precinct graph and find evidence for convergence when examining distributions of ranked ordered vote .}
\label{fig:budConvergence}
\end{figure}

\begin{figure}
\centering
\includegraphics[width = 0.45\textwidth]{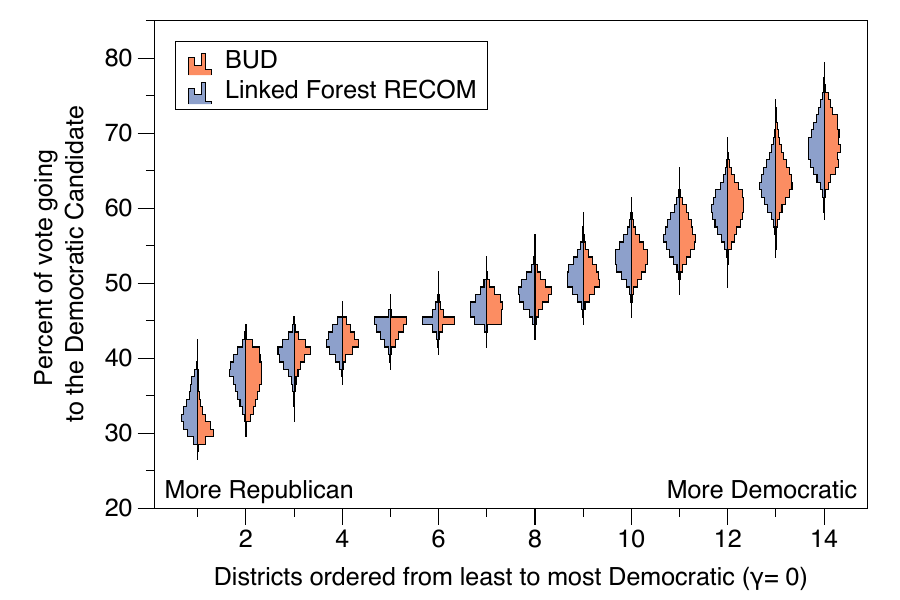}
\includegraphics[width = 0.45\textwidth]{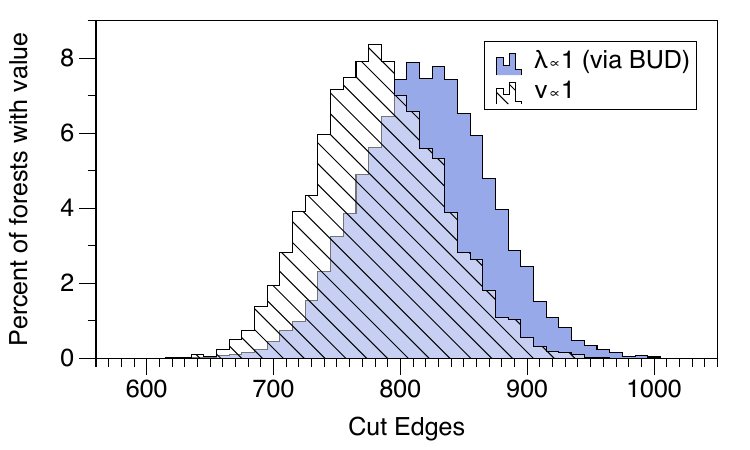}
\caption{We compare the randed-ordered vote marginals under the 2020 Presidential Election $\nu \propto 1$ to BUD with $\lambda \propto 1$ (both measures suited for each method). We also examine the difference in emperical distributions of cut edges between the measures and find that $\lambda \propto 1$ generally leads to more cut edges than the spanning forest distribution with $\nu \propto 1$.}
\label{fig:budAndRecom}
\end{figure}

\subsubsection{Does BUD have optimal scheduling?}
Linked Forest Recombination seeks to change $\mP$ at every step, whereas BUD will only change $\mP$ if an edge that spans to districts is chosen. Cycle Walk comes with a tunable parameters in which one can explicitly set  We emperically examine the ensemble with $\lambda$ uniform and find that only 16.8\% of the proposals do not sample internal edges, and therefore have any change at changing $\mP$. We also examine the ensemble with $\nu$ uniform and find that we sample an edge that spans districts in 15.9\% of the proposals (the smaller value due to the fewer cut edges as shown above).

With the initial investigation of Cycle Walk, the authors found evidence of a transition point when manually setting the ratio between internal and external moves. They found this by examining total variation error of some observables after a fixed number of external steps for a variety of external-to-internal step ratios. We examine the ratio that naturally arises from BUD moves and find that it neatly fits into this transition point: Having more internal moves would increase the work without enhancing mixing; having fewer internal moves may be detrimental to mixing. We plot our result in Figure~\ref{fig:optimalSchedule} by contextualizing BUD in a figure from \cite{cyclewalk}. We see that the BUD ratio seems to be large enough to not waste work with internal moves, but, as evidenced in the previous section, to still have enough internal moves to allow mixing. Further resolution of the referenced study would be needed to observe any tradeoffs between mixing and added computational effort.

It is an open question as to whether the natural ratio will hold as the number of cut edges continues to change according to other measures, or if one would be better served manually selecting a ratio for particular measures. 

\begin{figure}
\centering
\includegraphics[width = 0.45\textwidth]{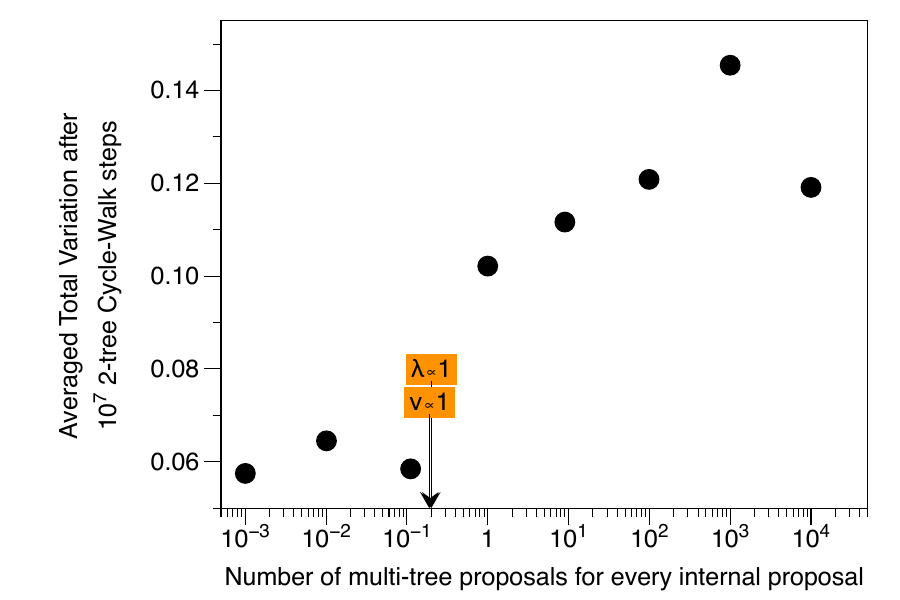}
\caption{We take the results from \cite{cyclewalk} that investigate optimal internal to external move ratios and find that BUD sits near an optimal balance between adding too much work via more internal moves and adding enough internal moves to still allow for efficient mixing.}
\label{fig:optimalSchedule}
\end{figure}

\subsubsection{Balanced up-down versus up-down walks}
We conclude by emperically comparing BUD with the classic Up-Down walk that samples all possible trees rather than restricting to splittable trees. We begin by noting that the $4\times 4$ square lattice with partitioned into 4 regions with exact balance has 35,624 splittable trees and 100352 total trees (ignoring symmetries). Thus, although BUD restricts the allowed moves in the up-down walk, it samples from a significantly smaller space and it is unclear what, if any, impact this would have when contrasting the mixing times of BUD with the classic up-down walk.

We investigate a difference in mixing times by again turing to the $8\times 8$ square lattice, this time with 4 districts and exact balance (so 16 nodes in each district). We again run for 200,000 steps. We cannot use cut edges as a proxy observable because cut edges is not well define for the classic up-down walk as there is no way of partitioning the non-splittable tree. We instead opt to examine the tree diameter as an observable.

We examine the auto-correlation plots for both BUD and the Up-Down walks in Figure~\ref{fig:budVud} along with looking at the resulting emperical histograms of tree diameters. We find that despite differences in the emperical observable distributions, the auto-correlation, and hence effective sample rates are nearly equivalent. 

\begin{figure}
\centering
\includegraphics[width = 0.45\textwidth]{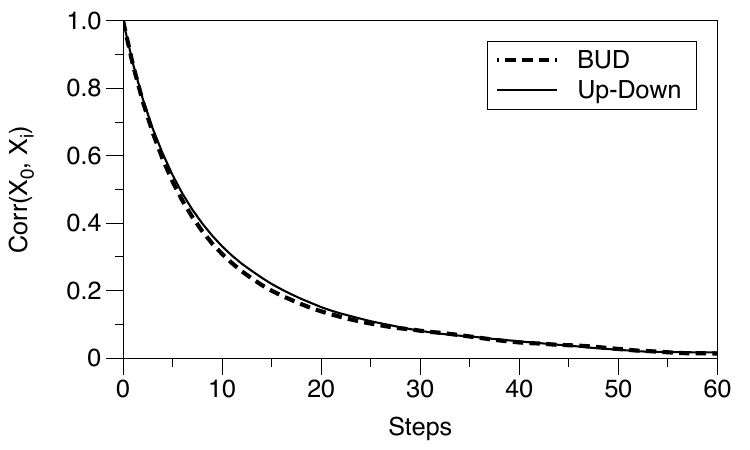}
\includegraphics[width = 0.45\textwidth]{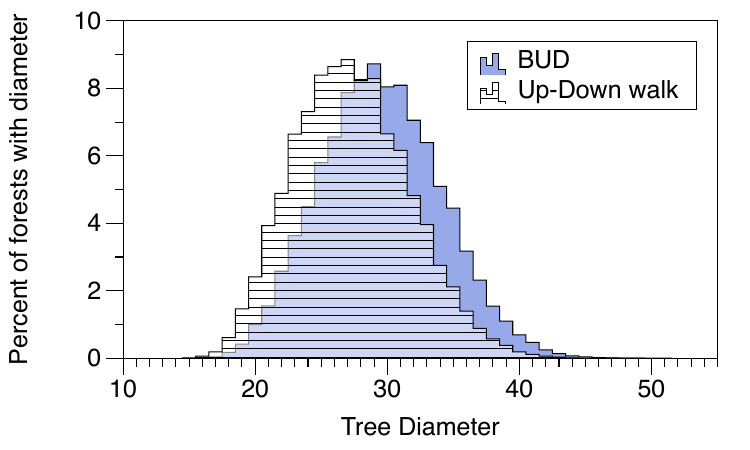}
\caption{We compare the auto-correlations of the BUD and Up-Down walks using the tree diameter as the observable. We also display the empirical distributions of the tree diameter.}
\label{fig:budVud}
\end{figure}

\section*{Acknowledgments}

The authors began this project at the American Institute of Mathematics during a workshop on ``Mathematical foundations of sampling connected balanced graph partitions" in June 2025.  They thank AIM for the funding support that made organizing and attending this workshop possible. The authors also thank Jonathan Mattingly for his early conversations that lead to the formation of the BUD walk and the initial discussions about its irreducibility; Daryl DeFord, who proposed the generic framework of gradually modifying existing forests which spurred much of our investigations; and Eyob Tsegaye, who was involved in creating some of our early examples where BUD was not irreducible.

\bibliographystyle{plain}
\bibliography{ref.bib}

\end{document}